\documentclass[times, 11pt, onecolumn]{article}
\usepackage{latex8}
\usepackage{times}

\usepackage{dsfont, microtype, amsmath, amsthm, amssymb, subcaption, url, subcaption}
\usepackage[linesnumbered,boxed]{algorithm2e}
\usepackage{tikz}
\usetikzlibrary{shapes.geometric, arrows}

\begin{document}

\title{A Geometric Structure Associated with the Convex Polygon}

\author{Kai Jin\\
The Hong Kong University of Science and Technology\\
Clear Water Bay, Hong Kong SAR \\ cscjjk@gmail.com
}

\maketitle
\thispagestyle{empty}

\newcommand{\Nest}{\mathrm{Nest}}
\newcommand{\T}{\mathcal{T}}
\newcommand{\C}{\mathcal{C}}
\newcommand{\block}{\mathsf{block}}
\newcommand{\sector}{\mathsf{sector}}

\theoremstyle{plain}
\newtheorem{remark}{Remark}
\newtheorem{note}{Note}

\theoremstyle{theorem}
\newtheorem{theorem}{Theorem}
\newtheorem{corollary}{Corollary}
\newtheorem{fact}{Fact}
\newtheorem{lemma}{Lemma}
\newtheorem{definition}{Definition}

\begin{abstract}
We propose a geometric structure induced by any given convex polygon $P$, called $\Nest(P)$, which is an arrangement of $\Theta(n^2)$ line segments, each of which is parallel to an edge of $P$, where $n$ denotes the number of edges of $P$. We then deduce six nontrivial properties of $\Nest(P)$ from the convexity of $P$ and the parallelism of the line segments in $\Nest(P)$. Among others, we show that $\Nest(P)$ is a subdivision of the exterior of $P$, and its inner boundary interleaves the boundary of $P$. They manifest that $\Nest(P)$ has a surprisingly good interaction with the boundary of $P$. Furthermore, we study some computational problems on $\Nest(P)$. In particular, we consider three kinds of location queries on $\Nest(P)$ and answer each of them in (amortized) $O(\log^2n)$ time. Our algorithm for answering these queries avoids an explicit construction of $\Nest(P)$, which would take $\Omega(n^2)$ time.

By applying the aforementioned six properties altogether, we find that the geometric optimization problem of finding the maximum area parallelogram(s) in $P$ can be reduced to answering $O(n)$ aforementioned location queries, and thus be solved in $O(n\log^2n)$ time. This application will be reported in a subsequent paper.
\end{abstract}

\begin{figure}[h]
\begin{minipage}[c]{0.5\textwidth}
\centering\includegraphics[width=.59\textwidth]{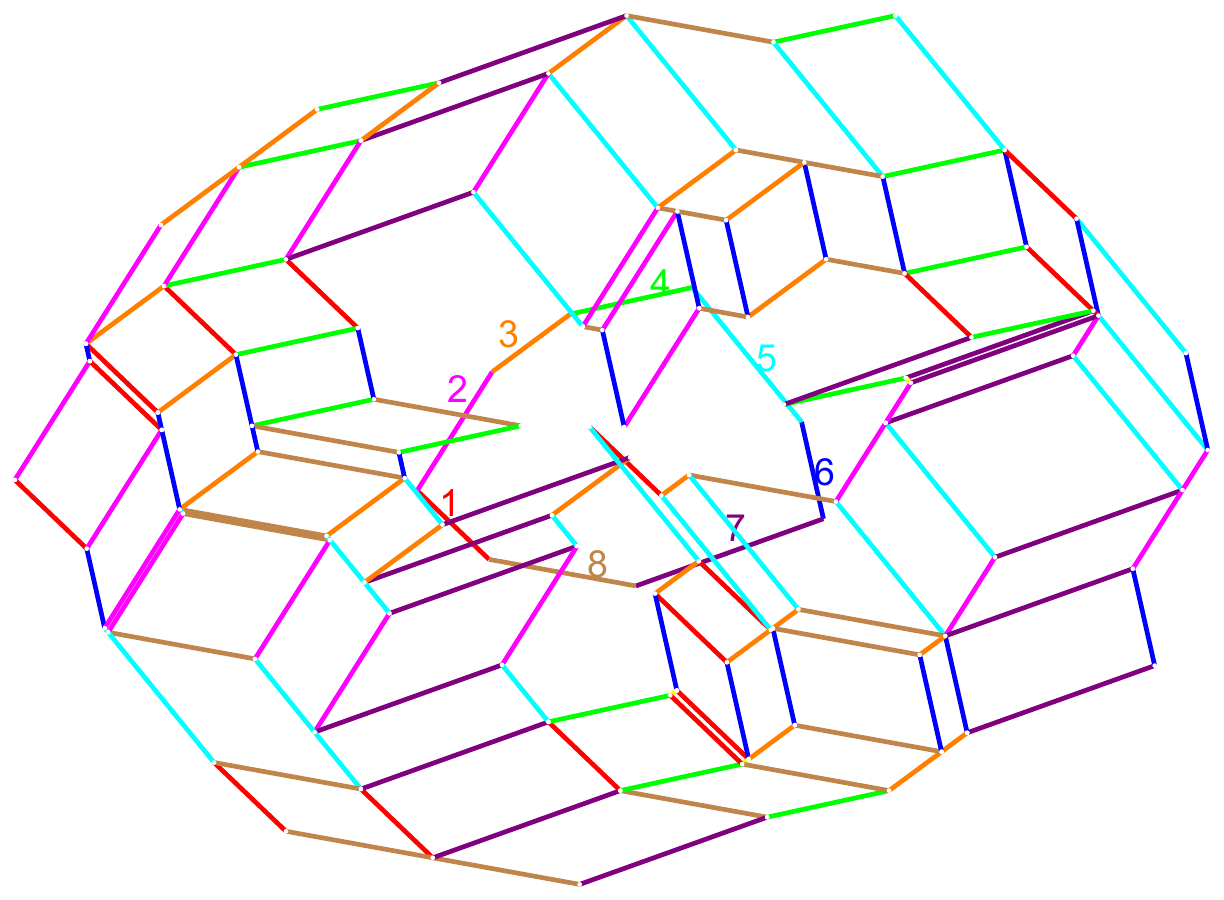}
\end{minipage}
\begin{minipage}[c]{0.5\textwidth}
\centering\includegraphics[width=.5\textwidth]{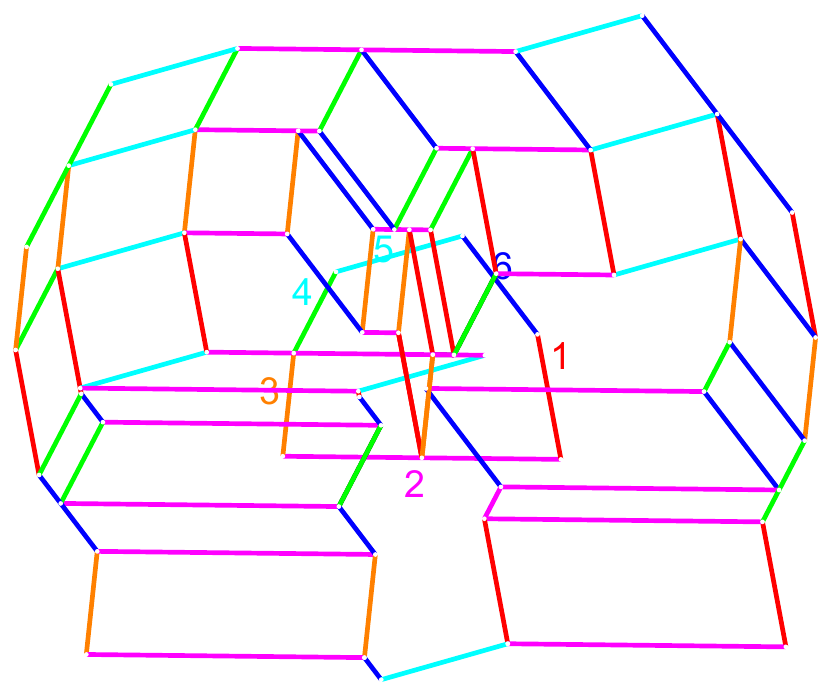}
\end{minipage}
\caption{Two examples of $\Nest(P)$. The line segments labeled from 1 to $n$ in clockwise order indicate the $n$ edges of $P$. The other line segments in the figure are the edges from $\Nest(P)$.}\label{fig:examples-only-NestP}
\end{figure}

\section{Introduction}\label{sect:introduction}

Geometric structures associated with every instance in a certain space have drawn a lot of attention in discrete and computational geometry.
Some prominent representatives are Delaunay triangulation~\cite{BergCG}, Voronoi diagram~\cite{BergCG} (with many variants),
      and Zonotope~\cite{LecturesPolytopes}, which have found enormous applications due to their rich structural properties.

In this manuscript we introduce a geometric structure that is induced by the convex polygons.
More specifically, for every convex polygon $P$, we define a structure called $\Nest(P)$, which is arrangement of $\Theta(n^2)$ line segments,
  each of which is parallel to an edge of $P$, where $n$ denotes the number of edges of $P$.
    See examples in Figure~\ref{fig:examples-only-NestP}.

The definition of $\Nest(P)$ is roughly as follows. Throughout, regard $P$ as a compact set that contains its boundary and interior.
Assume $(l,l')$ is taken over those pairs of nonparallel supporting lines of $P$ with the following property:
   the intersections $l\cap P$, $l\cap l'$, and $l'\cap P$ lie in clockwise order.
   Denote by $B_{l,l'}$ the (unique) point in $P$ that maximizes the product $d_l(B)\cdot d_{l'}(B)$, where $d_l(B)$ denotes the Euclidean distance from $B$ to $l$.
   Denote
     $F=\big\{ A+A'-B_{l,l'} \mid A\in l\cap P,A'\in l'\cap P, (l,l')\text{ is taken over the aforesaid pairs of supporting lines}\big\}.$
   We call each vertex or edge of $P$ a \emph{unit}, and regard the edges as open, i.e.\ they do not contain their endpoints. Thus each point in the boundary of $P$ lies in exactly one unit.
   Observe that we can define a subregion of $F$ by adding a constraint that $A,A'$ lie in a certain pair of units,
   and we can also define a subregion of $F$ by adding a constraint that $B_{l,l}$ is restrict to a certain unit.
   The union of the boundaries of all such subregions is defined to be $\Nest(P)$.

\smallskip \textbf{Note}: As shown below, to rigorously define $\Nest(P)$ we need several cascading notations, a few of which do not originate in this manuscript but in \cite{arxiv:n2}. However, only in this manuscript, we introduce and study $\Nest(P)$.

\smallskip This new structure turns out to have a number of surprising properties.
    In this manuscript, we prove six simple (yet nontrivial) properties of $\Nest(P)$,
       which all manifest that $\Nest(P)$ has good interactions with the boundary of $P$ (see Theorem~\ref{thm:nestp}),
    e.g.\ $\Nest(P)$ is a subdivision \cite{BergCG} of the exterior of $P$.
    These properties are essentially the properties of the convex polygon $P$, and hence may be of independent interest in convex geometry.
    Moreover, we study some computational aspects of $\Nest(P)$. We show that several location queries in $\Nest(P)$
      can be answered in $O(\log^2n)$ amortized time, without having a construction of $\Nest(P)$, which costs $\Omega(n^2)$ time (see Theorem~\ref{thm:nestp-location}).

\paragraph{Application.}
    In a subsequent paper (and also in my PhD. dissertation), we show that computing the parallelograms with the maximum area in convex polygon $P$ (as well as the parallelograms whose areas are locally maximum) reduces to answering $O(n)$ aforementioned location queries on $\Nest(P)$
      and hence can be solved in $O(n\log^2n)$ time by applying Theorem~\ref{thm:nestp-location}.
    This result is superior than the most related work
      including $O(n^3)$, $O(n^2\log n)$, $O(n^2)$, $O(n^2)$, and $O(n^2)$ time algorithms for
        finding the maximum rectangle \cite{Othershape-rect-EuroCG14}, maximum similar copy of a given triangle \cite{Placement-ST-CGTA94},
          maximum inscribed equilateral triangle and square \cite{Othershape-square-Allerton87}, and maximum parallelogram \cite{arxiv:n2} in a convex polygon.
    Interestingly, each property in Theorem~\ref{thm:nestp} has a particular value in proving this reduction.

    \medskip We state our results in this section, state important lemmas in section~\ref{sect:pre}, and outline our techniques in section~\ref{sect:techover}.

\newcommand{\unit}{\mathbf{u}}
\newcommand{\el}{\ell}        

\subsection{The so-called chasing relation between every two edges and every two units}\label{subsect:chasing}

\begin{description}
\item[Edges and vertices.]
Let $e_1,\ldots,e_n$ be a clockwise enumeration of the edges of $P$.
For simplicity of discussion, assume the edges are \textbf{pairwise-nonparallel}.
Denote the vertices of $P$ by $v_1,\ldots,v_n$ such that $e_i=(v_i,v_{i+1})$ (where $v_{n+1}=v_1$).
Denote the boundary of $P$ by $\partial P$.
As mentioned above, each point $X$ in $\partial P$ lies in a unique unit; denote the unit by $\unit(X)$.
    Unless otherwise stated, an edge, vertex, or unit refers to an edge, vertex, or unit of $P$, respectively.
Denote by $\el_i$ the extended line of $e_i$.

\item[Chasing relation between edges.]
    Edge $e_i$ is \emph{chasing} $e_j$, denoted by $e_i\prec e_j$, if the intersection of $\el_i$ and $\el_j$ lies between $e_i,e_j$ clockwise.
    Denote by $e_i\preceq e_j$ if $e_i=e_j$ or $e_i\prec e_j$.
    By the pairwise-nonparallel assumption of edges, exactly one in every pair of distinct edges is chasing the other.

    For example, in the left picture of Figure~\ref{fig:examples-only-NestP}, $e_2$ is chasing $e_3,\ldots,e_6$, whereas $e_7,e_8,e_1$ are chasing $e_2$.

\item[Backward and forward edges of every unit.]
The backward and forward edges of $v_i$ are $e_{i-1}$ and $e_i$, respectively.
The backward and forward edges of $e_i$ are $e_i$ itself.
 Intuitively, if you start at any point in a unit $u$ and move counter-clockwise (clockwise) along $\partial P$ by an infinitely small step, you will be located at the backward (forward) edge of $u$.
Denote the backward and forward edge of $u$ by $back(u)$ and $forw(u)$ respectively.

\item[Chasing relation between units.] We say unit $u$ is \emph{chasing} unit $u'$ if
\begin{equation}\label{def:unit-chasing}
back(u)\prec back(u')\hbox{ and }forw(u)\prec forw(u').
\end{equation}
\textbf{Note}: The relation ``chasing between units'' is a compatible extension of ``chasing between edges''.

\textbf{Note}: It may occur that neither $u$ is chasing $u'$, nor $u'$ is chasing $u$.
\end{description}

\subsection{Define $\zeta(u,u')$ for each unit pair $(u,u')$ in which $u$ is chasing $u'$}\label{subsect:zeta}

\newcommand{\dist}{\mathsf{disprod}}
The \emph{distance-product} from point $X$ to lines $(l,l')$, denoted by $\dist_{l,l'}(X)$, is defined as $d_{l}(X)\cdot d_{l'}(X)$.

\begin{description}

\item[$Z$-points \cite{arxiv:n2}.]
    Given two edges $e_i,e_j$ such that $e_i\prec e_j$, it is trivial and proved in \cite{arxiv:n2} that in domain $P$, function $\dist_{\el_i,\el_j}()$ achieves maximum value at a unique point; and this point must lie in $\partial P$.
Denote this maximum value point by $Z_{e_i}^{e_j}$ or $Z_i^j$ for short. We call the $\Theta(n^2)$ points in $\{Z_i^j\mid e_i\prec e_j\}$ the \emph{$Z$-points}.

\item[Boundary-portions of $P$.]
     Every continuous portion of $\partial P$ (including a single point of $\partial P$) is called a \emph{boundary-portion}.
        If we travel around $\partial P$ from one point $X$ to another point $X'$ clockwise, we pass through a boundary-portion;
          the endpoints-inclusive and endpoints-exclusive versions of this portion are denoted by $[X\circlearrowright X']$ and $(X\circlearrowright X')$ respectively. Points $X$ and $X'$ are referred to as their \emph{starting and terminal points}.

        \textbf{Note}: We assume that $[X\circlearrowright X']$ only contains the single point $X$ but not the entire $\partial P$ when $X=X'$.

        We regard every boundary-portion directed and the direction conforms with the clockwise order.

        For two points $A$ and $B$ in a boundary-portion $\rho$, we state that $A<_\rho B$ if $A$ would be encountered earlier than $B$ when we travel along $\rho$ (clockwise), and we state that $A\leq_\rho B$ if $A=B$ or $A<_\rho B$.
\end{description}

We recall a property of the $Z$-points given in \cite{arxiv:n2}. It is helpful for understanding (\ref{eqn:zeta_chasing}) -- the definition of $\zeta$.

\begin{fact}[\cite{arxiv:n2} Bi-monotonicity of the $Z$-points]\label{fact:Z_bi-monotonicity}
Given $e_s,e_t$ such that $e_s\preceq e_t$. Let \(S=\{(e_i,e_j)\mid e_i\prec e_j, \hbox{and $e_i,e_j$ both belong to $\{e_s,e_{s+1},\ldots,e_t\}$}\}.\)
We claim that all the $Z$-points in set $\{Z_i^j\mid (e_i,e_j)\in S\}$ lie in boundary-portion $\rho=[v_{t+1} \circlearrowright v_s]$, and more importantly, they obey the following bi-monotonicity:
\[\hbox{For $(e_{i},e_{j})\in S$ and $(e_{i'},e_{j'})\in S$, if $e_{i}\preceq e_{i'}$ and $e_{j}\preceq e_{j'}$, then $Z_{i}^{j}\leq_\rho Z_{i'}^{j'}$}.\]
\end{fact}

\begin{description}
\item[Boundary-portion $\zeta(u,u')$.]\cite{arxiv:n2} For every unit pair $(u,u')$ such that $u$ is chasing $u'$, define
    \begin{equation}\label{eqn:zeta_chasing}
        \zeta(u,u'):= [Z_{back(u)}^{back(u')} \circlearrowright Z_{forw(u)}^{forw(u')}]. \qquad \hbox{(See Figure~\ref{fig:def-zeta} for an illustration.)}
    \end{equation}
    Note that when $u$ is chasing $u'$, we have (\ref{def:unit-chasing}), and thus $Z_{back(u)}^{back(u')}$ and $Z_{forw(u)}^{forw(u')}$ are both defined.
\begin{figure}[h]
\centering\includegraphics[width=.63\textwidth]{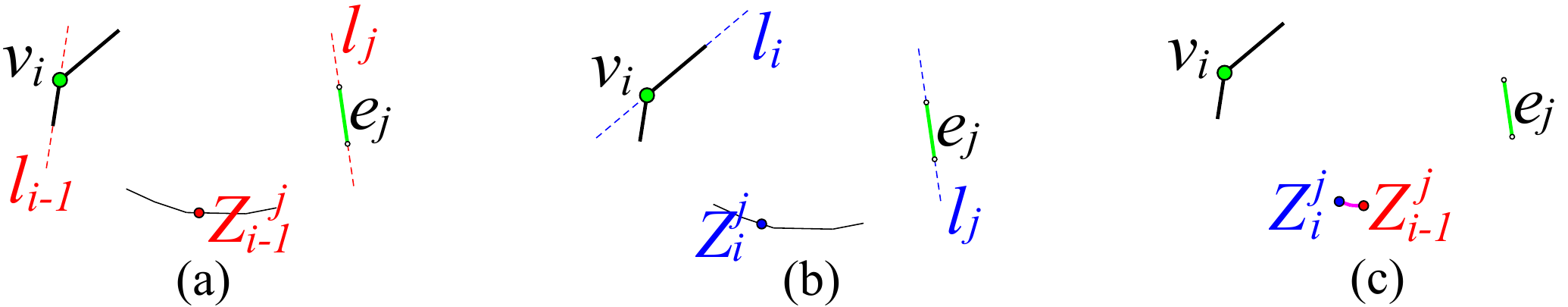}
\caption{Illustration of (\ref{eqn:zeta_chasing}). Suppose $u$ is a vertex whereas $u'$ is an edge; other cases are similar.
   Assume $u=v_i$ and $u'=e_j$.
     First, find the \emph{backward edges} of $u,u'$, which are $e_{i-1},e_j$, and
       find the point in $P$ with the maximum distance-product to the extended lines of these two edges, i.e., $Z_{i-1}^j$; see (a).
     Second, find the \emph{forward edges} of $u,u'$, which are $e_i,e_j$, and
       find the point in $P$ with the maximum distance-product to the extended lines of these two lines, i.e., $Z_i^j$; see (b).
     Then, $\zeta(u,u')$ is the boundary-portion from the first $Z$-point to the second; see (c).
     }\label{fig:def-zeta}
\end{figure}
\end{description}

\textbf{Note}: When $u,u'$ are both edges, $\zeta(u,u')$ is a point, which equals $Z_u^{u'}$.
When at least one of $u,u'$ is a vertex, $\zeta(u,u')$ could be either a boundary-portion that is not a single point, or just a single point --
  which occurs when the two $Z$-points $Z_{back(u)}^{back(u')}$, $Z_{forw(u)}^{forw(u')}$ coincide.
  For example, $Z_i^j$ can be equal to $Z_{i-1}^j$ when $v_i$ is chasing $e_j$.

\subsection{Set $\T$, and blocks and sectors of $f(\T)$}\label{subsect:T-fT}
\newcommand{\LVS}{\mathcal{L}^\star_V}
\newcommand{\RVS}{\mathcal{R}^\star_V}

Let notation $[\alpha,\beta,\gamma]$ be short for $\{(X_1,X_2,X_3)\mid X_1\in \alpha,X_2\in \beta,X_3\in \gamma \}$.
Define $f(X_1,X_2,X_3)=X_1+X_3-X_2$, i.e. the unique point $X$ so that $X_1X_2X_3X$ forms a parallelogram.
Moreover, for any $S\subseteq [\partial P, \partial P, \partial P]$, let $f(S)$ denote $\{f(X_1,X_2,X_3)\mid (X_1,X_2,X_3)\in S\}$, which is a region in the plane.

\smallskip We define $\T^P$, abbreviated as $\T$ when $P$ is clear, as the following subset of $[\partial P,\partial P,\partial P]$.
    \begin{equation}\label{def:T}
        \begin{split}
            \T^P & := \bigcup_{u \text{ is chasing } u'} [u', \zeta(u,u'),u]=
                      \bigcup_{u \text{ is chasing } u'} \left\{(X_1 \in u', X_2\in \zeta(u,u'), X_3\in u )\right\}\\
                & = \left\{(X_1,X_2,X_3)\in \partial P^3\mid \unit(X_3) \hbox{ is chasing }\unit(X_1), X_2\in \zeta(\unit(X_3),\unit(X_1))\right\}.
        \end{split}
    \end{equation}

For any unit pair $(u,u')$ in which $u$ is chasing $u'$, define
\begin{equation}\label{def:block}
    \block(u,u'):=f\left(\{(X_1,X_2,X_3)\in \T \mid X_3\in u,X_1\in u'\}\right)=f\left([u',\zeta(u,u'),u]\right).
\end{equation}

For any unit $w$, define
\begin{equation}\label{def:sector}
    \sector(w):=f(\{(X_1,X_2,X_3)\in \T \mid X_2\in w\}).
\end{equation}

We call each region in $\{\block(u,u')\mid u\hbox{ is chasing }u'\}$ a \emph{block} (of $f(\T)$)
and each region in $\{\sector(w)\mid w\hbox{ is a unit of }P\}$ a \emph{sector} (of $f(\T)$).
By the definition of blocks, the union of all the $\Theta(n^2)$ blocks is $f(\T)$.
Similarly, by the definition of sectors, the union of all the $2n$ sectors is also $f(\T)$.

\begin{figure}[h]
\includegraphics[width=.31\textwidth]{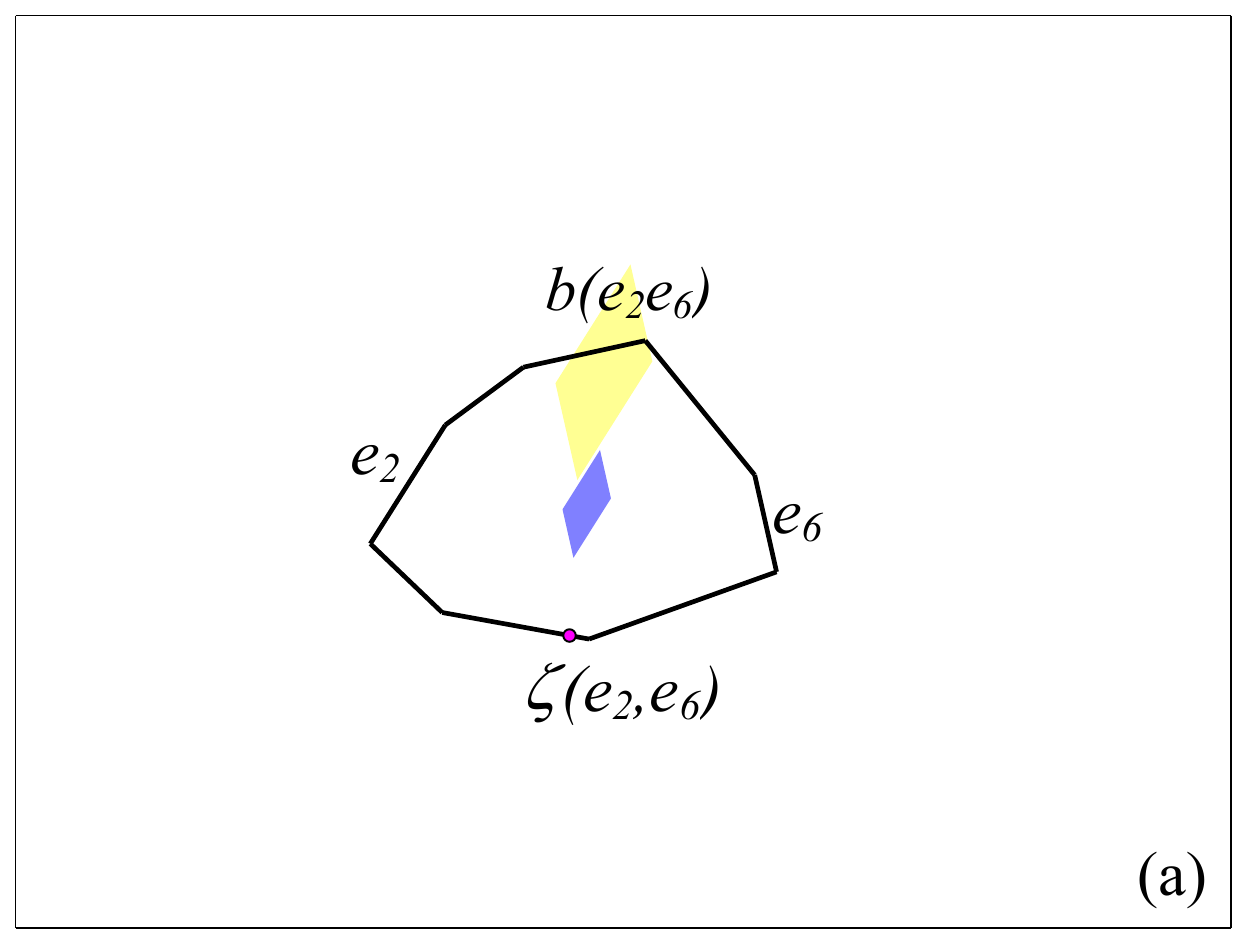}~~~~
\includegraphics[width=.31\textwidth]{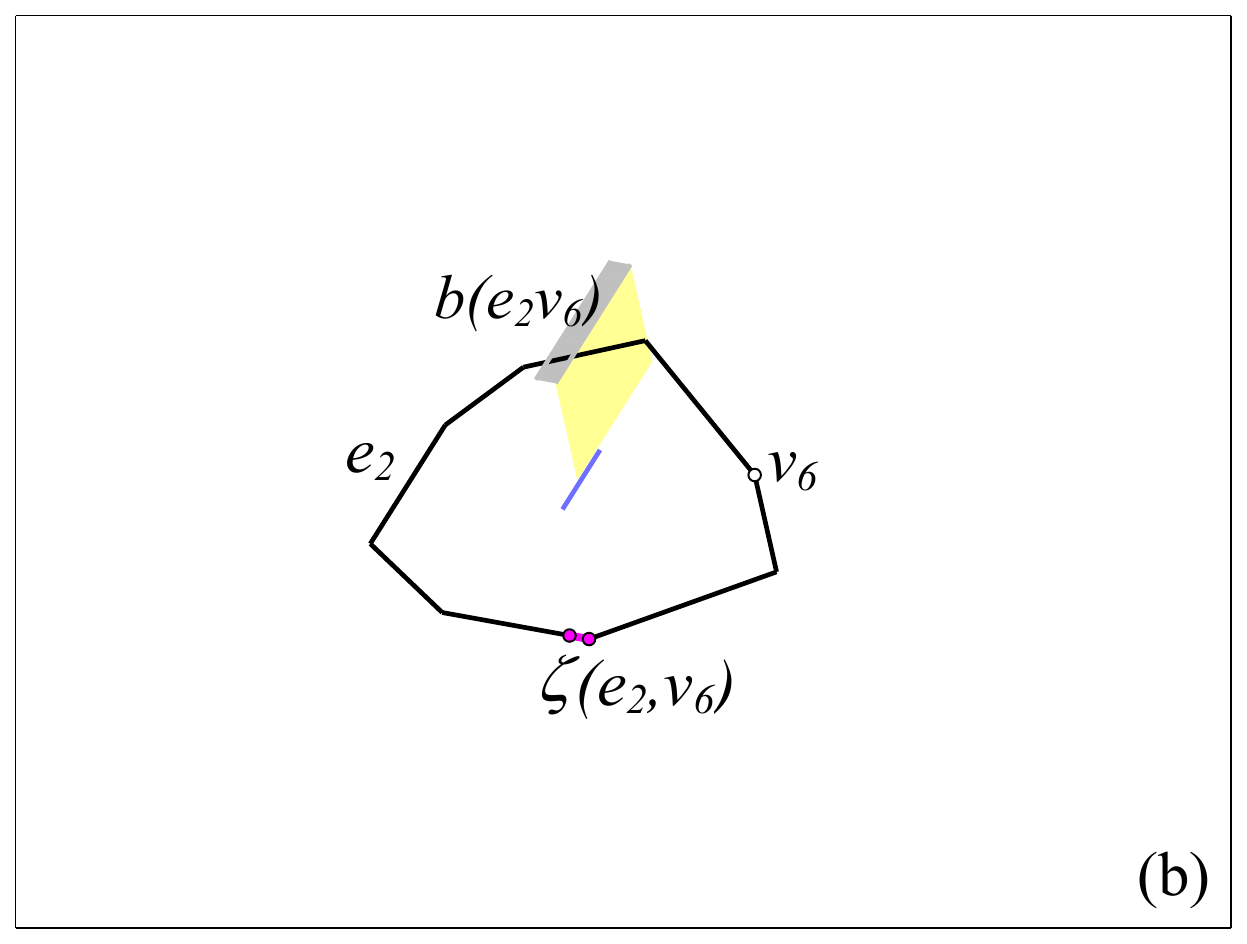}~~~~
\includegraphics[width=.31\textwidth]{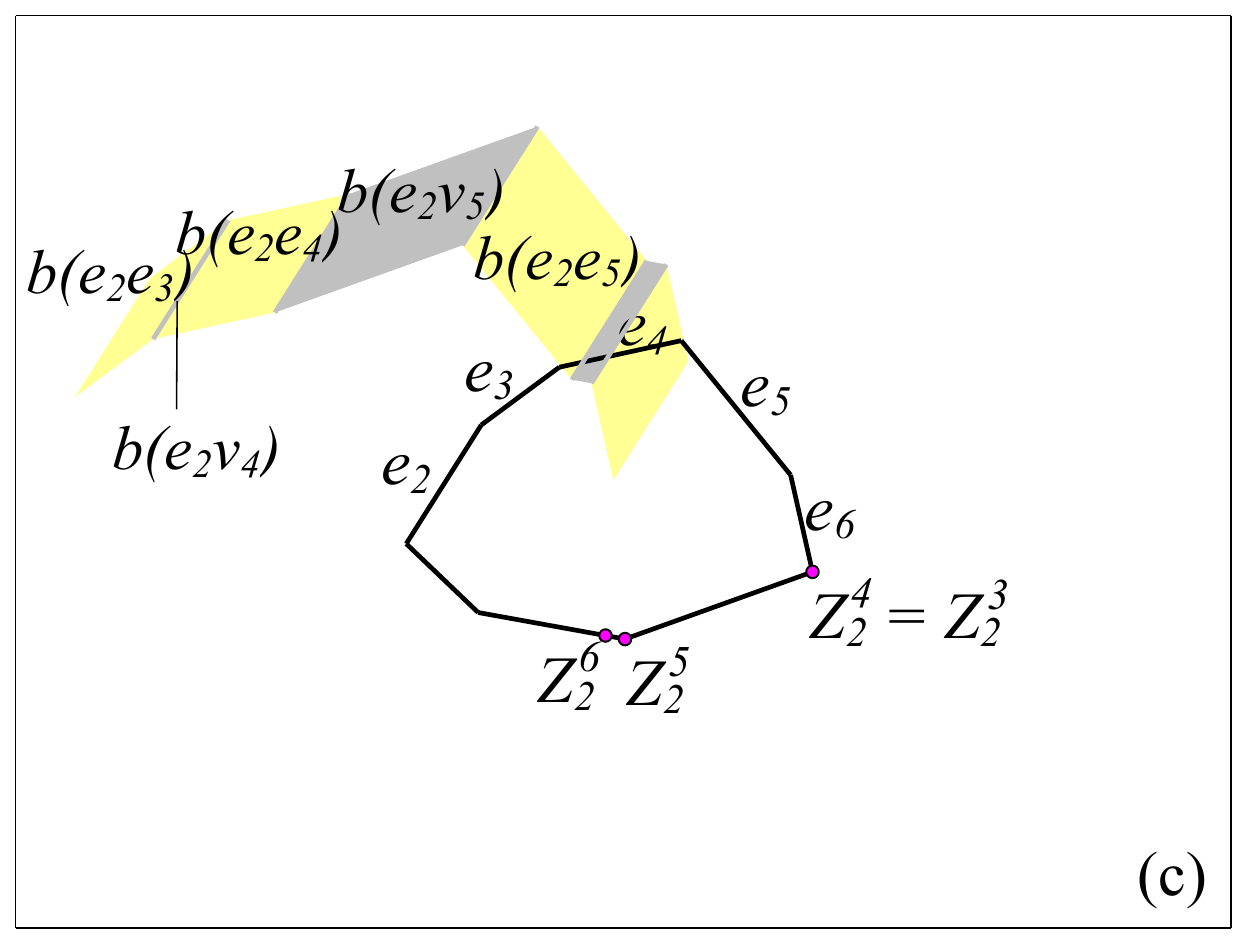}\\
\includegraphics[width=.31\textwidth]{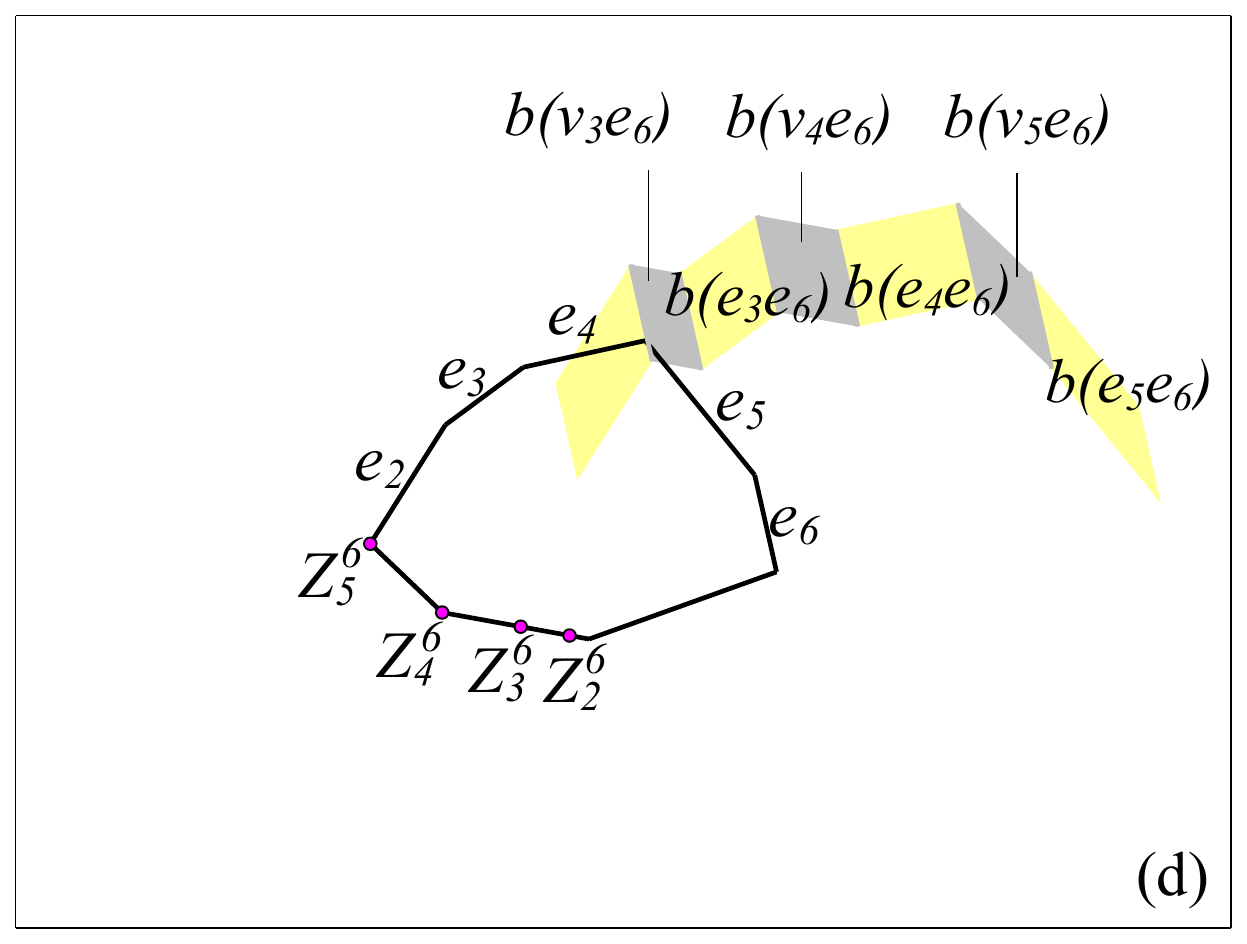}~~~~
\includegraphics[width=.31\textwidth]{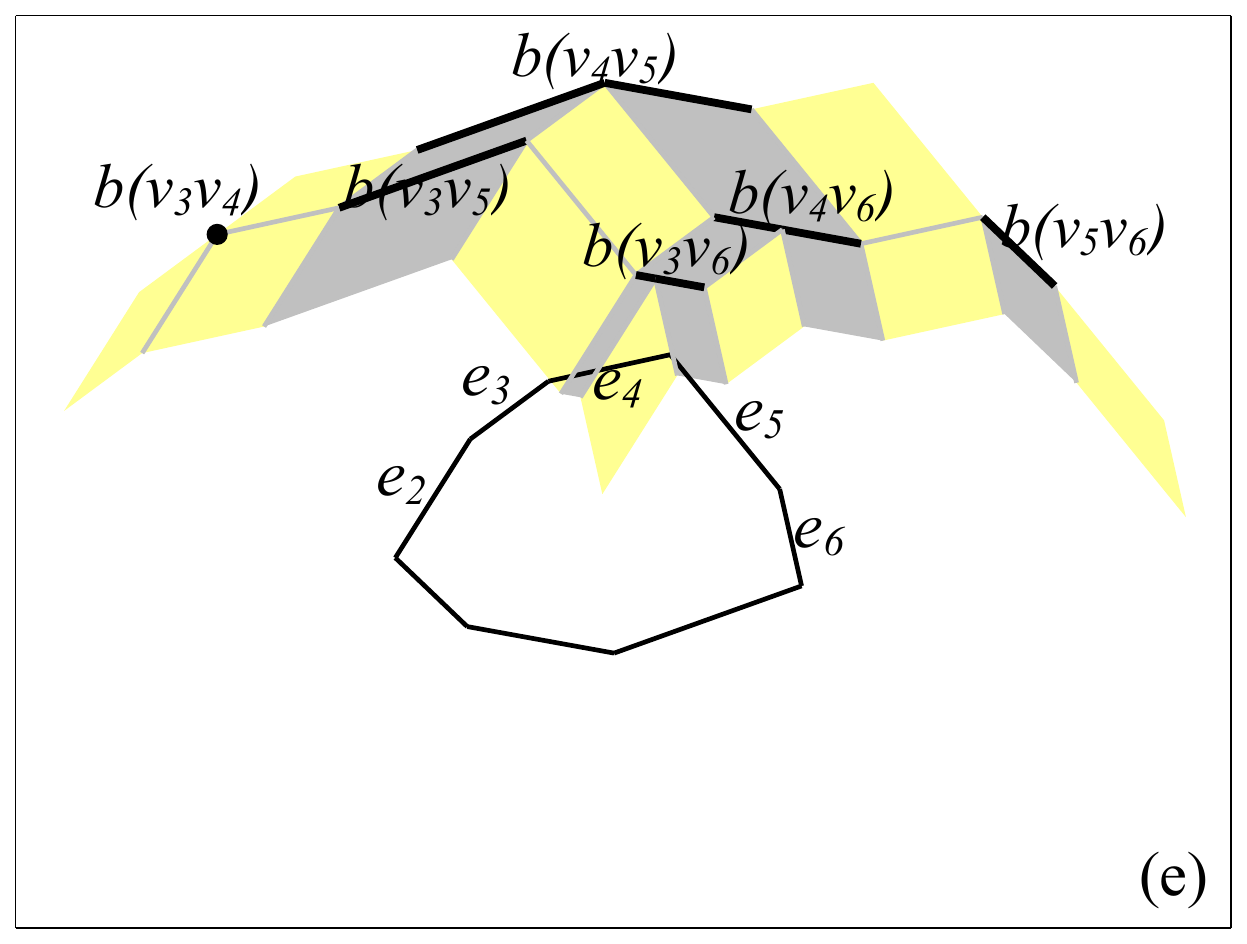}~~~~
\includegraphics[width=.31\textwidth]{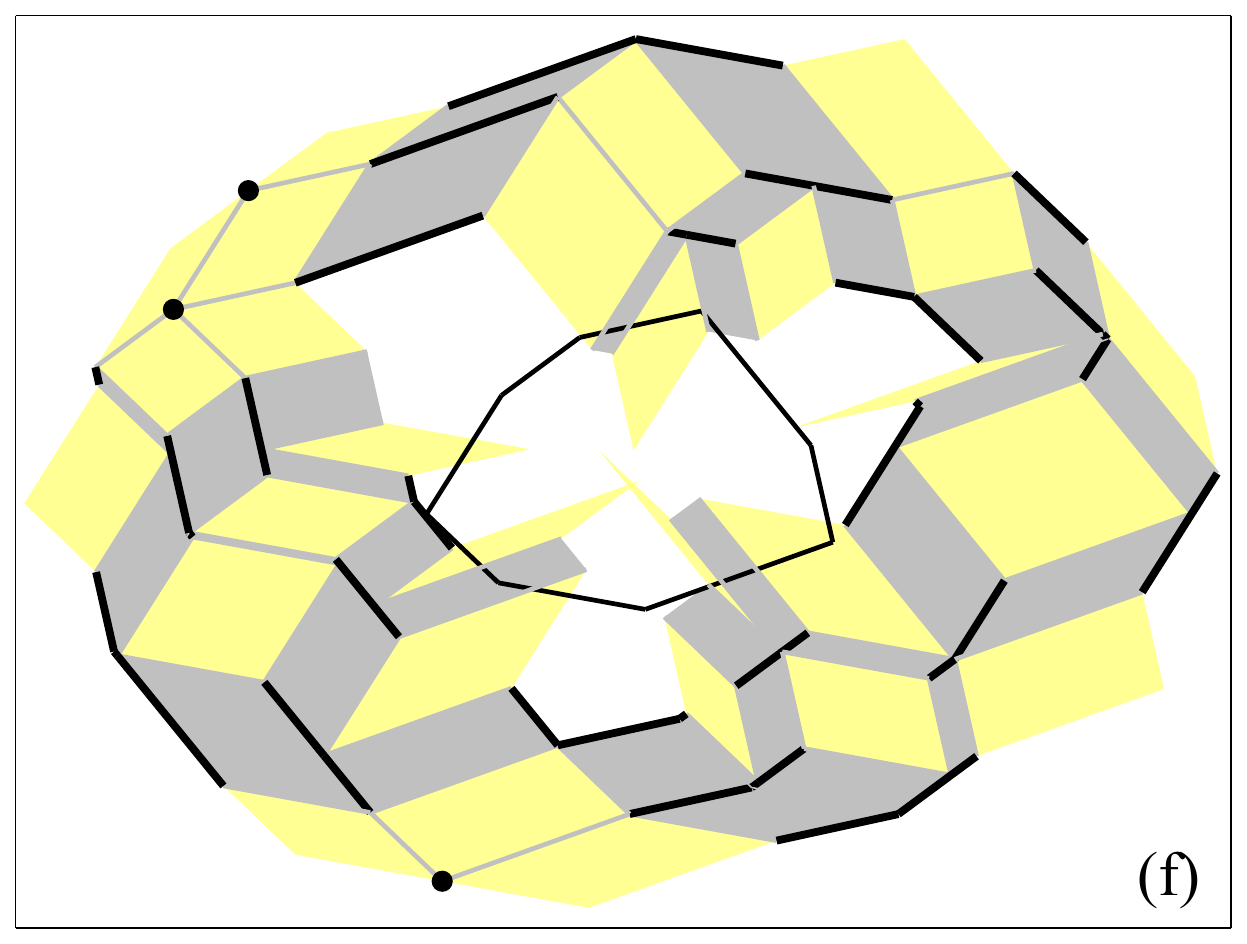}
\caption{Illustration of (\ref{def:block}). In this figure, notation $\block(u,u')$ is abbreviated as $b(uu')$.
    The region of $\block(u,u')$ is colored yellow, black, or gray, when $u,u'$ are both edges,
      both vertices, or one vertex and one edge.
}\label{fig:blocks_def}
\end{figure}

We illustrate the definition of the blocks by Figure~\ref{fig:blocks_def}.
In this figure, we consider a convex polygon with eight edges.
Picture~(a) draws $\block(e_2,e_6)$.
The blue region contains the midpoints of all those line segments connecting $e_2$ and $e_6$.
At each point in this region, draw the reflection of $Z_2^6$;
  the union of such reflections equals $f([e_6,\zeta(e_2,e_6),e_2])$, namely, $\block(e_2,e_6)$.
  It is a parallelogram with two sides congruent to $e_2$ and two sides congruent to $e_6$.
Picture~(b) draws not only $\block(e_2,e_6)$, but also $\block(e_2,v_6)$.
The blue segment contains the midpoints of all those line segments connecting $e_2$ and $v_6$.
At each point in this segment, draw the reflection of $\zeta(e_2,v_6)$;
   the union of such reflections equals $f([v_6,\zeta(e_2,v_6),e_2])$, namely, $\block(e_2,v_6)$.
It is again a parallelogram. This parallelogram has two sides congruent to $e_2$, and two sides congruent to $\zeta(e_2,v_6)$.
(However, because $\zeta(e_2,e_6)$ in general is a polygonal curve that consists of several (possibly zero) line segments,
   $\block(e_2,e_6)$ in general is a region that consists of several (possibly zero) parallelograms, each of which has two sides congruent to $e_2$.)
Picture~(c) and (d) respectively draw $\{\block(e_2,u')\mid e_2 \text{ is chasing }u'\}$ and $\{\block(u,e_6)\mid u \text{ is chasing }e_6\}$.
   Notice that $\block(e_2,v_4)$ is a line segment, because $\zeta(e_2,v_4)=[Z_2^3 \circlearrowright Z_2^4]$ contains only a single point.
Picture~(e) draws $\{\block(u,u')\mid \text{$u,u'$ are in $(v_2\circlearrowright v_7)$, and $u$ is chasing $u'$}\}$, where
  we only show labels of the black blocks.
Notice that $\block(v_3,v_4)$ is a single point because $\zeta(v_3,v_4)$ is a single point.
Finally, picture~(f) draws all the blocks.

The notion ``\emph{reflection}'' is formally defined as follows.
Given a figure $F$ and a point $O$,
  figure $F$'s reflection with respect to $O$ is another figure which is congruent to $F$ and is centrally-symmetric to $F$ with respect to $O$.

Figure~\ref{fig:sectors-nestp-sigma}~(a) below draws the $2n$ sectors for the example given in Figure~\ref{fig:blocks_def}.
This picture is drawn according to (\ref{def:sector}) as follows.
Consider any unit pair $(u,u')$ in which $u$ is chasing $u'$.
   As shown in Figure~\ref{fig:blocks_def},
      each point $X$ in $\block(u,u')$ is a reflection of some point $X_2$ in $\zeta(u,u')$.
   If $X_2$ is from unit $u$, then we put $X$ into $\sector(u)$.

\begin{figure}[h]
\begin{subfigure}[b]{.34\textwidth}
\centering \includegraphics[width=1.05\textwidth]{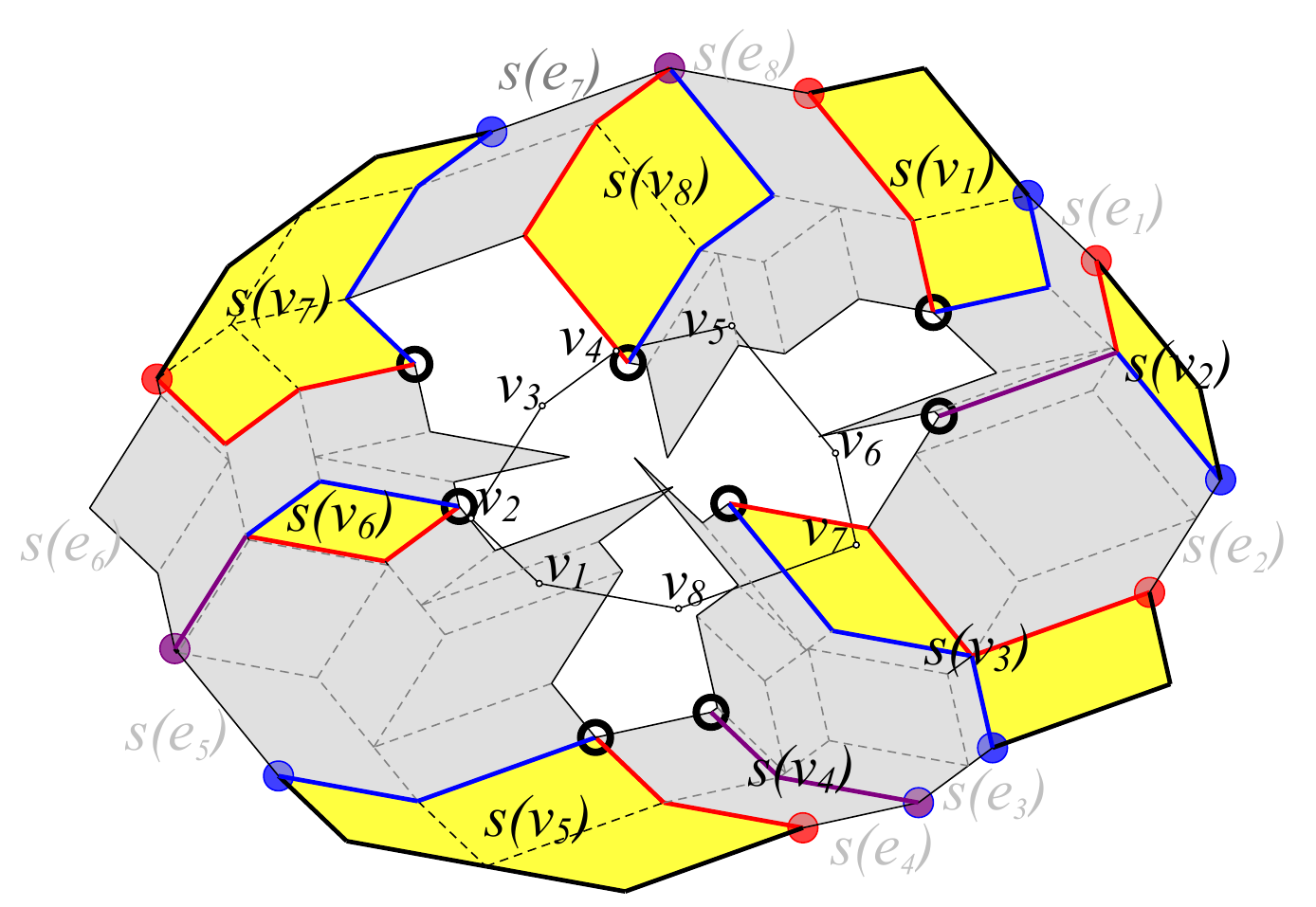}\caption{Sectors of $f(\T)$}
\end{subfigure}
\begin{subfigure}[b]{.34\textwidth}
\centering \includegraphics[width=0.91\textwidth]{NestP.pdf}\caption{$\Nest(P)$}
\end{subfigure}
\begin{subfigure}[b]{.32\textwidth}
\centering \includegraphics[width=0.95\textwidth]{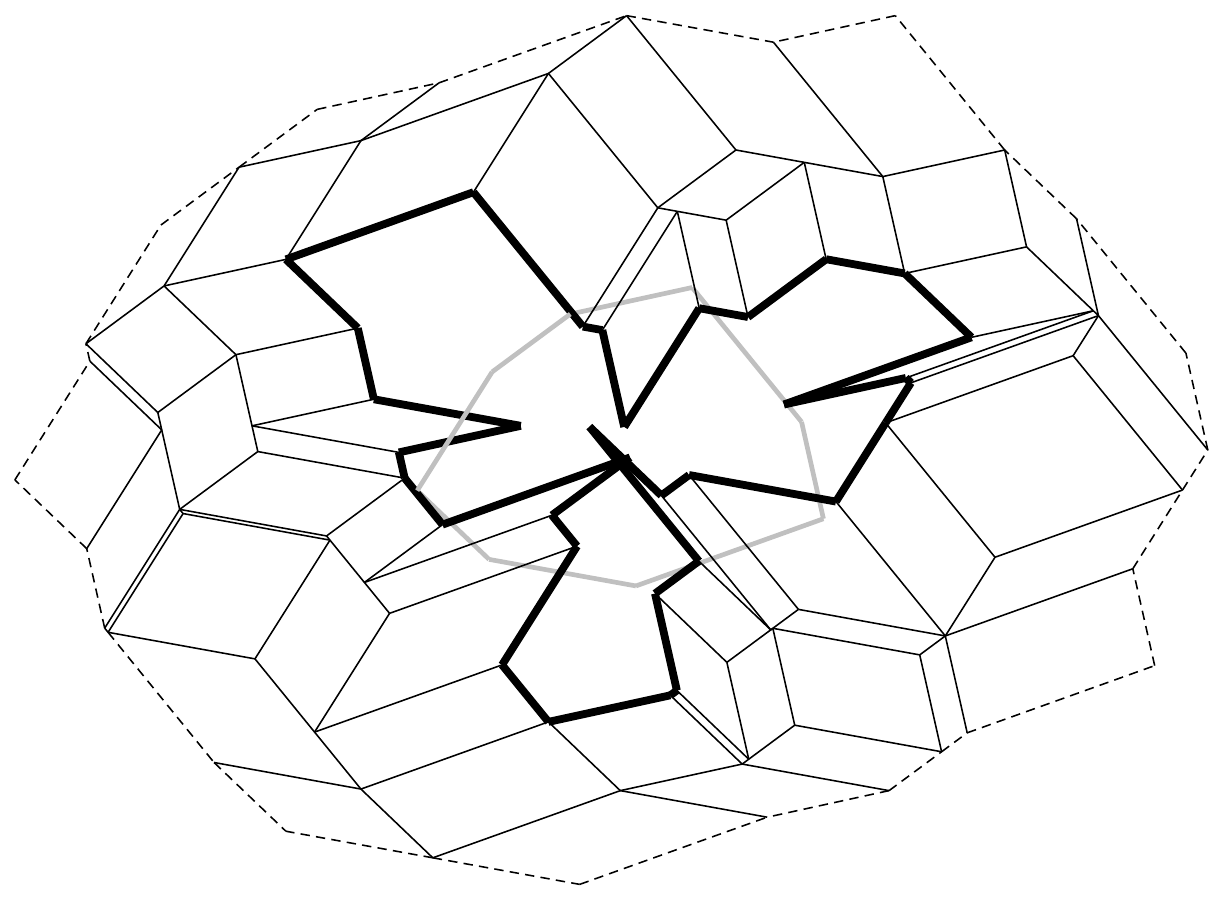}\caption{Inner and outer boundaries of $f(\T)$}
\end{subfigure}
\caption{Picture~(a) illustrates (\ref{def:sector}).
In this picture, $\sector(w)$ is abbreviated as $s(w)$ and
  its region is colored yellow or gray, distinguished by whether $w$ is a vertex or an edge.
Picture~(b) draws $\Nest(P)$ defined below.
The bolded and dashed (closed) curves in picture~(c) indicate the inner and outer boundaries of $f(\T^P)$ defined below.}\label{fig:sectors-nestp-sigma}
\end{figure}

It is worthwhile to mention that among the $2n$ sectors, those in $\{\sector(V) \mid V \hbox{ is a vertex}\}$ are more important (in this manuscript).
  We will study them more than the other sectors (especially in sections~\ref{sect:techover} and \ref{sect:fT-major}).

\subsection{An introduction of $\Nest(P)$ and some additional terminologies}\label{subsect:nestp}

It is not difficult to see that \emph{each block is a connected area bounded by a polygonal boundary}.
  An explicit definition of the boundary will be given in subsection~\ref{subsect:pre-borders-sigmap}.
  Gladly, the sectors admit the same property; at least, (i) \emph{when $w$ is a vertex, $\sector(w)$ is a connected area bounded by a polygonal boundary.}
  (When $w$ is an edge, it is also true; yet this result will neither be proved, nor be applied.)
  However, our proof of fact~(i) is complicated and is deferred for a while (sketched in section~\ref{sect:techover} and fully given in section~\ref{sect:fT-major}).
  In the proof, we will explicitly define two simple polygonal curves $\LVS$ and $\RVS$ for each vertex $V$ (colored red and blue in Figure~\ref{fig:sectors-nestp-sigma}~(a)) and prove that they are on the boundary of $\sector(V)$.
Based on those boundaries of the blocks and sectors, we can define $\Nest(P)$.

\begin{description}
\item [$\Nest(P)$.] We define $\Nest(P)$ as the union of $\LVS,\RVS \mid V\in \{v_1,\ldots,v_n\}$ and the boundaries of the blocks.
    (Alternatively, we could define $\Nest(P)$ as the union of the boundaries of the blocks and the boundaries of the sectors.
    Yet in this way we have to define the boundary of $\sector(w)$ for every edge $w$, which would be more complicated than the case where $w$ is a vertex.
    So we do not use this definition. The equivalence of the two definitions is useless and unimportant; so proof omitted.)
    On the terminology, since $\Nest(P)$ captures all the subregions of $f(\T)$ including blocks and sectors, it can be called a \textbf{skeleton} of $f(\T)$.
    Moreover, it is an \textbf{arrangement} \cite{BergCG} of several line segments, each of which is parallel to an edge of $P$;
      see Figure~\ref{fig:sectors-nestp-sigma}~(b).
\end{description}

\begin{remark}
Though $\Nest(P)$ plays an important role in this manuscript, we are not in a hurry to give its full description (which contains
  the definition of $\LVS,\RVS$ for example),
   because for describing our main results in the upcoming subsection, we only need those objects that are already well-defined, e.g., the blocks and sectors, and $\T$, but not $\Nest(P)$.
        Besides, it is unwise to give the full description here for the conciseness of this introduction.

\medskip Yet, why do we bother to formally define $\Nest(P)$ below?
First, we have to formally define the boundaries of the blocks and $\sector(V)$ for any vertex $V$ in order to prove our main results,
   and defining $\Nest(P)$ does not require any extra work.
Second, having a well-defined overall structure $\Nest(P)$ makes it much easier to abstract our work.
   Briefly, our main result (Theorem~\ref{thm:nestp} below) states that $\Nest(P)$ has a surprising good interaction with $\partial P$,
     whereas our secondary result (Theorem~\ref{thm:nestp-location} below) states that some location queries on $\Nest(P)$ can be answered efficiently.
\end{remark}

\begin{description}
\item [On the directions of the line segments in $\Nest(P)$.] Recall that the boundary-portions of $P$ are all directional.
    We also regard the boundaries of the blocks and the boundaries of the sectors ($\LVS$ and $\RVS$)
      as directional. So all the line segments in $\Nest(P)$ are directional.
    The following principle is applied in defining the directions: If a part of boundary (of a block or a sector) $\alpha$ is a copy of
        a boundary-portion $\beta$ (perhaps rearranged after the copy), the direction of $\alpha$ will conform with the direction of $\beta$.
        See an example in the first picture of Figure~\ref{fig:examples-regular}. The details can be found at the place
          where we define the boundaries of the blocks as well as $\LVS$ and $\RVS$.

\item [On the size of $\Nest(P)$.]
    The size of $\Nest(P)$ is the number of line segments in $\Nest(P)$,
      which is obviously $\Omega(n^2)$ since there are $\Omega(n^2)$ blocks.
    It can also be bounded by $O(n^2)$ easily; see a proof in the appendix.
      (But the fact that size of $\Nest(P)$ is $O(n^2)$ is unimportant for understanding our two results.)
\end{description}

\begin{figure}[h]
  \centering \includegraphics[width=.45\textwidth]{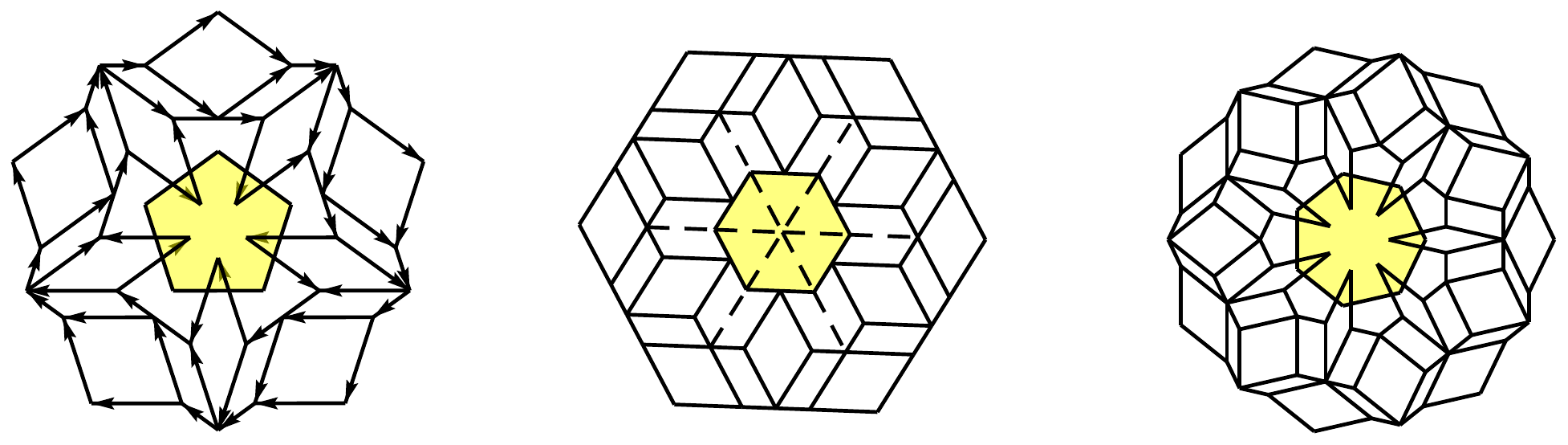}\\
  \caption{Examples of $\Nest(P)$ for regular $n$-side polygon for $n=5,6,7$.}\label{fig:examples-regular}
\end{figure}

Below we give some more notations and terminologies that are important for describing our results.
\begin{description}
\item [Subset $\T^*$.] Let $\T^*$ denote the subset of $\T$ that is mapped to $\partial P$ under $f$.
\item [The inner boundary $\sigma P$ of $f(\T)$.]
As the reader may have noticed, by the definition of chasing, $f(\T)$ should be annular --
    a connected region bounded by two disjoint closed curves, one of which contained in the other; see Figure~\ref{fig:sectors-nestp-sigma}~(a).
 Thus $f(\T)$ has an \emph{outer boundary} and an \emph{inner boundary}; see Figure~\ref{fig:sectors-nestp-sigma}~(c).
  We denote the inner one by $\sigma P$. Its direction conforms with the clockwise order.
  A formal definition is given in subsection~\ref{subsect:pre-borders-sigmap}.
\item [Interleaving.] Given two oriented closed curves, we say they \emph{interleave} if starting from any intersection between them,
  no matter we travel around which of them of a cycle, we meet all their intersections in identical order.
\end{description}

\subsection{Our results}\label{sect:main-nestp}

\newcommand{\BD}{\textsc{Block-disjointness} }
\newcommand{\IP}{\textsc{Interleavity-of-$f$} }
\newcommand{\MF}{\textsc{Monotonicity-of-$f$} }
\newcommand{\RF}{\textsc{Reversiblity-of-$f$} }
\newcommand{\SM}{\textsc{Sector-monotonicity} }
\newcommand{\SC}{\textsc{Sector-continuity} }

\begin{theorem}[Main result: Six structural properties of $\Nest(P)$]\label{thm:nestp} See Figure~\ref{fig:sectors-nestp-sigma}.
\begin{description}
\item[Block-disjointness.] The intersection of any pair of blocks lies in the interior of $P$.

    Be aware that it does \textbf{not} state that all blocks are pairwise-disjoint.
\item[Interleavity-of-$f$.] The inner boundary of $f(\T)$ (i.e.\ the curve $\sigma P$) interleaves $\partial P$.
\item[Reversibility-of-$f$.] Function $f$ is a bijection from  $\T^*$ to its image set $f(\T^*)=f(\T)\cap \partial P$.

    Henceforth, we denote the reverse function of $f$ on $f(\T)\cap \partial P$ by $f^{-1}$.
    Moreover, denote the 1st, 2nd, and 3rd dimension of $f^{-1}(X)$ by $f^{-1}_1(X),f^{-1}_2(X)$ and $f^{-1}_3(X)$ respectively.
\item[Monotonicity-of-$f$.] Function $f^{-1}_2$, a mapping from $f(\T)\cap \partial P$ to $\partial P$ defined above, is ``circularly monotone'':
    If a point $X$ travels around $f(\T)\cap \partial P$ in clockwise, $f^{-1}_2(X)$ would shift in clockwise around $\partial P$ non-strictly,
    and moreover, when $X$ has traveled exactly a cycle, $f^{-1}_2(X)$ would also have traveled exactly a cycle.
\item[Sector-monotonicity.] The $2n$ districts $\sector(v_1)\cap \partial P$, $\sector(e_1)\cap \partial P$, \ldots, $\sector(v_n)\cap \partial P$, $\sector(e_n)\cap \partial P$ are pairwise-disjoint and arranged in clockwise order around $\partial P$.
\item[Sector-continuity.] For any vertex $V$, the intersection between $\sector(V)$ and $\partial P$ is continuous.
\end{description}
\end{theorem}

Our proof of these properties is inevitably lengthy, due to the intricacy of $\Nest(P)$.
However, it is highly modularized and easy to follow.
We sketch the proof in section~\ref{sect:techover} and then give the full proof in sections~\ref{sect:lemmas-proof}, \ref{sect:fT-major}, and \ref{sect:othertwo}.

\begin{remark}
Among the elements in $\T$, those in $\T^*$ deserve special attention --
  all of the properties above essentially concern $f(\T^*)$, rather than $f(\T)$.
  As such, $\Nest(P)$ has a nice interaction with $\partial P$ (beyond our anticipation).
\end{remark}

\begin{description}
\item [Two more interesting properties.] 1. \emph{If we travel along the (oriented) segments in $\Nest(P)$ (see the first picture of Figure~\ref{fig:examples-regular}) of one cycle, starting and terminating at the same node, the total distance would be 3 times of the perimeter of $P$ no matter which path we choose.} 2. \emph{For any non-degenerate linear transformation $\Gamma$, it holds that $\Gamma(\Nest(P))=\Nest(\Gamma(P))$.}
    These properties are not applied in the paper; their easy proofs are omitted.
\end{description}

\begin{theorem}[Secondary result: efficient locations on $\Nest(P)$]\label{thm:nestp-location}
  We can answer each of the following location queries, in (amortized) $O(\log^2n)$ time,
    without a full construction of $\Nest(P)$ which takes $\Omega(n^2)$ time.
\begin{description}
\item[Sector-intersect-Units Query.] Given any vertex $V$, find the interval of units that intersect $\sector(V)$.\\
    \textbf{Note}: Due to \SC, those units intersect $\sector(V)$ are indeed an interval of units.
\item[Vertex-in-Sector Query.] Given any vertex $V$, find $w$ so that $\sector(w)$ contains $V$.\\
    \textbf{Note}: There is at most one such $w$ due to \SM. If there is no such $w$, output NIL.
\item[Vertex-in-Block Query.] Given any vertex $V$, find  $(u,u')$ so that $\block(u,u')$ contains $V$.\\
    \textbf{Note}: There is at most one such $(u,u')$ due to \BD. If no such pair exists, output NIL.
\end{description}
\end{theorem}

The proof of Theorem~\ref{thm:nestp-location} is given in section~\ref{sect:location-answers}, and sketched in section~\ref{sect:techover}.

\begin{remark}\label{remark:proof-need-bq}
Among others, \BD, \IP, and Vertex-in-Block query are the most difficult to prove. Our proofs of these three results
   utilize a group of auxiliary objects called \emph{bounding-quadrants}
     (introduced in subsection~\ref{subsect:bounding-quadrants}), each of which is a relax of a block.
   These auxiliary objects are so important that
        it is no exaggerate to credit the most nontrivial step within the entire proof to the introduction of these objects.
\end{remark}

\paragraph{Related work.} $\Nest(P)$ is novel and has few related work.
It has a similar appearance as \emph{Zonotopes} \cite{LecturesPolytopes} considering that many groups of parallel lines exist.
  It resembles \emph{Voronoi Diagrams} \cite{BergCG} since they are arrangements of line segments.
Whereas a Voronoi Diagram cares the distance, $\Nest(P)$ cares the distance-product.

\subparagraph*{Organization of this paper.}
    Section~\ref{sect:pre} formally defines some geometric objects, including the boundaries of the blocks, the inner boundary of $f(\T)$ (i.e., $\sigma P$), and the bounding-quadrants mentioned in Remark~\ref{remark:proof-need-bq}.
    It also states seven interesting lemmas that are crucial to our final proof. The proofs of these lemmas are deferred to section~\ref{sect:lemmas-proof}.
     Section~\ref{sect:techover} outlines our techniques for proving Theorems~\ref{thm:nestp} and \ref{thm:nestp-location}, and
        its subsequent sections provide the details.

\smallskip See Figure~\ref{fig:flow} in appendix for the key geometric objects studied in this manuscript and their relations.
\clearpage

\section{Preliminaries: some rigorous definitions and some important lemmas}\label{sect:pre}

The following notions will be frequently applied in this section and henceforth.
\begin{description}
\item[Region $u\oplus u'$.] When unit $u$ is chasing $u'$, we denote
    $u\oplus u'=\{(X+X')/2\mid X\in u, X'\in u'\}$.
    In particular, for $(u,u')=(e_i,e_j)$ where $e_i\prec e_j$, region $u\oplus u'=e_i\oplus e_j$ is a parallelogram as shown in Figure~\ref{fig:terms}~(a).
\item[$k$-scaling.] Given a figure $F$, a point $O$, and a ratio $k>0$, we define the \emph{$k$-scaling} of $F$ with respect to $O$ as the figure $F'$, which contains point $X$ if and only if $F$ contains $(X-O)/k+O$.  See Figure~\ref{fig:terms}~(b).
\item[Inferior portions.] We call $[v_i\circlearrowright v_{j+1}]$ an \emph{inferior portion} if and only if $e_i\preceq e_j$; see Figure~\ref{fig:terms}~(c) and (d).
\end{description}

\begin{figure}[h]
\centering\includegraphics[width=.8\textwidth]{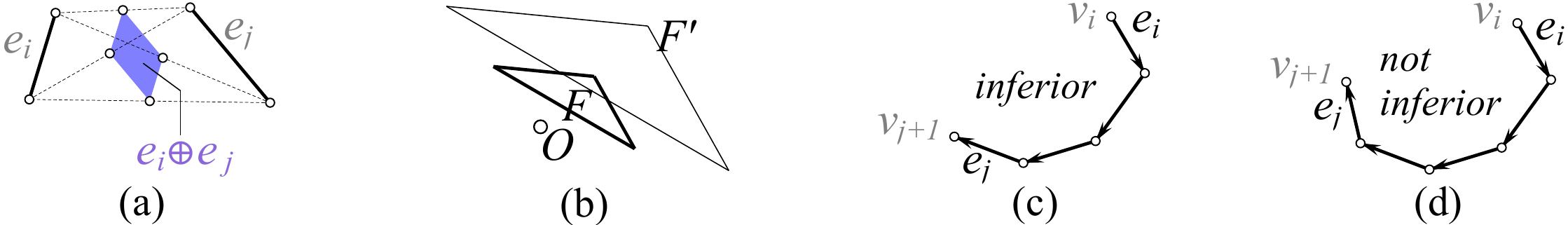}
\caption{Picture (a) shows $u\oplus u'$ for which $u=e_i$ and $u'=e_j$ are both edges; (b) shows 2-scaling of $F$ with respect to $O$; (c) shows a typical inferior portion; whereas (d) shows a boundary-portion that is not inferior.}\label{fig:terms}
\end{figure}

\subsection{The boundaries of the blocks and the inner boundary of $f(\T)$}\label{subsect:pre-borders-sigmap}

Before defining the boundaries of the blocks, we state two formulas of $\block(u,u')$.
\begin{align}
\block(u,u')&={\bigcup}_{X\in u\oplus u'}\hbox{the reflection of $\zeta(u,u')$ with respect to }X.  \label{eqn:block_reflect}\\[5pt]
\block(u,u')&={\bigcup}_{X\in \zeta(u,u')}\hbox{the $2$-scaling of $u\oplus u'$ with respect to }X.  \label{eqn:block_scale}
\end{align}

\noindent \emph{Proof}. \vspace{-22pt}
\[\begin{split}
\block(u,u')&={\bigcup}_{X_3\in u,X_1\in u',X_2\in\zeta(u,u')}f(X_1,X_2,X_3)\\
&={\bigcup}_{X_3\in u,X_1\in u'}{\bigcup}_{X_2\in \zeta(u,u')}\hbox{the reflection of $X_2$ with respect to $(X_3+X_1)/2$}\\
&={\bigcup}_{X_3\in u,X_1\in u'}\hbox{the reflection of $\zeta(u,u')$ with respect to $(X_3+X_1)/2$}\\
&={\bigcup}_{X\in u\oplus u'}\hbox{the reflection of $\zeta(u,u')$ with respect to $X$}.
\end{split}\]
\[\begin{split}
\block(u,u')&={\bigcup}_{X_3\in u,X_1\in u',X_2\in\zeta(u,u')}f(X_1,X_2,X_3)\\
&={\bigcup}_{X_2\in \zeta(u,u')}{\bigcup}_{X_3\in u,X_1\in u'}
    \hbox{the $2$-scaling of $(X_3+X_1)/2$ with respect to $X_2$}\\
&={\bigcup}_{X_2\in \zeta(u,u')}\hbox{the $2$-scaling of $u\oplus u'$ with respect to $X_2$}.
\end{split}\]
\qed

In the following we define the \emph{borders} of a block. The boundary of a block is the union of its borders.

\begin{definition}[Borders and boundaries of the blocks]\label{def:block-borders}
Assume $u$ is chasing $u'$.
\begin{itemize}
\item[Case~1] $(u,u')=(e_i,e_j)$. See Figure~\ref{fig:block-borders}~(a).\\
    By (\ref{eqn:block_scale}), $\block(e_i,e_j)$ is the $2$-scaling of $e_i\oplus e_j$ with respect to $Z_i^j$ --
       a parallelogram whose sides are congruent to $e_i$ or $e_j$.
    Each side of this parallelogram is called a \emph{border} of $\block(e_i,e_j)$.
    Those sides that are congruent to $e_i$ have the same direction as $e_i$.
    Those sides that are congruent to $e_j$ have the same direction as $e_j$.
\item[Case~2] $(u,u')=(v_i,v_j)$. See Figure~\ref{fig:block-borders}~(d).\\
    By (\ref{eqn:block_reflect}), $\block(v_i,v_j)$ is the reflection of $\zeta(v_i,v_j)$ with respect to $(v_i+v_j)/2$.
    This curve is referred to as the unique \emph{border} of $\block(v_i,v_j)$, directed from the reflection of $Z_{i-1}^{j-1}$ to the reflection of $Z_i^j$.
\item[Case~3] $(u,u')=(v_i,e_j)$. See Figure~\ref{fig:block-borders}~(b).\\
    In this case, by (\ref{eqn:block_reflect}) and (\ref{eqn:block_scale}), $\block(v_i,e_j)$ is the region bounded by four curves:

    \quad the $2$-scalings of segment $v_i\oplus e_j$ with respect to $Z_{i-1}^j$ and $Z_i^j$, respectively, and

    \quad the reflections of $\zeta(v_i,e_j)$ with respect to $(v_i+v_j)/2$ and $(v_i+v_{j+1})/2$, respectively.

    Each of these four curves is called a \emph{border} of $\block(v_i,e_j)$.
    The first two borders have the same direction as $e_j$; whereas the directions of the last two borders are from the reflection of $Z_{i-1}^j$ to the reflection of $Z_i^j$.
\item[Case~4] $(u,u')=(e_i,v_j)$. See Figure~\ref{fig:block-borders}~(c). This case is symmetric to Case~3; so we can define four borders.
\end{itemize}
\end{definition}

\begin{figure}[h]
  \centering \includegraphics[width=.95\textwidth]{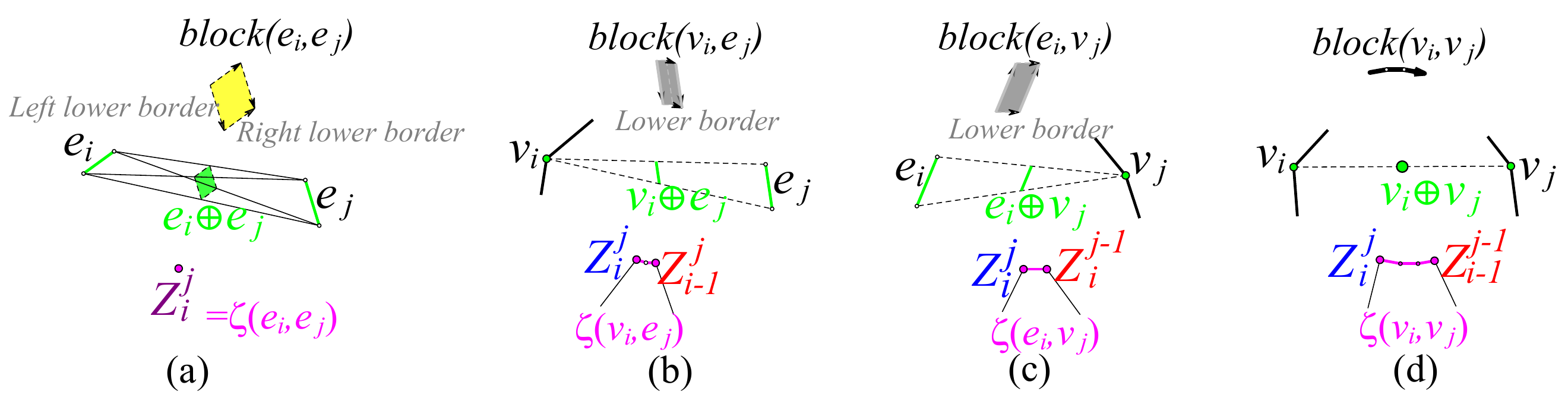}\\
  \caption{Illustration of the borders of the blocks. The directions of borders are indicated by the arrows.
     The ``left lower border'', ``right lower border'', and ``lower border'' marked in this figure
        are introduced in Definition~\ref{def:lower-borders}.}\label{fig:block-borders}
\end{figure}

\begin{definition}\label{def:lower-borders}
The following terms are applied for defining $\sigma P$ in the next. See Figure~\ref{fig:block-borders}.
\begin{itemize}
\item The \emph{left lower border} of $\block(e_i,e_j)$ refers to the $2$-scaling of $v_i\oplus e_j$ with respect to  $Z_i^j$.
\item The \emph{right lower border} of $\block(e_i,e_j)$ refers to the $2$-scaling of $e_i\oplus v_{j+1}$ with respect to  $Z_i^j$.
\item The \emph{lower border} of $\block(v_i,e_j)$ refers to the reflection of $\zeta(v_i,e_j)$ with respect to $(v_i+v_{j+1})/2$.
\item The \emph{lower border} of $\block(e_i,v_j)$ refers to the reflection of $\zeta(e_i,v_j)$ with respect to $(v_i+v_j)/2$.\medskip
\end{itemize}
\end{definition}

\subparagraph{Outline for defining $\sigma P$.}
We first introduce a group of blocks called \emph{frontier blocks} and define the \emph{bottom borders of the frontier blocks}.
Briefly, the frontier blocks are those that lie at the inner side of $f(\T)$.
   We will see the frontier blocks have an intrinsic order according to their definitions.
We then define the concatenation of the bottom borders of the frontier blocks (in the intrinsic order), which is a closed polygonal curve, to be $\sigma P$.
See Figure~\ref{fig:sigmaP-functionG}~(a).

\medskip To define the frontier blocks, we define a circular list of unit pairs, called \emph{frontier-pair-list}.
It is defined as $\mathsf{FPL}$ generated by Algorithm~\ref{alg:FPL-def}. See Figure~\ref{fig:sigmaP-functionG}~(b) for an illustration.
According to the frontier-pair-list, we can then find out all the frontier blocks ---
    $\block(u,u')$ is \emph{frontier} if and only if $(u,u')$ belongs to this circular list.

\begin{algorithm}[h]
\caption{An algorithm for defining $\mathsf{FPL}$}\label{alg:FPL-def}
Let $\mathsf{FPL}$ be empty, let $i=1$, and let $e_j$ be edge that $e_1\prec e_j$ but $e_{j+1}\prec e_1$;\\
\Repeat{$i=1$ and $e_j$ is the edge that $e_1\prec e_j$ but $e_{j+1}\prec e_1$}{
    Add unit pair $(e_i,e_j)$ to the tail of $\mathsf{FPL}$;\\
    \textbf{if} $e_i\prec e_{j+1}$ \textbf{then} Add unit pair $(e_i,v_{j+1})$ to the tail of $\mathsf{FPL}$ and increase $j$ by $1$;\\
    \Else{
        \textbf{if} $i+1\neq j$ \textbf{then} Add unit pair $(v_{i+1},e_j)$ to the tail of $\mathsf{FPL}$ and increase $i$ by $1$;\\
        \textbf{else} Add unit pair $(v_{i+1},v_{j+1})$ to the tail of $\mathsf{FPL}$ and increase $i,j$ both by $1$;
    }
}
\end{algorithm}

\begin{definition}\label{def:bottom-borders}
Suppose $(u,u')\in \mathsf{FPL}$. The \emph{bottom border} of $\block(u,u')$ is defined as follows.
\begin{itemize}
\item If $u,u'$ are vertices, the unique border of $\block(u,u')$ is also the \emph{bottom border} of $\block(u,u')$.
\item If $u,u'$ are an edge and a vertex, the lower border of $\block(u,u')$ is also the \emph{bottom border} of $\block(u,u')$.
\item If $u,u'$ are edges, e.g.\ $u=e_i,u'=e_j$, we define the \emph{bottom border} of $\block(u,u')$ to be
\[\begin{cases}
  \hbox{an empty set,}& \hbox{if }(e_{i-1},e_j)\in \mathsf{FPL}, (e_i,e_{j+1})\in \mathsf{FPL}.\\
  \hbox{its right lower border,}& \hbox{if }(e_{i-1},e_j)\in \mathsf{FPL}, (e_i,e_{j+1})\notin \mathsf{FPL};\\
  \hbox{its left lower border,}& \hbox{if }(e_{i-1},e_j)\notin \mathsf{FPL}, (e_i,e_{j+1})\in \mathsf{FPL};\\
  \hbox{concatenation of its two lower borders,}& \hbox{if }(e_{i-1},e_j)\notin \mathsf{FPL}, (e_i,e_{j+1})\notin \mathsf{FPL};
\end{cases}\]
\end{itemize}
\end{definition}

\begin{figure}[h]
\begin{subfigure}[b]{0.33\textwidth}
\centering \includegraphics[width=.8\textwidth]{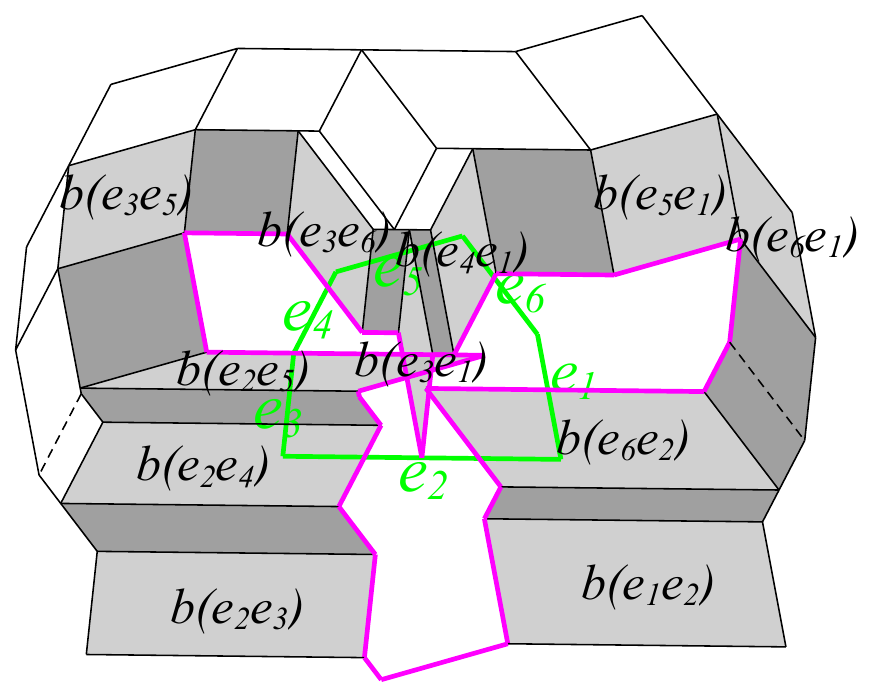}\caption{Frontier blocks}
\end{subfigure}
\begin{subfigure}[b]{0.33\textwidth}
\centering \includegraphics[width=.62\textwidth]{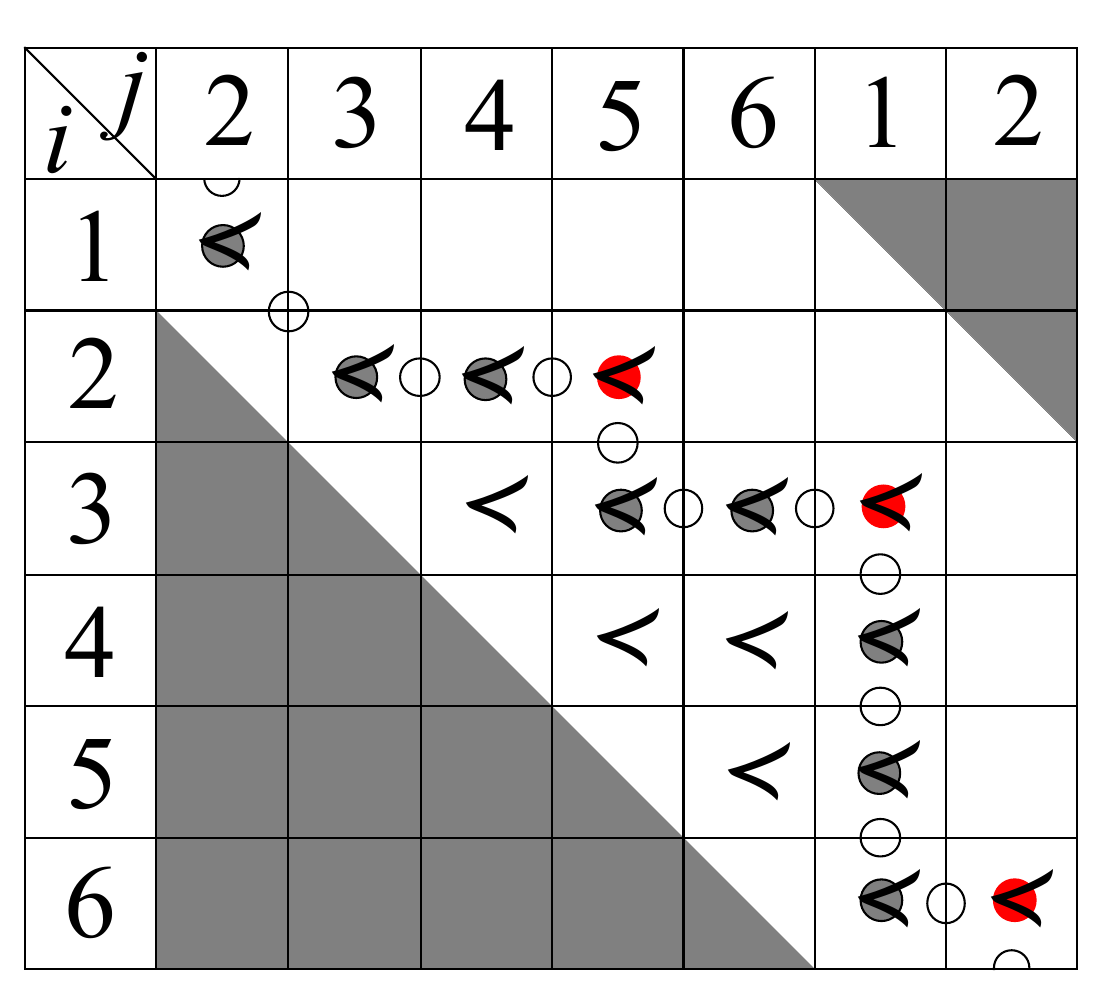}\caption{Frontier-pair-list}
\end{subfigure}
\begin{subfigure}[b]{0.33\textwidth}
\centering \includegraphics[width=.75\textwidth]{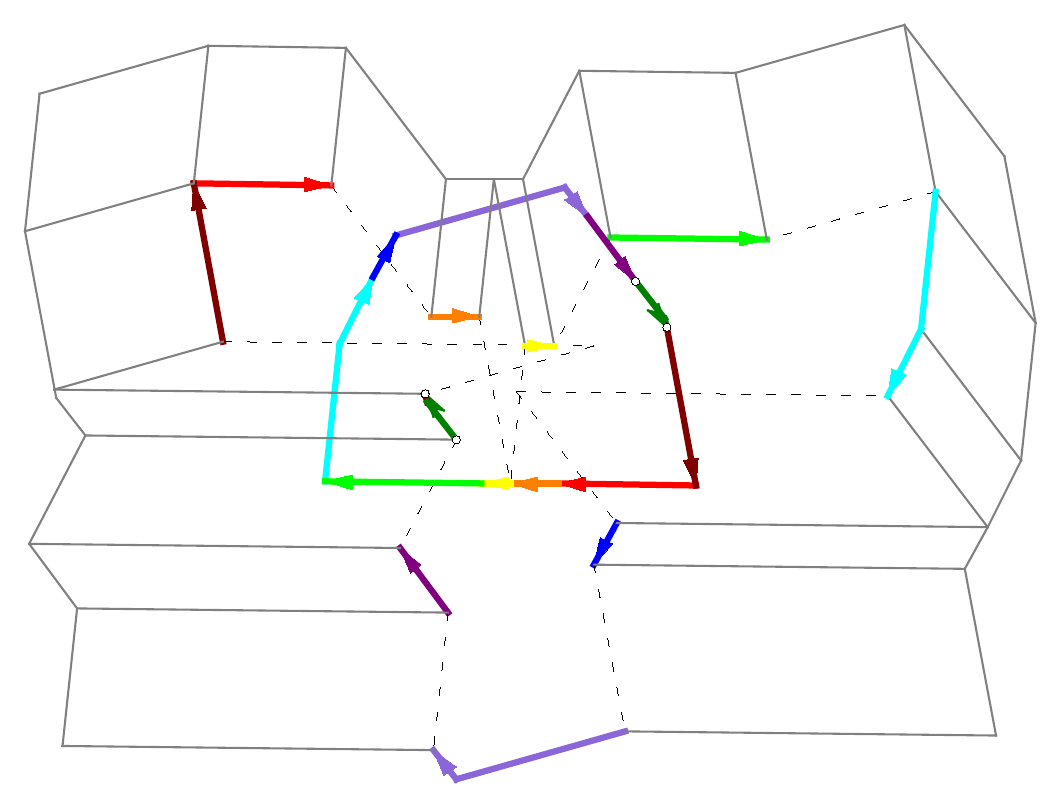}\caption{Function $g$}
\end{subfigure}
\caption{Picture (a) shows a convex polygon with six edges and draws all the blocks for this polygon.
        The frontier blocks are colored dark and light gray. Their bottom borders are pink.
    Picture (b) illustrates the frontier-pair-list, which indicates the frontier blocks.
        The table exhibits the chasing relation between the edges of $P$, where the solid circles indicate the edge pairs in the frontier-pair-list, and the hollow ones indicate other unit pairs in this list.
    Picture (c) illustrates a function $g$ from $\sigma P$ to $\partial P$ that will be introduced in the next subsection.} \label{fig:sigmaP-functionG}
\end{figure}

The following fact is easy to check: \emph{The starting point of the bottom border of $\block(u_{i+1},u'_{i+1})$ is the terminal point of the bottom border of $\block(u_i,u'_i)$, when $(u_i,{u'}_i),(u_{i+1},{u'}_{i+1})$ are adjacent pairs in the frontier-pair-list.} Trivial proof omitted.
We define the concatenation of the bottom borders (based on the order of $\mathsf{FPL}$) to be $\sigma P$.

\begin{note}
When $(e_{i-1},e_j)\notin \mathsf{FPL}$ and $(e_i,e_{j+1})\notin \mathsf{FPL}$, the bottom border of $\block(e_i,e_j)$ does not contain
     the common endpoint of its two lower borders ---  because these lower borders are open segments.
 So, in Figure~\ref{fig:sigmaP-functionG}, the lowermost corner of $\block(e_3,e_1)$, the leftmost corner of $\block(e_6,e_2)$,
   and the rightmost corner of $\block(e_2,e_5)$ are not contained in the bottom borders.
 Therefore, \textbf{none of these corner points are contained in $\sigma P$}.
This is important for understanding \textsc{Interleavity-of-$f$}.
 For example, if the lowermost corner of $\block(e_3,e_1)$ in Figure~\ref{fig:sigmaP-functionG}~(a) was counted as
    an intersecting point of $\partial P\cap \sigma P$, the \IP is wrong.
\end{note}

\subsection{Miscellaneous}
\newcommand{\LRF}{\textsc{Local-reversibility of $f$ }}
\newcommand{\LMF}{\textsc{Local-monotonicity of $f$ }}
In this subsection we state three lemmas (whose proofs are deferred to section~\ref{sect:lemmas-proof}) and introduce some notations.

Denote
 \begin{equation}\label{def:T(u,u')}\T(u,u')=\{(X_1,X_2,X_3)\in \T \mid X_3\in u,X_1\in u'\} = [u',\zeta(u,u'),u].\end{equation}

\begin{lemma}[\textsc{Local-reversibility} of $f$]\label{lemma:LRF}
Function $f$ is a bijection from $\T(u,u')$ to $\block(u,u')$.
\end{lemma}

\begin{definition}[$f^{-1}_{u,u'}()$ and $f^{-1,2}_{u,u'}()$]\label{def:f-local-reverse}
By Lemma~\ref{lemma:LRF}, there is a reverse function of $f$ on $\block(u,u')$, denoted by $f^{-1}_{u,u'}()$.
Moreover, notice that $f^{-1}_{u,u'}(X)$ is a tuple of three points; we denote the second point by $f^{-1,2}_{u,u'}(X)$.
\end{definition}

\subparagraph{Extend the domain of $f^{-1,2}_{u,u'}$.}
By Definition~\ref{def:f-local-reverse}, the domain of $f^{-1,2}_{u,u'}$ is $\block(u,u')$,
  which does not contain the lower border(s) of $\block(u,u')$ in general (unless both $u,u'$ are vertices).
    For convenience, we extend the domain of $f^{-1,2}_{u,u'}$ to include the lower order(s).
        Take any point $X$ from the lower border(s) of $\block(u,u')$.
        If $u=e_i,u'=e_j$, we define $f^{-1,2}_{u,u'}(X)=Z_i^j$.
        If $u=v_i,u'=e_j$, we define $f^{-1,2}_{u,u'}(X)=X'$, where $X'\in \zeta(v_i,e_j)$ is the reflection of $X$ with respect to $(v_i+v_{j+1})/2$.
        If $u=e_i,u'=v_j$, we define $f^{-1,2}_{u,u'}(X)=X'$, where $X'\in \zeta(e_i,v_j)$ is the reflection of $X$ with respect to $(v_i+v_j)/2$.
        If $u=v_i,u'=v_j$, the value of $f^{-1,2}_{u,u'}(X)$ is already defined.

\begin{definition}[Function $g$: $\sigma P\rightarrow \partial P$]\label{def:g}
For any point $X$ in $\sigma P$, assuming it comes from the bottom border of frontier block $\block(u,u')$, we define $g(X)=f^{-1,2}_{u,u'}(X)$.
  See Figure~\ref{fig:sigmaP-functionG}~(c) for an illustration.
\end{definition}

\begin{lemma}[\textsc{Local-monotonicity} of $f$]\label{lemma:LMF}
When a point $X$ travels along some boundary-portion of $P$ within $\block(u,u')$ (in clockwise),
$f^{-1,2}_{u,u'}(X)$ will move along $\partial P$ in clockwise (in the non-strict manner).
\end{lemma}

\begin{description}
\item[Interleaving.]
Assume $\mathcal{C}$ is an oriented yet not closed curve.
We say $\mathcal{C}$ \emph{interleaves} $\partial P$, if it is disjoint with $\partial P$ or the following holds.
Starting from their first intersection point (the first one that will be encountered when we travel along $\mathcal{C}$ in its positive direction),
  regardless of whether we travel along $\mathcal{C}$ in its positive direction or along $\partial P$ in clockwise,
we encounter the intersection points between $\mathcal{C}$ and $\partial P$ in the same order.
\end{description}

For example, in Figure~\ref{fig:definition-interleave}~(a), $\mathcal{C}$ interleaves $\partial P$.
In Figure~\ref{fig:definition-interleave}~(b), $\mathcal{C}$ does not interleave $\partial P$.

Also, recall the definition of \emph{interleaving between two oriented closed curves} in subsection~\ref{subsect:nestp}.

\begin{figure}[b]
\begin{minipage}[b]{.37\textwidth}
\centering\includegraphics[width=.9\textwidth]{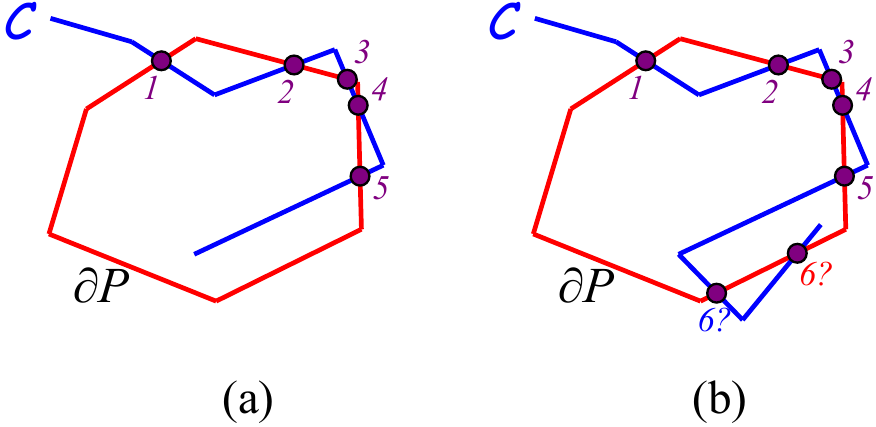}
\caption{interleaving and not interleaving.}\label{fig:definition-interleave}
\end{minipage}
\begin{minipage}[b]{.63\textwidth}
\centering\includegraphics[width=.9\textwidth]{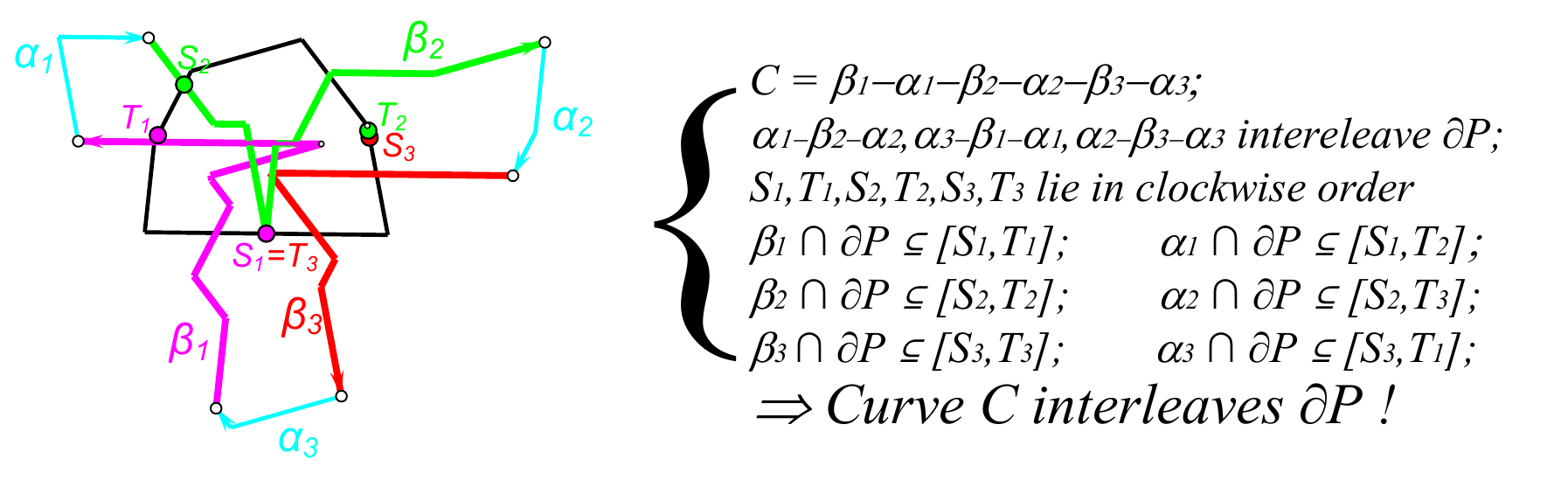}
\caption{Illustration of Lemma~\ref{lemma:interleave}. In this example, $q=3$.}\label{fig:outline_interleave}
\end{minipage}
\end{figure}

\begin{lemma}\label{lemma:interleave}
Assume a given oriented closed curve $\mathcal{C}$ is cut into $2q$ ($q\geq 3$) fragments: $\beta_1, \alpha_1, \ldots, \beta_q, \alpha_q$, such that
\begin{equation}\label{eqn:local-interleave}
\text{for $1\leq i\leq q$, the concatenation of $\alpha_{i-1},\beta_i,\alpha_i$ interleaves $\partial P$ (where $\alpha_0=\alpha_q$)}.
\end{equation}
Further assume that we can find $2q$ points $S_1,T_1,\ldots,S_q,T_q$ lying in clockwise order around $\partial P$ which ``delimitate'' the $2q$ fragments, by which we mean that
\begin{eqnarray}
\text{for $1\leq i\leq q$, the intersections between $\beta_i$ and $\partial P$ are contained in $[S_i \circlearrowright T_i]$, \text{ and}}  \label{eqn:delimitate-beta}\\
\text{for $1\leq i\leq q$, the intersections between $\alpha_i$ and $\partial P$ are contained in $[S_i \circlearrowright T_{i+1}]$.} \label{eqn:delimitate-alpha}
\end{eqnarray}
Then, the given curve $\mathcal{C}$ interleaves $\partial P$. See an example in Figure~\ref{fig:outline_interleave}.
\end{lemma}

\subsection{Introduction of the bounding-quadrants for the blocks}\label{subsect:bounding-quadrants}

\newcommand{\qd}{\mathsf{quad}}
\newcommand{\bp}{\langle\qd\rangle}
\begin{definition}[$\qd_i^j$ \& $\bp_i^j$]\label{def:quad}~ Assume $e_i\preceq e_j$. let $M=(v_i+v_{j+1})/2$.
\begin{itemize}
\item[Case1:] $e_i\prec e_j$; see Figure~\ref{fig:bounding-quadrants}~(a).
Make two rays at $M$,
  one with the opposite direction to $e_j$ whereas the other with the same direction as $e_i$.
Denote by $\qd_i^j$ the \textbf{open} region bounded by these two rays (which lies on the left of $\overrightarrow{v_iv_{j+1}}$),
  and denote by $\bp_i^j$ the intersection between $\qd_i^j$ and $\partial P$.
\item[Case2:] $e_i=e_j$; see Figure~\ref{fig:bounding-quadrants}~(b). Denote by $\qd_i^j$ the \textbf{open} half-plane that is bounded by the extended line of $e_i$ and lies to the left of $e_i$,
and denote by $\bp_i^j$ the midpoint of $e_i$.
\end{itemize}
Moreover, recall the backward and forward edges of units. For unit pair $(u,u')$ in which $u$ is chasing $u'$, define
\begin{equation}\label{eqn:def_bb}
\qd_u^{u'} := \qd~_{forw(u)}^{back(u')}; \qquad \bp_u^{u'}:= \bp_{forw(u)}^{back(u')}. \qquad (\text{Notice that $forw(u)\preceq back(u')$.})
\end{equation}
\end{definition}

\textbf{Note}: We regard the half-plane $\qd_i^i$ as a special quadrant whose apex lies at the midpoint of $e_i$; thus all the regions in $\{\qd_i^j \mid e_i\preceq e_j\}$ are quadrants in the plane.
The regions in $\{\bp_i^j \mid e_i\preceq e_j\}$ are boundary-portions of $P$.
It is worthwhile to point out that $\bp_i^j$ always contains $\qd_i^j\cap \partial P$ for $e_i\preceq e_j$ by the above definition.

\begin{figure}[h]
\centering \includegraphics[width=.75\textwidth]{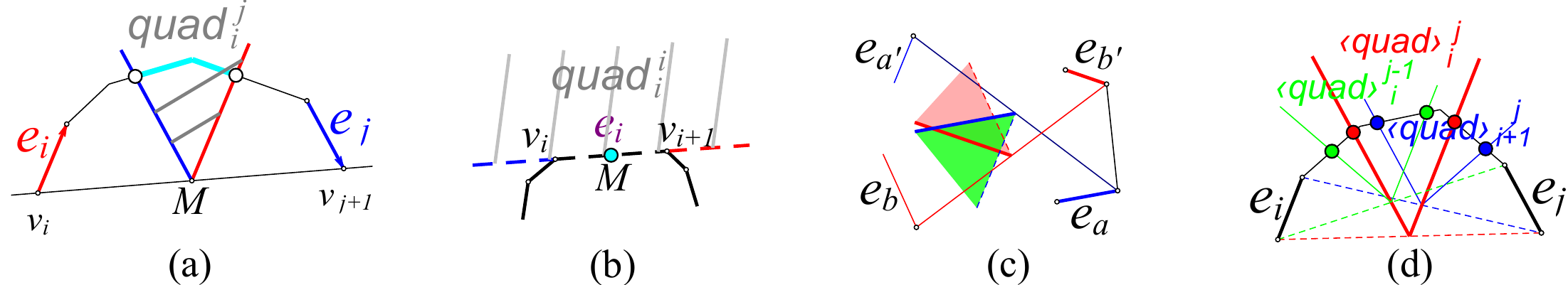}
\caption{Definition and two properties of the bounding-quadrants (Lemmas~\ref{lemma:br_peculiar} and \ref{lemma:br_monotone}).}\label{fig:bounding-quadrants}
\end{figure}

\begin{lemma}\label{lemma:block-in-quad}
For unit pair $(u,u')$ in which $u$ is chasing $u'$, region $\block(u,u')$ is contained in $\qd_u^{u'}$.
\end{lemma}

According to Lemma~\ref{lemma:block-in-quad}, we call $\qd_u^{u'}$ the \emph{bounding-quadrant} of $\block(u,u')$ henceforth.

\smallskip For the next three lemmas, recall that the borders of blocks are directional, as shown in Figure~\ref{fig:block-borders},
  and recall the relation $\leq_\rho$ defined on boundary-portion $\rho$.
    Also, recall $[v_i\circlearrowright v_{j+1}]$ is an inferior portion if and only if $e_i\preceq e_j$.

\begin{lemma}\label{lemma:border-monotone}
Suppose a point $X$ travels along some border of $\block(u,u')$ (in its positive direction) and suppose we stand both in $P$ and in the opposite quadrant of $\qd_u^{u'}$. Then, $X$ moves in clockwise order (\textbf{strictly}) around us.
\end{lemma}
\textbf{Note}: Since $\qd_u^{u'}$ is open, its opposite quadrant is also regarded as \textbf{open}; so it does not contain its boundary.

\begin{lemma}[\textbf{Peculiarity of $\qd$}]\label{lemma:br_peculiar}
For any $a,a',b,b'$ such that
$e_a\preceq e_{a'}, e_b\preceq e_{b'}\text{ and that } e_a,e_{a'},e_b,e_{b'}$ are not contained in any inferior portion,
  $\qd_a^{a'}\cap \qd_b^{b'}$ lies in the interior of $P$, as illustrated in Figure~\ref{fig:bounding-quadrants}~(c).
\end{lemma}

\begin{lemma}[\textbf{Monotonicity of $\bp$}]\label{lemma:br_monotone}
See Figure~\ref{fig:bounding-quadrants}~(d). Assume $e_i\prec e_j$. Denote $\rho=[v_i\circlearrowright v_{j+1}]$. Then,
\begin{equation*}
(\bp_i^{j-1}).s \leq_\rho (\bp_i^j).s \leq_\rho (\bp_{i+1}^j).s \text{~~and~~}
 (\bp_i^{j-1}).t \leq_\rho (\bp_i^j).t \leq_\rho (\bp_{i+1}^j).t,
\end{equation*}
where $\gamma.s$ and $\gamma.t$ denote the starting and terminal point of $\gamma$.
Moreover, for a list of $m$ boundary-portions $\bp_{u_1}^{u'_1},\ldots,\bp_{u_m}^{u'_m}$,
  their starting points lie in clockwise order and so do their terminal points,
  provided that
(1) $u_1,\ldots,u_m$ lie in clockwise order,
(2) $u'_1,\ldots,u'_m$ lie in clockwise order, and
(3) $u_k$ is chasing $u'_k$ for $1\leq k\leq m$.
\end{lemma}
\textbf{Note}: When some objects are said lying in clockwise order, some of them are allowed to coincide.\smallskip

Lemmas~\ref{lemma:br_peculiar} and \ref{lemma:br_monotone}  are proved in subsection~\ref{subsect:br_properties}.
Lemmas~\ref{lemma:block-in-quad} and \ref{lemma:border-monotone} are proved in subsection~\ref{subsect:br-relation-block}.

\section{Technique overview}\label{sect:techover}

In this section, we sketch our proofs of Theorem~\ref{thm:nestp} and Theorem~\ref{thm:nestp-location}.

\medskip Figure~\ref{fig:proof_outline}
  shows how the lemmas given in the last section are applied in our proof of Theorem~\ref{thm:nestp}
  and illustrates the interconnections between the first five properties of $f(\T)$.
The last property \SC is quite independent from these five properties and is not drawn in this figure.
As marked in the figure, \BD and \IP are strongly connected. Indeed, we will see their proofs are analogous.
The easy proofs of \RF and \SM will not only be sketched but completed in this section.

\tikzstyle{LEMMA} = [rectangle, minimum width=3cm, minimum height=.5cm,text centered, draw=black]
\tikzstyle{THEOREM} = [thick, rectangle, rounded corners, minimum width=3cm, minimum height=1cm, text centered, draw=black, fill=green!30]
\tikzstyle{KernelTHEOREM} = [thick, rectangle, rounded corners, minimum width=3cm, minimum height=1cm, text centered, draw=black,fill=blue!30]
\tikzstyle{arrow} = [black, thick,->,>=stealth]
\tikzstyle{connected} = [purple, very thick, -,>=stealth]
\begin{figure}[h]
\centering
\begin{tikzpicture}[node distance=1.6cm]
\small
\node (MonoBorder)  [LEMMA]{
    $\begin{array}{c}$Relations between blocks $\\$and bounding-quadrants$\\$(Lemmas~\ref{lemma:block-in-quad} and \ref{lemma:border-monotone})$\end{array}$};
\node (GlobalLemma) [LEMMA, below of=MonoBorder,yshift=-0.1cm]{
    $\begin{array}{c}$Peculiarity of the$\\$bounding-quadrants$\\$(Lemma~\ref{lemma:br_peculiar})$\end{array}$};
\node (MonoChi)     [LEMMA, below of=GlobalLemma,yshift=-0.1cm]{
    $\begin{array}{c}$Monotonicity of the $\\$bounding-quadrants$\\$(Lemma~\ref{lemma:br_monotone})$\end{array}$};
\node (InterleaveLemma)[LEMMA, below of=MonoChi,yshift=0.4cm]{Lemma~\ref{lemma:interleave}};
\node (BlocksFun)   [KernelTHEOREM, right of=MonoBorder,xshift=3.8cm,yshift=-1.1cm]{
    $\begin{array}{c}$\BD$\end{array}$};
\node (SigmaPMono)  [KernelTHEOREM, below of=BlocksFun,yshift=-0.3cm]{
    $\begin{array}{c}$\textsc{Interleavity-of-}$f\end{array}$};
\node (PReve_f)     [THEOREM, right of=BlocksFun,xshift=3.75cm,yshift=0.85cm]{$\RF$};
\node (SMono_f)     [THEOREM, below of=PReve_f,yshift=-0.2cm]{
    $\begin{array}{c}$\textsc{Monotonicity-of-}$f\end{array}$};
\node (SectorsMono) [THEOREM, below of = SMono_f,yshift=-0.2cm]{$\SM$};
\node (LRF) [LEMMA, above of=BlocksFun,yshift=0cm]{
    $\begin{array}{c}$\LRF$\\$(Lemma~\ref{lemma:LRF})$\end{array}$};
\node (LMF) [LEMMA, below of=SigmaPMono,yshift=0cm]{
    $\begin{array}{c}$\LMF$\\$(Lemma~\ref{lemma:LMF})$\end{array}$};
\draw [arrow](LRF) -- (PReve_f);
\draw [arrow](LMF) -- (SMono_f);
\draw [arrow](InterleaveLemma) -- (SigmaPMono);
\draw [arrow](GlobalLemma) -- (BlocksFun);
\draw [arrow](GlobalLemma) -- (SigmaPMono);
\draw [arrow](MonoChi) -- (SigmaPMono);
\draw [arrow](MonoBorder) -- (BlocksFun);
\draw [arrow](MonoBorder) -- (SigmaPMono);
\draw [arrow](BlocksFun) -- node {\textcolor[rgb]{1.00,0.50,0.00}{easy}} (PReve_f);
\draw [arrow](BlocksFun) -- (SMono_f);
\draw [arrow](SigmaPMono) -- (SMono_f);
\draw [arrow](SMono_f) -- node[anchor=mid] {\textcolor[rgb]{1.00,0.50,0.00}{easy}} (SectorsMono);
\draw [connected] (BlocksFun) -- node[anchor=mid] {strongly connected} (SigmaPMono);
\end{tikzpicture}
\setcaptionwidth{.85\textwidth}
\caption{Interconnections between the five properties (not including \SC).}\label{fig:proof_outline}
\end{figure}

\subsection*{Technique overview for proving the first five properties.}

\begin{definition}\label{def:extremal}
Denote by $e_i\nprec e_j$ if $e_i$ is not chasing $e_j$.
Edge pair $(e_c,e_{c'})$ is \emph{extremal}, if $e_c\prec e_{c'}$ and
the inferior portion $[v_c\circlearrowright v_{c'+1}]$ is not contained in any other inferior portions.
Equivalently, if $e_c\prec e_{c'},e_{c-1}\nprec e_{c'},e_c\nprec e_{c'+1}$.
For example, in Figure~\ref{fig:sigmaP-functionG}~(b), the edge pairs indicated by the red solid circles are extremal.
\end{definition}

For any extremal pair $(e_c,e_{c'})$, denote
\[\Delta(c,c'):=\left\{(u,u') \left | \begin{array}{c}\text{unit $u$ is chasing $u'$, and } forw(u),back(u')\in \{e_c,e_{c+1},\ldots,e_{c'}\}\end{array} \right. \right\}.\]

\begin{proof}[Proof of \BD (sketch)]
We regard $(\block(u,u'),\block(v,v'))$ as a \textbf{\emph{local pair}} if
the four edges $forw(u),back(u'),forw(v),back(v')$ are contained in an inferior portion;
and as a \textbf{\emph{global pair}} otherwise.

It suffices to prove \BD if we can prove the following results:
\begin{itemize}
\item[(I)] \emph{For a global pair $\block(u,u'),\block(v,v')$, their intersection lies in the interior of $P$.}
\item[(II)] \emph{For a local pair $\block(u,u'),\block(v,v')$, their intersection is always empty.}
\end{itemize}

We can employ the bounding-quadrants to prove (I) as follows.

Because these blocks are global, $forw(u)$, $back(u')$, $forw(v)$, $back(v')$ not contained in any inferior-portion.
This means $\qd_{forw(u)}^{back(u')} \cap \qd_{forw(v)}^{back(v')}$ lies in the interior of $P$, according to Lemma~\ref{lemma:br_peculiar}.
According to Lemma~\ref{lemma:block-in-quad},
  $\block(u,u')\cap \block(v,v')\subset \qd_u^{u'}\cap \qd_v^{v'}=\qd_{forw(u)}^{back(u')} \cap \qd_{forw(v)}^{back(v')}.$
Together, we get (I).

Be aware of the following equivalent form of (II): \emph{For every extremal pair $(e_c,e_{c'})$, the blocks in $\{\block(u,u')\mid (u,u')\in \Delta(c,c')\}$ are pairwise-disjoint}. See Figure~\ref{fig:sketch-key-properties}~(a) for an illustration of such blocks.
Surprisingly, the monotonicity of the borders stated in Lemma~\ref{lemma:border-monotone} somehow implies this result.
Details are given in section~\ref{sect:fT-major}.
\end{proof}

\begin{figure}[h]
\begin{minipage}[b]{.32\textwidth}
\includegraphics[width=.98\textwidth]{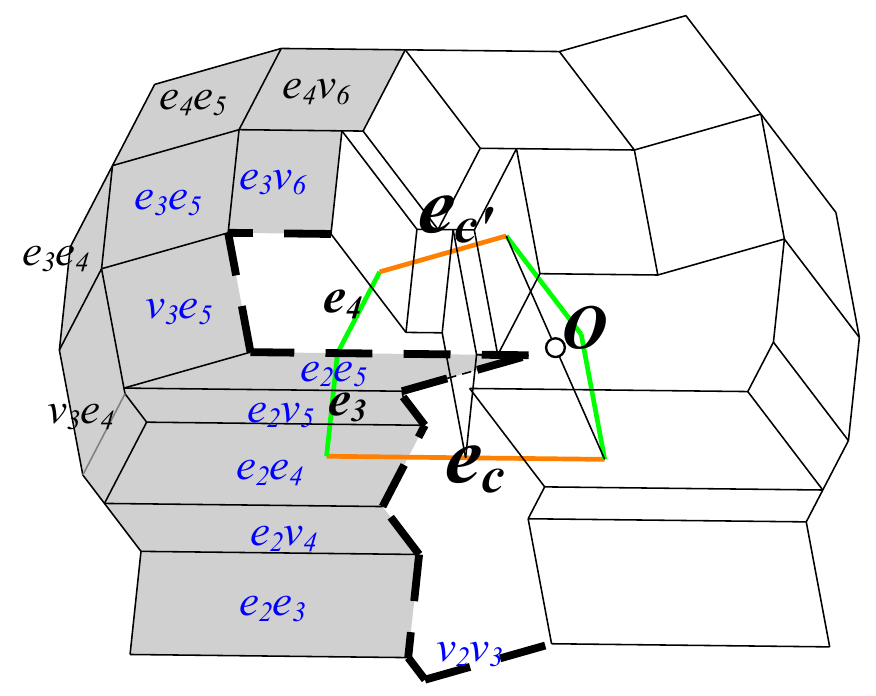}
\subcaption{Sketch 1.}
\end{minipage}
\begin{minipage}[b]{.32\textwidth}
\includegraphics[width=.98\textwidth]{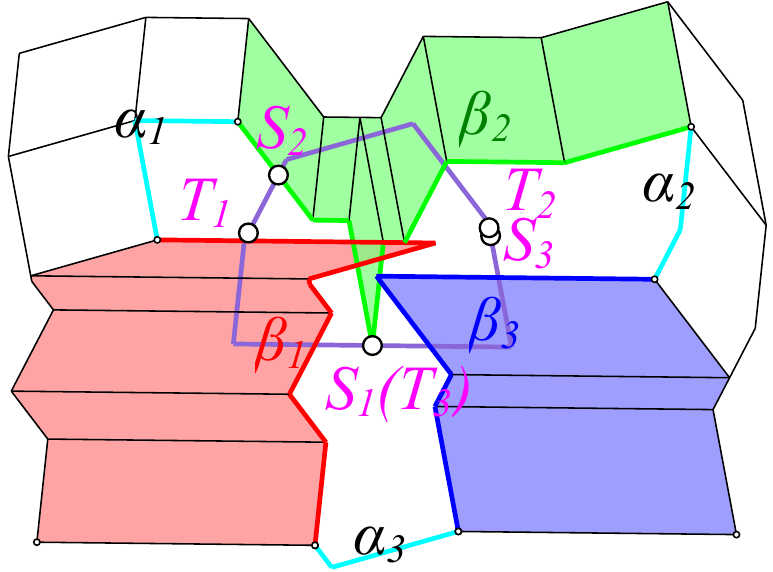}
\subcaption{Sketch 2.}
\end{minipage}
\begin{minipage}[b]{.32\textwidth}
\includegraphics[width=.98\textwidth]{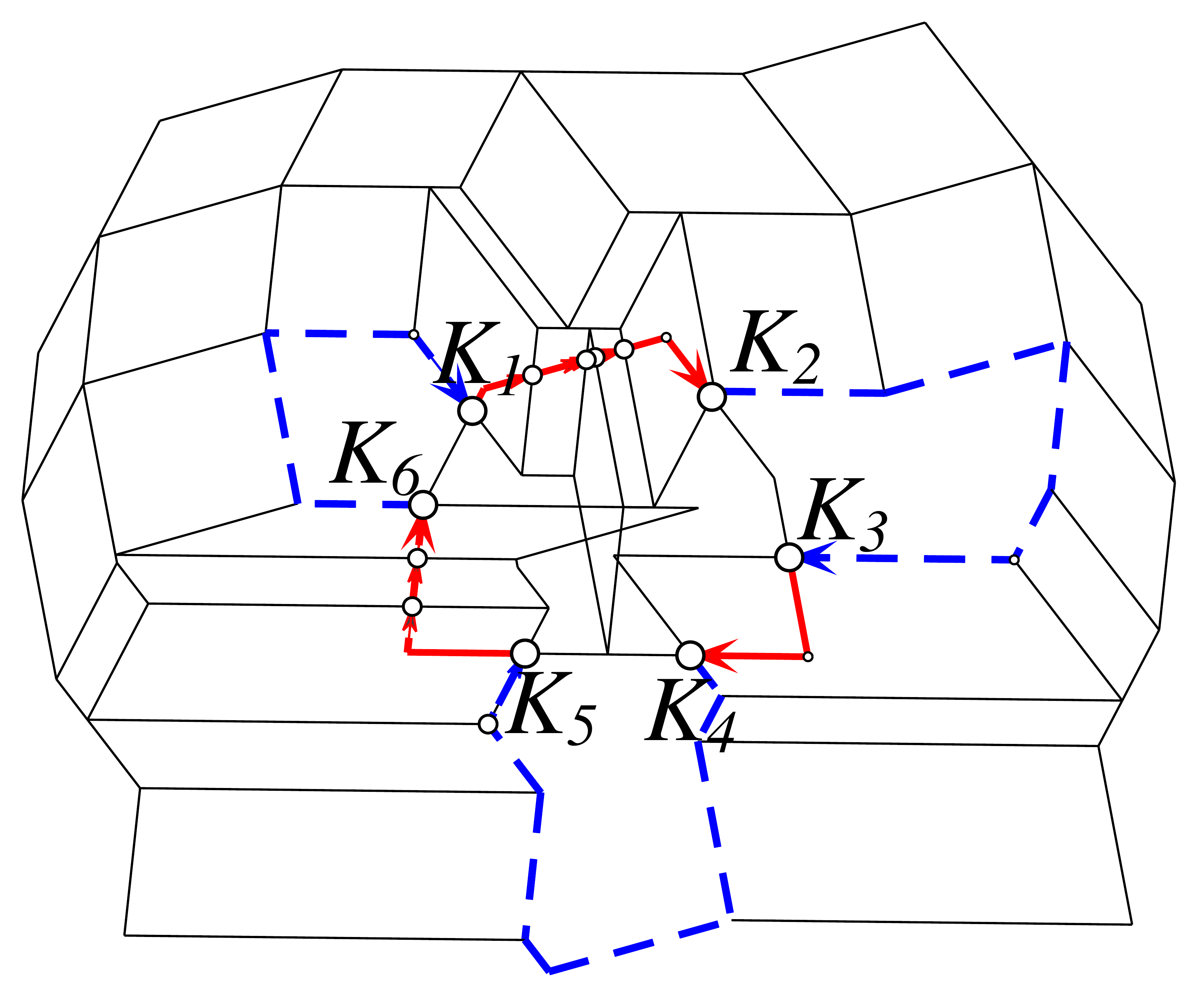}
\subcaption{Sketch 3.}
\end{minipage}
\caption{Illustrations of the proof of the properties of $f(\T)$.}\label{fig:sketch-key-properties}
\end{figure}

\begin{proof}[Proof of \IP (sketch )]
In general, we use Lemma~\ref{lemma:interleave} to prove that $\sigma P$ interleaves $\partial P$.

First, we cut $\sigma P$ into $2q$ fragments $\beta_1$, $\alpha_1$, $\ldots$, $\beta_q$, $\alpha_q$, where $q$ denotes the number of extremal pairs.
For every extremal pair $(e_c,e_{c'})$,
   denote by $\sigma(c,c')$ the concatenation of the bottom borders of the frontier blocks
     in $\{\block(u,u')\mid (u,u')\in\Delta(c,c')\}$; as shown by the dashed polygonal curve in Figure~\ref{fig:sketch-key-properties}~(a).
Denote by $(e_{c_1},e_{c'_1}), \ldots, (e_{c_q},e_{c'_q})$ an enumeration of the extremal pairs (in the order given by the frontier-pair-list).
  We define $\alpha_i$ to be $\sigma(c_i,{c'}_i)\cap \sigma(c_{i+1},{c'}_{i+1})$;
         and $\beta_i$ to be $\sigma(c_i,c'_i)-\alpha_{i-1}-\alpha_{i}$.
           See Figure~\ref{fig:sketch-key-properties}~(b) for an illustration.

As an easy corollary of the monotonicity of the borders stated in Lemma~\ref{lemma:border-monotone}, $\sigma(c_i,{c'}_i)$ interleaves $\partial P$.
In other words, the first condition (\ref{eqn:local-interleave}) listed in Lemma~\ref{lemma:interleave} holds, since
   the concatenation of $\alpha_{i-1},\beta_i,\alpha_i$ equals $\sigma(c_i,{c'}_i)$.

We then select $2q$ points $S_1,T_1,\ldots,S_q,T_q$ from $\partial P$.
       Denote by $\block(u,u')$ and $\block(v,v')$ the first and last frontier blocks whose bottom borders contribute to $\beta_i$.
       We define the starting point of $\bp_u^{u'}$ and the terminal point of $\bp_v^{v'}$ to be $S_i$ and $T_i$, respectively.
Applying the properties of the bounding-quadrants (Lemmas~\ref{lemma:br_peculiar} and \ref{lemma:br_monotone}),
     it can be proved that $S_1,T_1,\ldots,S_q,T_q$ lie in clockwise order around $\partial P$,
        and moreover, the conditions (\ref{eqn:delimitate-beta}) and (\ref{eqn:delimitate-alpha}) listed in Lemma~\ref{lemma:interleave} hold,
    namely, $\beta_i\cap \partial P\subset [S_i \circlearrowright T_i]$ and $\alpha_i\cap \partial P\subset [S_i \circlearrowright T_{i+1}]$.

At last, apply Lemma~\ref{lemma:interleave} and we obtain that $\sigma P$ interleaves $\partial P$. See the details in section~\ref{sect:fT-major}.
\end{proof}

As a summary, the \IP follows from two facts:
\emph{(I) locally, every local fraction of $\sigma P$ (i.e. $\sigma(c_i,c'_i)$) interleaves $\partial P$}, and
\emph{(II) globally, all these local fractions are well-scattered around $\partial P$.}\medskip

The \RF follows from the \BD and \LRF.

\begin{proof}[Proof of \RF]
Recall $\T(u,u')$ in (\ref{def:T(u,u')}).
For each unit pair $(u,u')$ in which $u$ is chasing $u'$, we call $\T(u,u')$ a \emph{component} of $\T$.
Notice that each element of $\T$ belongs to exactly one component.

Now, consider any two elements in subset $\T^*$.
  If they belong to the same component, their images under function $f$ are distinct according to the \LRF (Lemma~\ref{lemma:LRF}).
  If they belong to distinct components, their images under $f$ do not coincide, since otherwise
     there would be two distinct blocks with an intersection on the boundary of $P$, which contradicts the \textsc{Block-disjointness}.
  Therefore, $f$ is a bijection from $\T^*$ to $f(\T^*)$.
\end{proof}

\begin{proof}[Proof of \MF (sketch)]
See Figure~\ref{fig:sketch-key-properties}~(c).
Let $K_1,\ldots,K_m$ denote all the intersection points between $\sigma P$ and $\partial P$,
  where $K_1,\ldots,K_m$ lie in clockwise order around $\partial P$.
These intersection points divide $\partial P$ into $m$ boundary-portions, each of which is called a \emph{$K$-portion}.
We state three crucial observations:
\begin{enumerate}
\item[(i)] The outer boundary of $f(\T)$ is a simple closed curve whose interior contains $P$.
\item[(ii)] Points $f^{-1}_2(K_1),\ldots,f^{-1}_2(K_m)$ lie in clockwise order around $\partial P$.
\item[(iii)] Function $f^{-1}_2()$ is monotone on any $K$-portion that lies in $f(\T)$.
\end{enumerate}

The proof of observation~(i) is trivial.
   Observation~(ii) follows from three facts:
     (a) It always holds that $g(K_i)=f^{-1}_2(K_i)$ (recall $g$ in Definition~\ref{def:g});
     (b) $g$ is a circularly monotone function from $\sigma P$ to $\partial P$; and
     (c) Points $K_1,\ldots,K_m$ lie in clockwise order around $\sigma P$.
   Facts~(a) and (b) are due to the definition of $g$.
   Fact (c) follows from the \IP and the assumption that $K_1,\ldots,K_m$ lie in clockwise order around $\partial P$.
Observation~(iii) follows from \LMF (Lemma~\ref{lemma:LMF}), which states a similar property about $f^{-1,2}_{u,u'}()$.

Combining the \IP with observation (i), a $K$-portion either lies in or lies outside $f(\T)$.
Now, suppose a point $X$ travels in $f(\T)\cap \partial P$ in clockwise.
    Observation (iii) assures that $f^{-1}_2(X)$ is monotone inside each $K$-portion.
    Observation (ii) assures that it is monotone between the $K$-portions. Details are given in section~\ref{sect:othertwo}.
\end{proof}

The \SM follows from the \MF immediately.

\begin{proof}[Proofs of \SM]
For any unit $w$, we have
\[\begin{split}
\sector(w)\cap \partial P &= \left\{f(X_1,X_2,X_3)\mid (X_1,X_2,X_3)\in \T^*, X_2\in w\right\}\\
&= \left\{Y\in f(\T)\cap \partial P\mid  f^{-1}_2(Y)\in w\right\} =\left\{Y\in f(\T)\cap \partial P\mid \unit(f^{-1}_2(Y))=w\right\}.
\end{split}\]

Consider the points in $f(\T)\cap \partial P$. Clearly, $\unit(f^{-1}_2())$ is a function on these points that maps them to the $2n$ units of $P$.
Following from the \MF, function $\unit(f^{-1}_2())$ is circularly monotone on these points.
So, $\unit(f^{-1}_2())$ implicitly divides $f(\T)\cap \partial P$ into $2n$ parts which are pairwise-disjoint and lie in clockwise order around $\partial P$.
Moreover, according to the above equation, these $2n$ parts are precisely $\sector(v_1)\cap \partial P$, $\sector(e_1)\cap \partial P$, \ldots, $\sector(v_n)\cap \partial P$, $\sector(e_n)\cap \partial P$.
Therefore, we obtain the \SM.
\end{proof}

\subsection*{Technique overview for defining $\LVS,\RVS$ and for proving the \SC}\label{subsect:sector-continue}

Assume $V$ is a fixed vertex of $P$. In the following, we depict the general shape of $\sector(V)$ and define explicitly its two boundaries $\LVS,\RVS$ mentioned in section~\ref{sect:introduction}, and then sketch how do we prove the \textsc{Sector-continuity}.

\smallskip Recall  $u\oplus u'=\{(X+X')/2\mid X\in u, X'\in u'\}$. The following formula of $\sector(V)$ follows from (\ref{def:T}) and (\ref{def:sector}):
\begin{equation}\label{eqn:sector-zeta}
\sector(V)= \text{2-scaling of } \bigg({\bigcup}_{ \text{$u$ is chasing $u'$, and $\zeta(u,u')$ contains $V$}} u\oplus u'\bigg)\text{ with respect to }V.
\end{equation}

\newcommand{\LV}{\mathcal{L}_V}
\newcommand{\RV}{\mathcal{R}_V}
\newcommand{\MID}{\mathsf{mid}_V}
\newcommand{\MIDSCALE}{\mathsf{mid}^\star_V}

Based on (\ref{eqn:sector-zeta}), in order to know the shape of $\sector(V)$, it is important to study the structure of $\Lambda_V$, where
\begin{equation}\label{def:Lambda_V}
    \Lambda_V:=\{(u,u')\mid u\text{ is chasing }u'\text{, and }\zeta(u,u')\text{ contains }V\}.
\end{equation}

In the following we first introduce some symbols or notions and then introduce a structural property of $\Lambda_V$.
\begin{description}
\item [Delimiting edges $e_{s_V}$ and $e_{t_V}$.]
        Two particular edges of $P$ will be specified as $e_{s_V}$ and $e_{t_V}$, for each vertex $V$.
        Note that the exact values of $s_V$ and $t_V$ however will not be defined in this overview section,
           because the definition is involved and needs additional notations to explain, as we will see in section~\ref{sect:othertwo};
           and also because there is no need to know the exact values $s_V$ and $t_V$
                    for understanding the following notations and facts within this section.
                    Indeed, only in the proofs of the following facts will we look at the exact values of $s_V,t_V$.

        We abbreviate $s_V,t_V$ as $s,t$ when $V$ is clear, and call $e_s,e_t$ the \emph{delimiting edges}. Two crucial facts are:
\begin{itemize}
\item[(i)] \emph{$e_s\preceq e_t$. Moreover, the inferior portion $[v_s\circlearrowright v_{t+1}]$ does not contain $V$.}
\item[(ii)] \emph{If $(u,u')\in \Lambda_V$, units $u,u'$ both lie in $[v_s\circlearrowright v_{t+1}]$.} (This suggests the name ``delimiting''.)
\end{itemize}
\item [Incident relation.] Each unit has two \emph{incident units}: $e_i$ has $v_i,v_{i+1}$; whereas $v_i$ has $e_{i-1},e_{i}$. \\
        In particular, pay attention that $e_i,e_{i+1}$ are \textbf{not} regarded as incident.
\item [$\Delta_V$ -- A superset of $\Lambda_V$.]
   Notice that there is an order, denoted by `$<$', among the units in $[v_s\circlearrowright v_{t+1}]$ conforming to the clockwise order.
    We denote $\Delta_V= \left\{(u,u') \mid  u,u'\in [v_s\circlearrowright v_{t+1}], u<u', \text{and they are not incident} \right\}.$\\
    \textbf{Note}: By fact~(ii), $\Delta_V$ is a superset of $\Lambda_V$.\\
    \textbf{Note}: By fact~(i), $e_s\preceq e_t$, hence those regions $u\oplus u' \mid (u,u')\in \Delta_V$ are pairwise-disjoint. See Figure~\ref{fig:mid_V}~(a).
\item[Roads \& routes.]
    See Figure~\ref{fig:mid_V}~(b).
    For any $(e_i,v_j)\in \Delta_V$, we assume segment $e_i\oplus v_j$ has \emph{the same} direction as $e_i$;
    for any $(v_i,e_j)\in \Delta_V$, we assume segment $v_i\oplus e_j$ has \emph{the opposite} direction to $e_j$; and we refer to these segments as \emph{roads}.
    See Figure~\ref{fig:mid_V}~(c). By concatenating several roads, we may obtain directional polygonal curves starting from $(v_s+v_{t+1})/2$ to some point in $[v_s \circlearrowright v_{t+1}]$;
        such curves are called \emph{routes}. \smallskip
\item[$\LV,\RV$.] Denote $\rho=[v_{t+1} \circlearrowright v_s]$. For any edge pair $(e_i,e_j)$ in $\Delta_V$, we mark region $e_i\oplus e_j$ by `-' if $Z_i^j <_\rho V$;
`+' if $V <_\rho Z_i^j$; and `0' if $V= Z_i^j$.
    By the bi-monotonicity of the $Z$-points (Fact~\ref{fact:Z_bi-monotonicity}),
   there exists a unique route, $\LV$, which separates the regions marked by `-' from the regions marked by `+/0'.
  Similarly, there exists a unique route, $\RV$, which separates the regions marked by `+' from the regions marked by `-/0'.\smallskip
\item[A region $\MID$.] By the definitions of $\LV,\RV$ and bi-monotonicity of the $Z$-points (Fact~\ref{fact:Z_bi-monotonicity}),
  the region bounded by $\LV,\RV$ and $\partial P$ is well-defined (see Figure~\ref{fig:mid_V}~(f)); we denote it by $\MID$.\\
  \textbf{Note}: To be more clear, $\MID$ contains its two boundaries $\LV$ and $\RV$.\\
  \textbf{Note}: For the special case $s=t$, we will have $\LV=\RV=\MID=v_s\oplus v_{t+1}$.
\item [$\LVS,\RVS,$ \& $\MIDSCALE$.] Denote the 2-scalings of $\MID,\LV,\RV$ with respect to $V$ by $\MIDSCALE,\LVS,\RVS$ respectively.
\end{description}

\begin{figure}[t]
\centering\includegraphics[width=.7\textwidth]{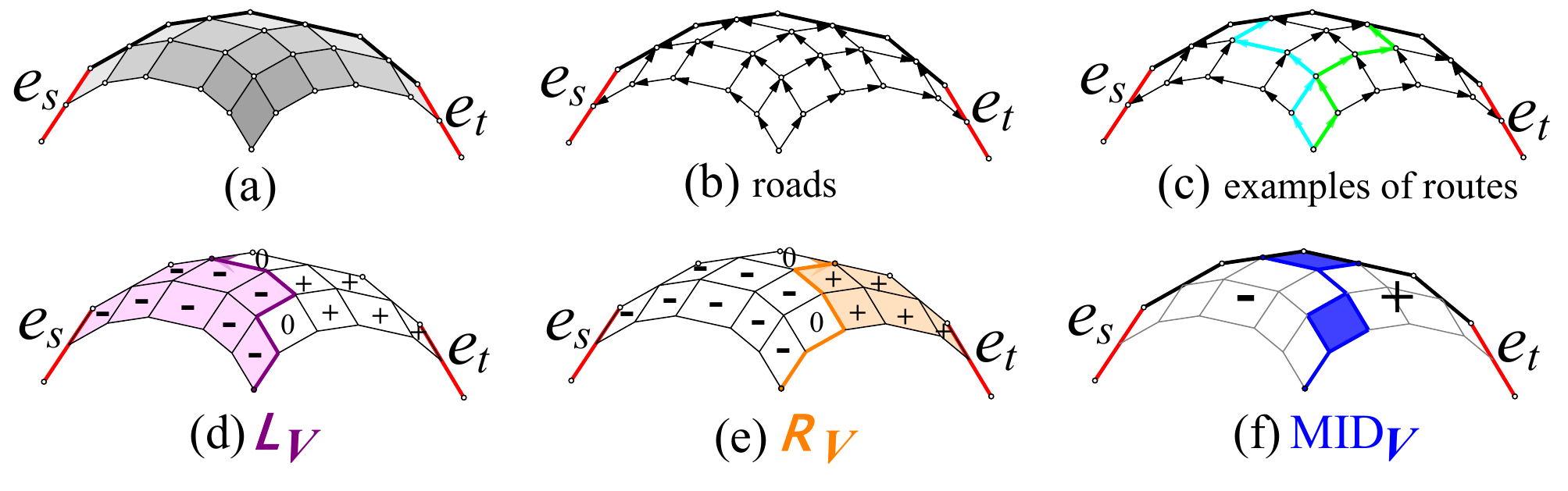}
\caption{Illustration of the definition of curves $\LV,\RV$ and region $\MID$.}\label{fig:mid_V}
\end{figure}

We now describe the aforementioned structural property of $\Lambda_V$.
\begin{enumerate}
\item[(iii)] Because $\MID$ is the union of several (disjoint) regions $u\oplus u'$ for which $(u,u')\in \Delta_V$,
                it defines a subset of $\Delta_V$. To be clear, this subset contains $(u,u')$ if and only if $u\oplus u'\subseteq \MID$.
        Now, remove from this subset those unit pairs $(u,u')$ in which $u$ is not chasing $u'$,
          the remaining subset (of $\Delta_V$) is exactly $\Lambda_V$.
\end{enumerate}

The proof of fact~(iii) again requires the definitions of $e_s$ and $e_t$ and is thus omitted in this overview section.

\medskip Notice that $\sector(V)= \text{2-scaling of } ({\bigcup}_{(u,u')\in \Lambda_V} u\oplus u')\text{ with respect to }V$, by combining (\ref{eqn:sector-zeta}) with (\ref{def:Lambda_V}).
Now, connecting this equation of $\sector(V)$ with the structural property of $\Lambda_V$ given in fact~(iii), we get:
\begin{enumerate}
\item[(iv)] Region $\sector(V)$ equals the 2-scaling of $(\MID-\epsilon_V)$ with respect to $V$,
        where $\epsilon_V$ denotes the union of those $u\oplus u'$ for which $(u,u')\in \Delta_V$ and $u$ is not chasing $u'$.
\end{enumerate}

We remark that $\epsilon_V$ and $\MID$ are such regions whose shapes are easy to understand.
   Therefore, based on fact~(iv) it is easy to understand what the shape of $\sector(V)$ could be.\medskip

\subparagraph{Two boundaries of $\sector(V)$ and the proof of \SC (overview).}
Following fact~(iv), we can further prove that \emph{the closed set of $\sector(V)$ equals $\MIDSCALE$}, which implies
that $\MIDSCALE$ and $\sector(V)$ have the same boundaries.
Further since $\LVS,\RVS$ are boundaries of $\MIDSCALE$, they are boundaries of $\sector(V)$ indeed as we mentioned.
Moreover, we observe that either of $\LVS,\RVS$ has at most one intersection with $\partial P$ (see Figure~\ref{fig:sectors-nestp-sigma}~(a)); so $\MIDSCALE=[\LVS\cap \partial P\circlearrowright \RVS\cap \partial P]$.
Following this equation and the fact that $\MIDSCALE$ is the closed set of $\sector(V)$, we see $\sector(V)\cap \partial P$ is a boundary-portion from $\LVS\cap \partial P$ to $\RVS\cap \partial P$. Thus \SC holds.

\begin{remark}
Although the particular values of $s_V,t_V$ are of no use for understanding the contents sketched in this overview section,
    defining $s_V,t_V$ is the most nontrivial step in the entire proof of the \textsc{Sector-continuity}.
\end{remark}

\subparagraph*{A slightly modification on the directions of the roads.}
Recall that $\Nest(P)$ is defined as the union of the boundaries of all blocks and $\mathcal{L}^\star_{v_1},\ldots,\mathcal{L}^\star_{v_n}$ and $\mathcal{R}^\star_{v_1},\ldots,\mathcal{R}^\star_{v_n}$, where $\mathcal{L}^\star_V,\mathcal{R}^\star_V$ are 2-scalings of $\LV,\RV$ with respect to $V$,
  and where $\LV,\RV$ consist of several roads.
  Previously we assume road $e_i\oplus v_j$ has the same direction as $e_i$, whereas road $v_i\oplus e_j$ has the \textbf{opposite} direction to $e_j$.
  This is convenient for proving the \textsc{Sector-continuity}.
  However, when defining $\Nest(P)$, it is more reasonable to assume that $v_i\oplus e_j$ has the \textbf{same} direction as $e_j$.

\subparagraph*{Open problem.} Can we find a structure similar to $\Nest(P)$ in a higher dimensional space?\medskip

\subsection*{Technique overview for proving Theorem~\ref{thm:nestp-location}}\label{subsect:preprocess}

Theorem~2 states that we can answer in (amortized) $O(\log ^2n)$ time \emph{Sector-intersect-Units}, \emph{Vertex-in-Sector}, and \emph{Vertex-in-Block},
which ask the interval of units that intersect $\sector(V)$, or the sector or the block that contains $V$.

\subparagraph*{1. Compute the endpoints of $\sector(V)\cap \partial P$ for a given vertex $V$.}
According to the proof sketch of \SC above, computing the endpoints of $\sector(V)\cap \partial P$  reduces to
        computing $\LVS\cap \partial P$ and $\RVS\cap \partial P$.
Recall that either of $\LVS,\RVS$ is a polygonal curve and has at most one intersection with $\partial P$.
Moreover, it is not difficult to generate an arbitrary edge of $\LVS$ (or $\RVS$), and in $O(\log n)$ time decide whether it lies inside, outside, or intersects $P$.
  Therefore, we can compute $\LVS\cap \partial P$ (or $\RVS\cap \partial P$) by a binary search in $O(\log^2n)$ time.

\subparagraph*{2. Compute the units that intersect $\sector(V)$ for a given vertex $V$.}
Let $u_L,u_R$ respectively be the unit containing $\LVS\cap \partial P$ and the unit containing $\RVS\cap \partial P$, which can be computed when we compute $\LVS\cap \partial P$ and $\RVS\cap \partial P$.
By \textsc{Sector-continuity}, the units that intersect $\sector(V)$ are (roughly) the units from $u_L$ to $u_R$ in clockwise order. (This holds almost in every case. A degenerate case is discussed in subsection~\ref{subsect:alg_S}).

\subparagraph*{3. Compute the respective sectors that contain each vertex.}
We build two groups of \emph{event-points}:
the points in $\{\LVS\cap \partial P,\RVS\cap \partial P\mid V\text{ is a vertex of }P\}$, and the intersections between $\sigma P$ and $\partial P$.
Two tags, \emph{current-tag} and \emph{future-tag}, are assigned to each event-point;
  the former one indicates the sector containing this event-point;
  the latter one indicates the sector containing the boundary-portion starting from this event-point and terminating at its clockwise next event-point.
By a \textbf{sweeping} around $\partial P$, we find the sector containing each vertex utilizing the tags.

\newcommand{\cell}{\mathsf{cell}}

\subparagraph*{4. Compute the block containing $V$ for a given vertex $V$.}\label{sect:alg_B} This is the most nontrivial part in our algorithm.

Given $V$ which lies in $\sector(w)$, where $w$ is given, we shall find $(u^*_1,u^*_2)$ so that
        $\block(u^*_1,u^*_2)$ contains $V$.

\medskip Similarly as in the proof of \SC, we need to introduce a few notations first.
(For \textsc{Sector-continuity}, we search $(u,u')$ such that $V\in \zeta(u,u')$;
  and here we search $(u,u')$ such that $V\in \block(u,u')$.)

\medskip Let $[X\circlearrowright X')$ denote $[X\circlearrowright X']-\{X'\}$ and let $(X\circlearrowright X']$ denote $[X\circlearrowright X']-\{X\}$.

\begin{description}
\item [Delimiting edges $e_{p_V},e_{q_V}$.]
Two particular edges of $P$ will be specified as $e_{p_V}$ and $e_{q_V}$, for each vertex $V$. Note
that the exact values of $p_V$ and $q_V$ however will not be defined in this overview section, because the definition is
involved and needs additional notations to explain, as we will see in subsection~\ref{subsect:alg_B-pq}; and also because there is no
need to know the exact values $p_V$ and $q_V$ for understanding the following notations and facts within this section.
Indeed, only in the proofs of the following facts will we look at the exact values of $p_V,q_V$.

We abbreviate $p_V,q_V$ as $p,q$ when $V$ is clear, and call $e_p,e_q$ the \emph{delimiting edges}. Two crucial facts are:
\begin{enumerate}
\item[(i)] \emph{$e_p\prec e_{q}$. Moreover, the boundary-portion $(v_p\circlearrowright v_{q+1})$ contains $V$.}
\item[(ii)] \emph{$u^*_1 \in [v_p\circlearrowright V)$, and $u^*_2\in (V\circlearrowright v_{q+1}]$.} (This suggests the name ``delimiting''.)
 \end{enumerate}

\item[Alive pairs.] For any unit pair $(u,u')$ such that $u$ is chasing $u'$, we regard it as \emph{alive} if it satisfies the conditions given in fact~(ii).
        In other words, it is alive if $u\in [v_p\circlearrowright V)$ and $u'\in (V\circlearrowright v_{q+1}]$.

\item[Active pairs.] For any alive pair $(u,u')$, we regard it as \emph{active} if $\zeta(u,u')\text{ intersects }w$.
    In particular, $(u^*_1,u^*_2)$ is active due to fact~(ii) and $V\in \sector(w)$.
     This means we only need to search $(u^*_1,u^*_2)$ among the active pairs.

    Figure~\ref{fig:cells-layers}~(a) draws all the blocks of the alive pairs, where the colored blocks are blocks of the active pairs.

\end{description}

\begin{figure}[h]
 \begin{minipage}[b]{.4\textwidth}
  \centering\includegraphics[width=.67\textwidth]{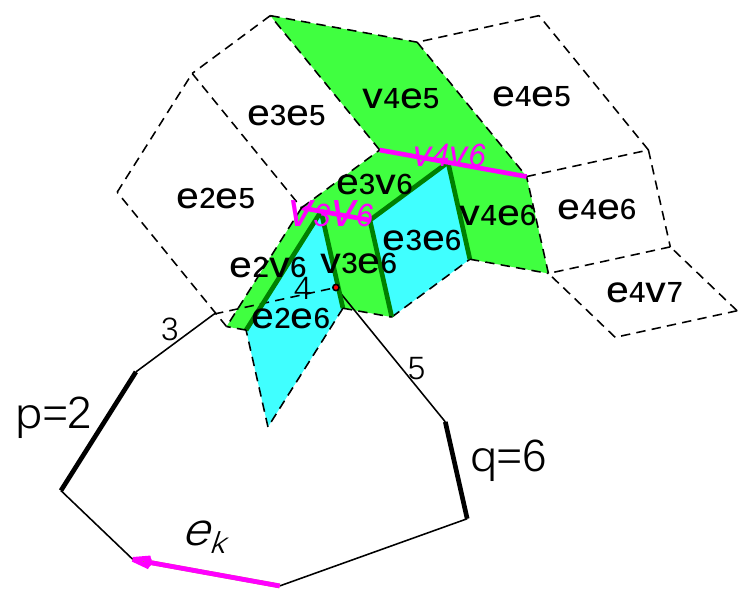}
  \end{minipage}
 \begin{minipage}[b]{.6\textwidth}
  \centering\includegraphics[width=.78\textwidth]{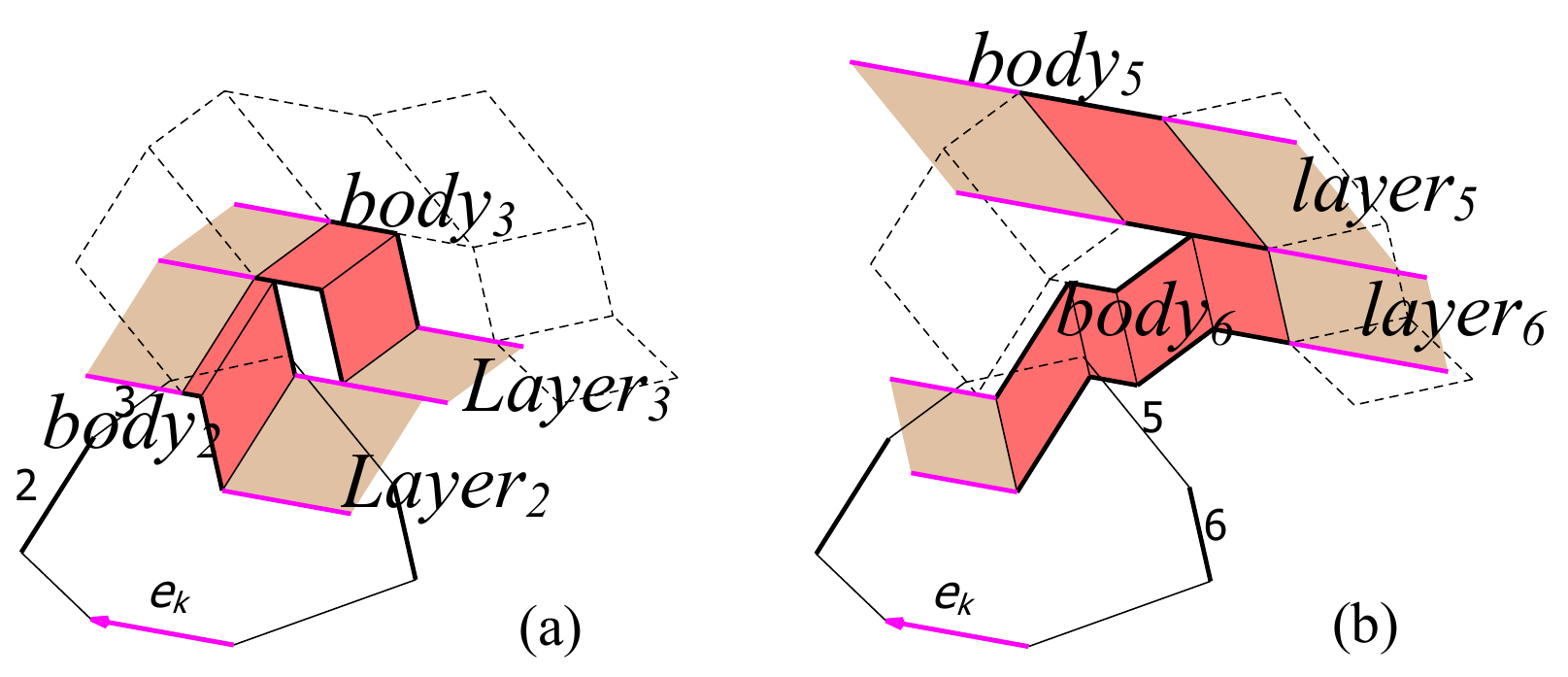}
 \end{minipage}
 \caption{Illustration of cells and layers.}\label{fig:cells-layers}
\end{figure}

Next, we have to discuss two cases depending on whether $w$ is an edge or a vertex. The edge case is more typical and the vertex case can be regarded as a degenerate case. Assume in the following $w$ is an edge, e.g.\ $w=e_k$.

\begin{description}
\item[Cells \& layers.] For each active pair $(u,u')$, define a region $\cell(u,u'):=\block(u,u')\cap \sector(w)$ and call it a \emph{cell}.
    For each edge $e_j$ in $(v_p\circlearrowright v_{q+1})$,
        we define a region $\mathsf{layer}_j$, called a \emph{layer}, which contains all the cells parallel to $e_j$ together with two more infinite strip region parallel to $e_k$ as shown Figure~\ref{fig:cells-layers}~(b).
\end{description}
\textbf{Note}: The ``cells'' and ``layers'' are counterparts of the ``roads'' and ``routes'' (recall them in Figure~\ref{fig:mid_V}).

\subparagraph*{Algorithm for computing $u^*_1$ and $u^*_2$ (overview).} We prove a monotonicity between the cells within the same layer and a monotonicity between the different layers.
  Based on these monotonicities,
     in $O(\log n)$ time we can determine the relative position between any layer and $V$.
      So, in $O(\log^2n)$ time we find the layer that contains $V$.
      In another $O(\log n)$ time we find the cell (within this layer) that contains $V$, which must be $\cell(u^*_1,u^*_2)$.

\begin{remark}\label{remark:symmetric}
In our algorithm, computing the block containing $V$ is \textbf{highly symmetric} to computing the endpoints of $\sector(V)\cap \partial P$.
The biggest challenges in both of them lie in defining the delimiting edges.
Both of them take $O(\log^2n)$ time for each query.
Moreover, $(s,t)$ and $(p,q)$ can both be computed in $O(\log n)$ time (not shown in this section).
This symmetry suggests that if one of them can be optimized to $O(\log n)$ time, so may the other.

The tricky definitions of $s,t$ and $p,q$ are given in subsection~\ref{subsect:singlesector_st} and subsection~\ref{subsect:alg_B-pq}.
To define $p,q$ we need to apply once again the bounding-quadrants of the blocks and their properties introduced in subsection~\ref{subsect:bounding-quadrants}.

\smallskip In the future, we would like to study whether the term $O(\log^2n)$ can be optimized to $O(\log n)$ applying \cite{TPStechnique}.
\end{remark}
\clearpage

\section{Proofs of the lemmas stated in section~\ref{sect:pre}}\label{sect:lemmas-proof}

This section introduces some facts about the $Z$-points
  and proves the seven lemmas stated in section~\ref{sect:pre}.

\newcommand{\pl}{\mathsf{p}}
\smallskip Given point $X$ and edge $e_i$, denote by $\pl_i(X)$ the unique line at $X$ that is parallel to $e_i$.

When a point $X$ lies in $\partial P$, we abbreviate $back(\unit(X)),forw(\unit(X))$ as $back(X),forw(X)$ respectively.

\subsection{Several facts about the $Z$-points (in addition to Fact~\ref{fact:Z_bi-monotonicity})}

\newcommand{\D}{\mathsf{D}}
\newcommand{\I}{\mathsf{I}}   
Denote by $\D_i$ the unique vertex with the largest distance to $\el_i$. The uniqueness follows from the pairwise-nonparallel assumption of edges.
Denote by $\I_{i,j}$ the intersection of $\el_i$ and $\el_j$.

\begin{fact}\cite{arxiv:n2}\label{fact:dist-unique-location}
Point $Z_i^j$ lies in $[\D_i\circlearrowright\D_j]\cap (v_{j+1}\circlearrowright v_i)$.
Moreover, if it lies in some edge $e_k$, it lies at $(\I_{i,k}+\I_{j,k})/2$.
\end{fact}

\begin{fact}\label{fact:Z_advanced_bound}
Point $Z_i^j$ lies in or on the boundary of the opposite quadrant of $\qd_i^j$.
\end{fact}

\begin{proof}
See Figure~\ref{fig:Z-zeta-oppoquad}~(a). Let $M=(v_i+v_{j+1})/2$.
Denote by $H_1$ and $H_2$ the closed half-plane bounded by $\pl_i(M)$ and containing $v_{j+1}$,
and the closed half-plane bounded by $\pl_{j}(M)$ and containing $v_i$.
We shall prove that $Z_i^j\in H_1\cap H_2$.
Denote $e_b=back(Z_i^j)$. Because $Z_i^j$ has the largest distance-product to $(\el_i,\el_j)$ in $P$,
it has a larger distance-product to $(\el_i,\el_j)$ than all the other points on $e_b$.
Applying this observation with the concavity of $\dist_{\el_i,\el_j}()$ on segment $\overline{\I_{j,b}\I_{i,b}}$ (see Lemma~3 in \cite{arxiv:n2}),
   $|\I_{i,b}Z_i^j|\geq \frac{1}{2}|\I_{j,b}\I_{i,b}|$. So $d_{\el_i}(Z_i^j)\geq \frac{1}{2}d_{\el_i}(\I_{j,b})$.
Further since $\frac{1}{2}d_{\el_i}(\I_{j,b})\geq \frac{1}{2}d_{\el_i}(v_{j+1})=d_{\el_i}(M)$, we get $d_{\el_i}(Z_i^j)\geq d_{\el_i}(M)$.
Thus $Z_i^j\in H_1$. Symmetrically, $Z_i^j\in H_2$.
\end{proof}

\begin{figure}[h]
\centering \includegraphics[width=.7\textwidth]{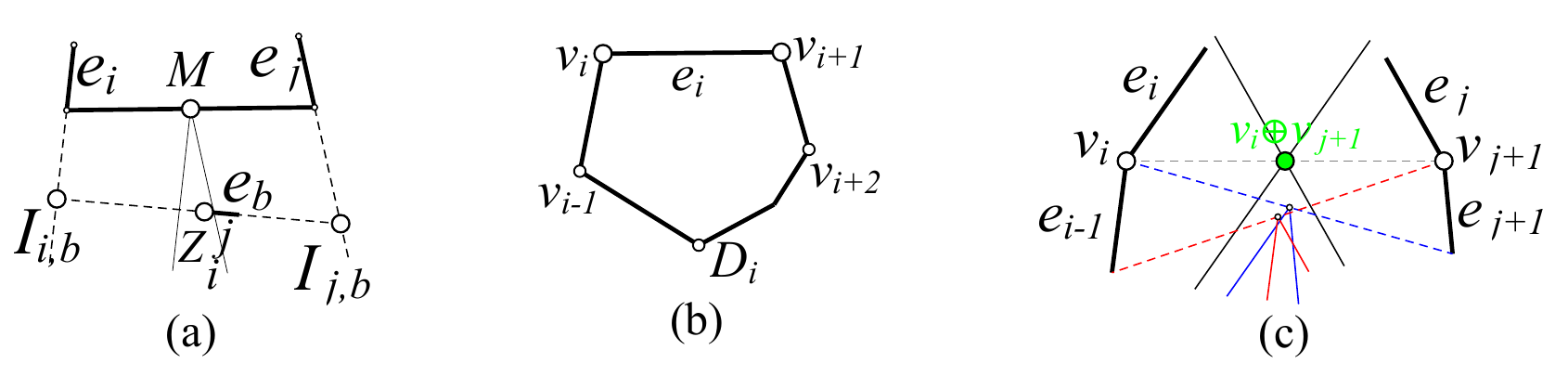}
\caption{Picture (a) illustrates the proof of Fact~\ref{fact:Z_advanced_bound}; whereas (b) and (c) illustrate the proof of
        Fact~\ref{fact:zeta-quad}.}\label{fig:Z-zeta-oppoquad}
\end{figure}

\begin{fact}\label{fact:zeta-quad}
When $v_i$ is chasing $v_{j+1}$, boundary-portion $\zeta(v_i,v_{j+1})$ lies in the opposite quadrant of $\qd_i^j$.
\end{fact}

\begin{proof}
Case~1: $i=j$. See Figure~\ref{fig:Z-zeta-oppoquad}~(b). By Fact~\ref{fact:dist-unique-location},
$Z_{i-1}^i\in (v_{i+1}\circlearrowright \D_i]$ whereas $Z_i^{i+1}\in [\D_i\circlearrowright v_i)$.
This means $[Z_{i-1}^i\circlearrowright Z_i^{i+1}]\subset (v_{i+1}\circlearrowright v_i)$, i.e.\ $\zeta(v_i,v_{i+1})\subset (v_{i+1}\circlearrowright v_i)$.
Further since $(v_{i+1}\circlearrowright v_i)$ is contained in the opposite quadrant of $\qd_i^i$,
   boundary-portion $\zeta(v_i,v_{i+1})$ is contained in the opposite quadrant of $\qd_i^i$.

\medskip \noindent Case~2: $i\neq j$. See Figure~\ref{fig:Z-zeta-oppoquad}~(c).
Let $\gamma$ be the intersection of $\partial P$ and the opposite quadrant of $\qd_i^j$.
Let $\rho=[v_{j+1}\circlearrowright v_i]$.
We claim: (i) \emph{$Z_{i-1}^j$ and $Z_i^{j+1}$ both lie in $\gamma$}, and (ii) \emph{$Z_{i-1}^j \leq_\rho Z_i^{j+1}$}.

Clearly, $\gamma$ is contained in $\rho$.
So (i) and (ii) together imply that $Z_{i-1}^j \leq_\gamma Z_i^{j+1}$.
This means $[Z_{i-1}^j \circlearrowright Z_i^{j+1}]$, i.e. $\zeta(v_i,v_{j+1})$, lies in $\gamma$ and hence in the opposite quadrant of $\qd_i^j$.

\medskip  Claim (ii) follows from Fact~\ref{fact:Z_bi-monotonicity} -- the bi-monotonicity of the $Z$-points.  We prove claim (i) in the following.
   We only show the proof of point $Z_{i-1}^j$; the proof of the other point $Z_i^{j+1}$ is symmetric and omitted.
    Since $v_i$ is chasing $v_{j+1}$, we get $e_{i-1}\prec e_j$.
    This implies that point $(v_{i-1}+v_{j+1})/2$ lies in the opposite quadrant of $\qd_i^j$.
    Further since this point is the apex of the opposite quadrant of $\qd_{i-1}^{j}$,
       we get (a): the opposite quadrant of $\qd_{i-1}^{j}$ and its boundary are contained in the opposite quadrant of $\qd_i^j$.
    By Fact~\ref{fact:Z_advanced_bound}, we get (b): $Z_{i-1}^j$ lies in or on the boundary of the opposite quadrant of $\qd_{i-1}^j$.
    Together, $Z_{i-1}^j$ lies in the opposite quadrant of $\qd_i^j$ and thus lies in $\gamma$.
\end{proof}

\subsection{\LRF and \LMF (Lemmas~\ref{lemma:LRF} and \ref{lemma:LMF})}

\begin{proof}[Proof of Lemma~\ref{lemma:LRF}]
Take distinct tuples $A=(A_1,A_2,A_3)$ and $B=(B_1,B_2,B_3)$ from $\T(u,u')$.
According to (\ref{def:T(u,u')}), we have: (i) $A_3\in u,A_1\in u'$; (ii) $B_3\in u,B_1\in u'$; and (iii) $A_2,B_2\in \zeta(u,u')$.
We shall prove that $f(A)\neq f(B)$.

\smallskip First, assume $A_2=B_2$ and $(A_1,A_3) \neq (B_1,B_3)$.
    Because $(A_1,A_3) \neq (B_1,B_3)$ and by (i) and (ii), $A_1A_3$ and $B_1B_3$ are distinct line segments connecting $u,u'$.
    Further since all the line segments connecting $u,u'$ have different midpoints, $(A_3+A_1)/2\neq (B_3+B_1)/2$.
    Therefore, $A_3+A_1-A_2\neq B_3+B_1-B_2$, namely, $f(A)\neq f(B)$.

\smallskip Next, consider the case where $A_2\neq B_2$. By (iii), $\zeta(u,u')$ contains both $A_2$ and $B_2$ and hence is not a single point.
      This means at least one of $u,u'$ is a vertex.
    If $u,u'$ are both vertices, $(A_1+A_3)/2=u \oplus u'=(B_1+B_3)/2$,
     and combining this with $A_2\neq B_2$ would result $f(A)\neq f(B)$.
    So, we assume $u,u'$ are an edge and a vertex.
    Without loss of generality, assume $(u,u')=(v_i,e_j)$; the case $(u,u')=(e_i,v_j)$ is symmetric.

    We know $A_2$ and $B_2$ both lie in $\zeta(v_i,e_j)$ according to (iii), whereas
      $\zeta(v_i,e_j)=[Z_{i-1}^j\circlearrowright Z_i^j]\subseteq [v_{j+1}\circlearrowright \D_j]$ according to Fact~\ref{fact:dist-unique-location}.
    Moreover, all the points in $[v_{j+1}\circlearrowright \D_j]$ have different distances to $\el_j$.
    Altogether, $A_2$ and $B_2$ have different distances to $\el_j$.
    However, $(A_1+A_3)/2$ and $(B_1+B_3)/2$ have the same distance to $\el_j$ because they both lie in $v_i\oplus e_j$.
    Together, $f(A)$ and $f(B)$ have different distances to $\el_j$. Therefore, $f(A)\neq f(B)$.
\end{proof}
\smallskip

\begin{proof}[Proof of Lemma~\ref{lemma:LMF}]
Denote $(J_X,K_X,L_X)=f^{-1}_{u,u'}(X)$ for any point $X$ in $\block(u,u')$.
Note that $J_X\in u',K_X\in \zeta(u,u')$ and $L_X\in u$ since $(J_X,K_X,L_X)\in \T(u,u')$.
We shall prove that when $X$ travels (in clockwise) along a boundary-portion $\rho$ that lies in $\block(u,u')$,
 function $K_X$ moves along $\partial P$ in clockwise non-strictly.

The result is trivial for the case where $u,u'$ are both edges. Because in this case $K_X$ is invariant. It always equals $Z_u^{u'}$.
The result on the case where $u,u'$ are both vertex is easy to prove, because in this case $\block(u,u')$ is a curve,
  and we only need to show that when $X$ moves along $\block(u,u')$, function $K_X$ moves along $\zeta(u,u')$ in clockwise, which is obvious.
In the following, we focus on the case where $u,u'$ are a vertex and an edge.

Without loss of generality, assume $(u,u')=(v_i,e_j)$.
    Abbreviate $d_{\el_j}()$ by $d()$. We first state three arguments.
    \begin{enumerate}
    \item[(i)] When point $X$ travels along $\rho$ in clockwise, $d(X)$ (non-strictly) decreases.
    \item[(ii)] For any point $X$ in $\block(v_i,e_j)\cap \partial P$, the sum $d(X)+d(K_X)$ is a constant.
    \item[(iii)] If the position of a dynamic point $Y$ is restricted to $\zeta(v_i,e_j)$,
        and we observe that $d(Y)$ (non-strictly) increases during the movement of $Y$,
          we can conclude that $Y$ moves in clockwise (non-strictly) along $\zeta(v_i,e_j)$.
    \end{enumerate}
    Altogether, we can obtain the monotonicity property of $K_X$ as follows.
            Since $X$ travels along $\rho$,
                $d(X)$ non-strictly decreases due to (i).
            So, $d(K_X)$ non-strictly increases due to (ii).
            Finally, applying (iii) for $Y=K_X$, point $K_X$ travels along $\zeta(v_i,e_j)$ in clockwise non-strictly.
            We prove (i), (ii), and (iii) in the following. \smallskip

    \noindent \emph{Proof of (i):} It suffices to prove two facts:
    (a) When $X$ travels along $[v_i \circlearrowright v_{j+1}]$ in clockwise, $d(X)$ non-strictly decreases.
    (b) When a boundary-portion $\rho$ lies $\block(v_i,e_j)$, it must lie in $[v_i \circlearrowright v_{j+1}]$.
    Because $v_i$ is chasing $e_j$, we get $e_i\prec e_j$, which implies (a).
    The proof of (b) is given in the following.

    See Figure~\ref{fig:block-borders}~(b).
    Let $H$ denote the half-plane delimited by the extended line of $\overline{v_iv_{j+1}}$ and not containing $e_i$.
    By Fact~\ref{fact:dist-unique-location}, $Z_i^j$ and $Z_{i-1}^j$ both lie in $(v_{j+1}\circlearrowright v_i)$.
    Therefore, $\zeta(v_i,e_j)=[Z_i^j\circlearrowright Z_{i-1}^j]$ is contained in $H$.
    Since $\zeta(v_i,e_j)\subset H$ whereas $v_i\oplus e_j$ lies in the opposite half-plane of $H$, applying (\ref{eqn:block_reflect}), $\block(v_i,e_j)$ lies in the opposite half-plane of $H$.
    So, $\block(v_i,e_j)\cap \partial P\subseteq [v_i \circlearrowright v_{j+1}]$.
    Further since $\rho\subseteq \block(v_i,e_j)\cap \partial P$, we get (b).\medskip

    \noindent \emph{Proof of (ii):} Since $f(J_X,K_X,L_X)=X$, we get $(X+K_X)/2=(J_X+L_X)/2$.
    Because $J_X\in u'$ and $L_X\in u$, point $(J_X+L_X)/2\in u\oplus u'=v_i\oplus e_j$.
    Therefore, $(X+K_X)/2\in v_i\oplus e_j$. Hence $d((X+K_X)/2)$ is a constant.
    Since $X,K_X\in \partial P$, they lie on the same side of $\el_j$, so $d(X)+d(K_X)=2d((X+K_X)/2)$. Together, we get (ii).\medskip

    \noindent \emph{Proof of (iii):} By Fact~\ref{fact:dist-unique-location}, $\zeta(v_i,e_j)=[Z_{i-1}^j\circlearrowright Z_i^j]\subseteq [v_{j+1}\circlearrowright \D_j]$, which implies that $d(Y)$ strictly increases when $Y$ travels along $\zeta(v_i,e_j)$. This simply implies (iii).
\end{proof}

\subsection{Peculiarity and monotonicity of the bounding-quadrants (Lemmas~\ref{lemma:br_peculiar} and \ref{lemma:br_monotone})}\label{subsect:br_properties}

\newcommand{\hp}{\mathsf{hp}}
We add one more notation before the next proofs.
For edge pair $(e_i,e_j)$ such that $e_i\preceq e_j$, denote by $\hp_i^j$ the \textbf{open} half-plane delimited
  by the extended line of $\overline{v_iv_{j+1}}$ and lies on the left side of
    $\overrightarrow{v_iv_{j+1}}$. Notice that $\qd_i^j\subseteq \hp_i^j$.

\begin{proof}[Proof of Lemma~\ref{lemma:br_peculiar}]
Assume that $e_a\preceq e_{a'}, e_b\preceq e_{b'}$ and $e_a,e_{a'},e_b,e_{b'}$ are not contained in any inferior portion.
We shall prove that $\qd_a^{a'}\cap \qd_b^{b'}$ lies in the interior of $P$.
We have to analyze several cases.

In the first three cases shown below, $\qd_a^{a'}\cap\qd_b^{b'}$ will be empty and hence it will lie in the interior of $P$.
\begin{itemize}
\item[Case~1] $a=a'$. Since $e_a,e_{a'},e_b,e_{b'}$ are not contained in any inferior portion, $e_a\prec e_b$ and $e_{b'}\prec e_a$.
    See Figure~\ref{fig:bb-peculiar-trivial}~(a).
    Denote by $H$ the half plane delimited by $\pl_a(M)$ and on the right of $e_a$, where $M$ is the apex of $\qd_{b}^{b'}$.
    Since $M=(v_b+v_{b'+1})/2$, this point lies in or on the right of $e_a$.
    Therefore, (1) $H$ is disjoint with $\qd_a^{a'}$.
    Since $e_a\prec e_b$, the part of boundary of $\qd_{b}^{b'}$ that is parallel to $e_b$ lies in $H$.
    Since $e_{b'}\prec e_a$, the part of boundary of $\qd_{b}^{b'}$ that is parallel to $e_{b'}$ also lies in $H$.
    Together, the entire boundary of $\qd_{b}^{b'}$ lies in $H$, and hence (2) $\qd_{b}^{b'}\subseteq H$.
    Combining (1) and (2), the subregion $\qd_b^{b'}$ of $H$ is disjoint with $\qd_a^{a'}$.
\item[Case~2] $b=b'$. This case is symmetric to Case~1.
\item[Case~3] $e_a,e_{a'},e_b,e_{b'}$ are distinct and lie in clockwise order on $\partial P$. See Figure~\ref{fig:bb-peculiar-trivial}~(b).
    Make four rays $r_a,r_{a'},r_b,r_{b'}$, which originate at $v_{a'+1}$, $v_a$, $v_{b'+1}$, $v_b$, respectively,
       and which have their directions the same as $e_a$, $e_{a'}$, $e_b$, $e_{b'}$, respectively.
    Let $\Pi_1$ be the region bounded by $r_{a'},\overline{v_av_{a'+1}},r_a$ and containing $\qd_a^{a'}$.
    Let $\Pi_2$ be the region bounded by $r_{b'},\overline{v_bv_{b'+1}},r_b$ and containing $\qd_b^{b'}$.
    Assume that $\Pi_1,\Pi_2$ do not contain their boundaries.
    Since $e_a,e_{a'},e_b,e_{b'}$ are not containing in any inferior portion, $e_{b'}\prec e_a$ whereas $e_{a'}\prec e_b$.
    This implies that $\Pi_1,\Pi_2$ are disjoint.
    Therefore, $\qd_a^{a'},\qd_b^{b'}$ are disjoint, since they are respectively subregions of $\Pi_1,\Pi_2$.
\end{itemize}

\begin{figure}[h]
\begin{minipage}[b]{0.5\textwidth}
\flushleft \includegraphics[width=.95\textwidth]{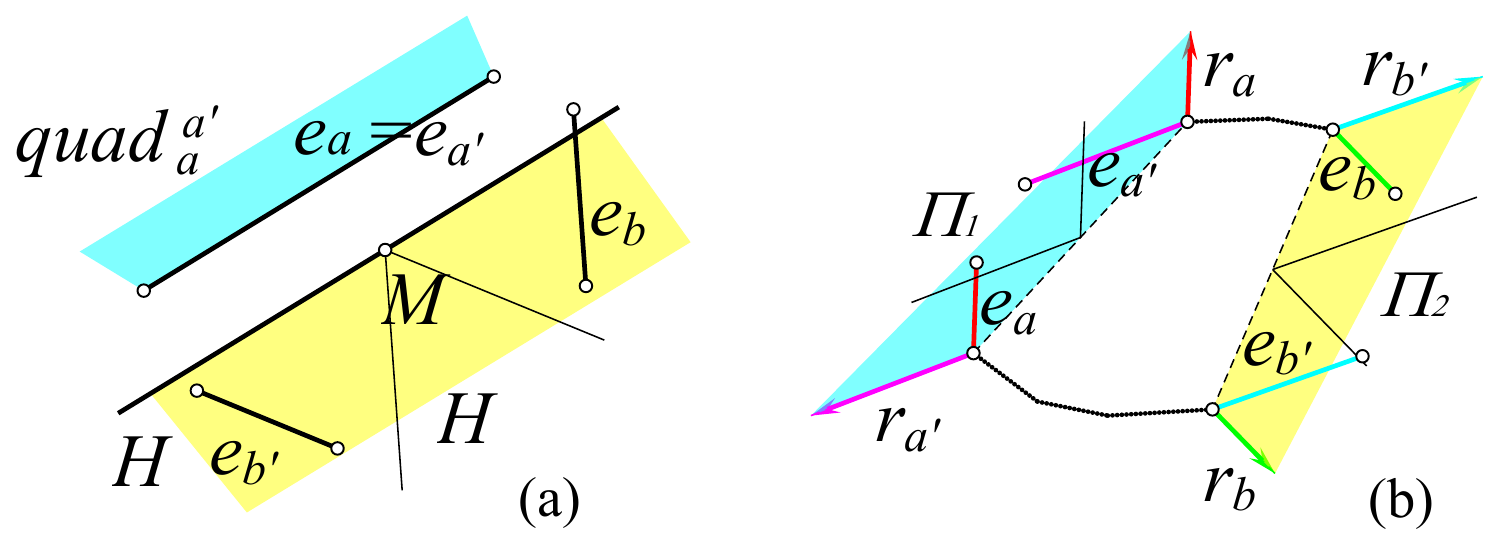}
\caption{Trivial cases in proving the peculiarity.}\label{fig:bb-peculiar-trivial}
\end{minipage}
\begin{minipage}[b]{0.5\textwidth}
\flushright \includegraphics[width=.95\textwidth]{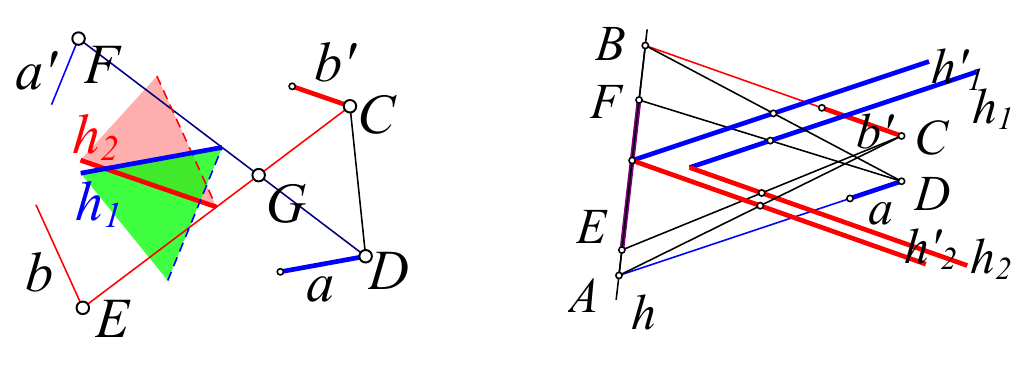}
\caption{Nontrivial cases in proving the peculiarity.}\label{fig:bb-peculiar}
\end{minipage}
\end{figure}

When none of the preceding (trivial) cases occur, only two cases remain:
either $e_a\prec e_b\preceq e_{a'}\prec e_{b'}\prec e_a$, or, $e_b\prec e_a\preceq e_{b'}\prec e_{a'}\prec e_b$.
Assume the first case occurs without loss of generality; the other case is symmetric.

See Figure~\ref{fig:bb-peculiar}. Let $C=v_{b'+1},D=v_a,E=v_b,F=v_{a'+1}$. Let $G$ denote the intersection of $CE$ and $DF$.
Obviously, $\triangle EFG$ lies in $P$.
So, proving that $\qd_a^{a'}\cap \qd_b^{b'}$ lies in the interior of $P$ reduces to proving that it lies in the interior of $\triangle EFG$, which further reduces to proving the following two facts:
\begin{enumerate}
\item[i.] $\qd_a^{a'}\cap \qd_b^{b'}$ lies in both $\hp_a^{a'}$ and $\hp_b^{b'}$. (See the definition of $\hp$ at the beginning of this subsection.)
\item[ii.] $\qd_a^{a'}\cap \qd_b^{b'}$ lies in half-plane $h$, where $h$ denotes the \textbf{open} half-plane bounded by the extended line of $\overline{EF}$ and containing $G$. (So, $h$ is the complementary half-plane of
    $\hp_b^{a'}$.)
\end{enumerate}

We have $\qd_a^{a'}\subseteq \hp_a^{a'}$ and $\qd_b^{b'}\subseteq \hp_b^{b'}$, which imply (i).
One of the two (open) half-planes defining $\qd_a^{a'}$ is parallel to $e_a$, denoted by $h_1$.
One of the two (open) half-planes defining $\qd_b^{b'}$ is parallel to $e_{b'}$, denoted by $h_2$.
See $h_1,h_2$ in the figure.
Clearly, $\qd_a^{a'}\cap \qd_b^{b'}\subseteq h_1\cap h_2$. So proving (ii) reduces to proving that $h_1\cap h_2\subseteq h$.

\smallskip Next, see the right picture of Figure~\ref{fig:bb-peculiar}.
Assume that the extended line of $\overline{EF}$ intersects $\el_a,\el_{b'}$ at $A,B$ respectively.
Denote by $h'_1$ the open half-plane bounded by $\pl_{a}((D+B)/2)$ and containing $e_a$, and
$h'_2$ the open half-plane bounded by $\pl_{b'}((A+C)/2)$ and containing $e_{b'}$.
Because $P$ is convex, points $E,F$ both lie in $\overline{AB}$.
Since $F$ lie in $\overline{AB}$, we know $(D+F)/2$ lies in or on the boundary of $h'_1$,
which implies $h_1\subseteq h'_1$ because $h_1$ is parallel to $h'_1$ and has $(D+F)/2$ on its boundary.
Since $E$ lie in $\overline{AB}$, we know $(C+E)/2$ lies in or on the boundary of $h'_2$,
which implies $h_2\subseteq h'_2$ because $h_2$ is parallel to $h'_2$ and has $(C+E)/2$ on its boundary.
It remains to show that $h'_1\cap h'_2\subseteq h$. By the definition of $h'_1,h'_2$, their boundaries pass through $(A+B)/2$.
Combining this with the facts that $h'_1$ is parallel to $e_a$ and $h'_2$ is parallel to $e_{b'}$, as well as $e_{b'}\prec e_a$, we obtain that $h'_1\cap h'_2\subseteq h$.
\end{proof}

\begin{figure}[h]
  \centering\includegraphics[width=.79\textwidth]{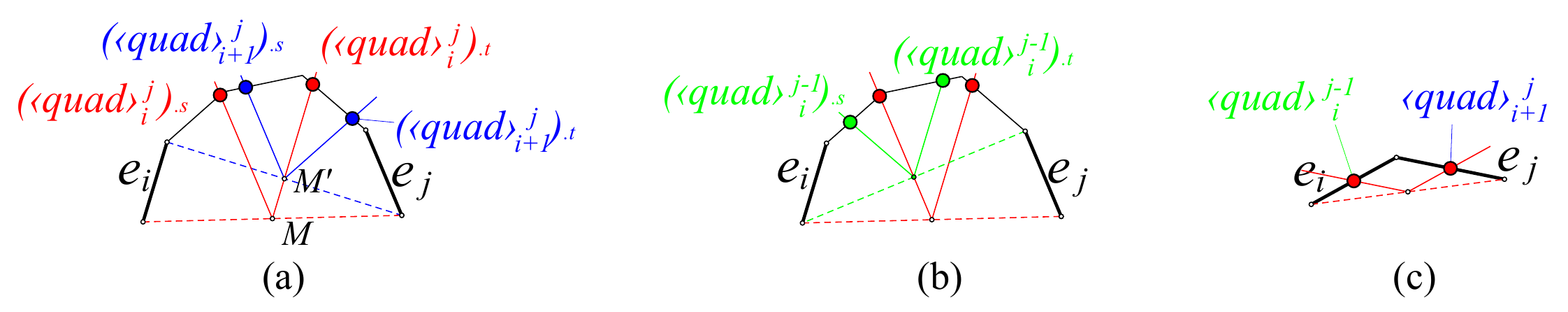}
  \caption{Illustration of the proof of the monotonicity of $\bp$.}\label{fig:bb-monotone}
\end{figure}

\begin{proof}[Proof of Lemma~\ref{lemma:br_monotone}]
(Recall this lemma contains two parts.) First, for $e_i\prec e_j$ and $\rho=[v_i\circlearrowright v_{j+1}]$, we shall prove:
\begin{eqnarray}
(\bp_i^{j-1}).s \leq_\rho (\bp_i^j).s \leq_\rho (\bp_{i+1}^j).s,\label{eqn:quad_s_monotone}\\
(\bp_i^{j-1}).t \leq_\rho (\bp_i^j).t \leq_\rho (\bp_{i+1}^j).t.\label{eqn:quad_t_monotone}
\end{eqnarray}

When $j=i+1$, the following (trivial) facts together imply (\ref{eqn:quad_s_monotone}) and (\ref{eqn:quad_t_monotone}). See Figure~\ref{fig:bb-monotone}~(c).

(i) $\bp_i^{j-1}$ contains a single point, which is the midpoint of $e_i$.

(ii) $\bp_{i+1}^j$ contains a single point, which is the midpoint of $e_j$.

(iii) $\bp_i^j$ starts at the midpoint of $e_i$ and terminates at the midpoint of $e_j$.

\medskip Now, assume $j\neq i+1$. See Figure~\ref{fig:bb-monotone}~(a). Let $M=(v_i+v_{j+1})/2$ and $M'=(v_{i+1}+v_{j+1})/2$.

\emph{Compare $(\bp_i^j).s$ and $(\bp_{i+1}^j).s$.}
Clearly, their distance to $\el_j$ are respectively equal to the distance from $M,M'$ to that line.
Moreover, since $MM'$ is parallel to $e_i$ whereas $e_i\prec e_j$, point $M$ is further to $\el_j$ than $M'$.
Together, $(\bp_i^j).s$ is further to $\el_j$ than $(\bp_{i+1}^j).s$.
This means $(\bp_i^j).s \leq_\rho (\bp_{i+1}^j).s$.

\emph{Compare $(\bp_i^j).t $ and $(\bp_{i+1}^j).t$.} Because segments
$M'(\bp_i^j).t$ and $M'(\bp_{i+1}^j).t$ are parallel to $e_i,e_{i+1}$ respectively,
  whereas $e_i\prec e_{i+1}$, it follows that $(\bp_i^j).t\leq_\rho (\bp_{i+1}^j).t$.

\smallskip Symmetrically, $(\bp_i^{j-1}).s\leq_\rho (\bp_i^j).s$ and $(\bp_i^{j-1}).t\leq_\rho (\bp_i^j).t$. See Figure~\ref{fig:bb-monotone}~(b).

\medskip Next, consider the second part of the lemma. For a list of $m$ boundary-portions $\bp_{u_1}^{u'_1},\ldots,\bp_{u_m}^{u'_m}$ where
(1) units $u_1,\ldots,u_m$ lie in clockwise order,
(2) units $u'_1,\ldots,u'_m$ lie in clockwise order, and
(3) $u_k$ is chasing $u'_k$ for $1\leq k\leq m$,
  we shall prove that their starting points lie in clockwise order and so do their terminal points.

\smallskip Denote $a_k=forw(u_k)$ and $a'_k=back(u'_k)$ for $1\leq k\leq m$.
According to the assumptions (1), (2), and (3), we have:
  (i) the edges $a_1,\ldots,a_m$ lie in clockwise order,
        and (ii) the edges ${a'}_1,\ldots,{a'}_m$ lie in clockwise order,
        and (iii) $a_k\preceq {a'}_k$ for $1\leq k\leq m$.
According to these facts and by applying inequalities (\ref{eqn:quad_s_monotone}) and (\ref{eqn:quad_t_monotone}), we get
\begin{itemize}
\item the starting points of $\bp_{a_1}^{a'_1},\ldots,\bp_{a_m}^{a'_m}$ lie in clockwise order around $\partial P$, and
\item the terminal points of $\bp_{a_1}^{a'_1},\ldots,\bp_{a_m}^{a'_m}$ lie in clockwise order around $\partial P$.
\end{itemize}
We complete the proof of the second part by recalling (\ref{eqn:def_bb}), which defines $\bp_{u_k}^{u'_k}:=\bp_{a_k}^{a'_k}$.
\end{proof}

\subsection{Two relations between blocks and bounding-quadrants (Lemmas~\ref{lemma:block-in-quad} and \ref{lemma:border-monotone})}\label{subsect:br-relation-block}

\begin{proof}[Proof of Lemma~\ref{lemma:block-in-quad}]
We shall prove $\block(u,u')\subset \qd_u^{u'}$. Recall $\qd_u^{u'}:=\qd_{forw(u)}^{back(u')}$. So, it reduces to proving
\begin{align*}
\block(e_i,e_j)&\subset \qd_i^j, & \block(v_i,v_{j+1})&\subset \qd_i^j, & \block(e_i,v_{j+1})&\subset \qd_i^j, & \block(v_i,e_j)&\subset \qd_i^j.
\end{align*}

\begin{itemize}
\item Proof of $\block(e_i,e_j)\subset \qd_i^j$. See Figure~\ref{fig:singleblock_Bregion}~(a).
    By Fact~\ref{fact:Z_advanced_bound}, $Z_i^j$ lies in or on the boundary of the opposite quadrant of $\qd_i^j$.
    By the definition of $\qd_i^j$, clearly $e_i\oplus e_j\subset \qd_i^j$.
    Together, the $2$-scaling of $e_i\oplus e_j$ with respect to $Z_i^j$, which equals $\block(e_i,e_j)$ due to (\ref{eqn:block_scale}), is contained in $\qd_i^j$.

\item Proof of $\block(v_i,v_{j+1})\subset \qd_i^j$. See Figure~\ref{fig:singleblock_Bregion}~(b).
    By Fact~\ref{fact:zeta-quad}, $\zeta(v_i,v_{j+1})$ lies in the opposite quadrant of $\qd_i^j$.
    So, its reflection with respect to $(v_i+v_{j+1})/2$, which equals $\block(v_i,v_{j+1})$ due to (\ref{eqn:block_reflect}), is in $\qd_i^j$.

\begin{figure}[h]
\centering \includegraphics[width=.75\textwidth]{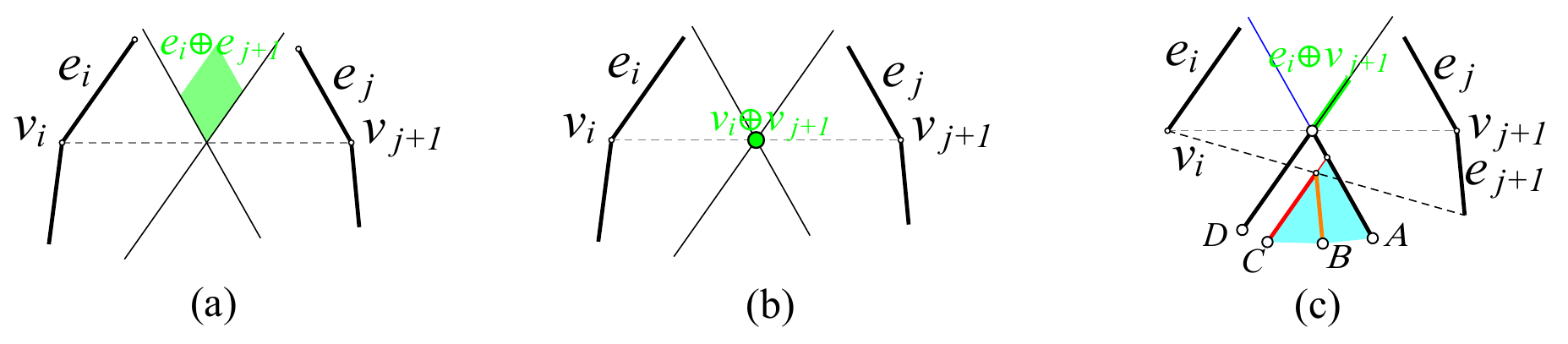}
\caption{Illustration of the proof of Lemma~\ref{lemma:block-in-quad}}\label{fig:singleblock_Bregion}
\end{figure}

\item Proof of $\block(e_i,v_{j+1})\subset \qd_i^j$. See Figure~\ref{fig:singleblock_Bregion}~(c).
   Denote by $H_1$ the closed half-plane delimited by line $\pl_{j}((v_i+v_{j+1})/2)$ and not containing $e_j$,
   and $H_2$ the closed half-plane delimited by line $\pl_{i}((v_i+v_{j+2})/2)$ and not containing $e_i$.
  The colored region in Figure~\ref{fig:singleblock_Bregion}~(c) shows $H_1\cap H_2$.
   We claim (i) $\zeta(e_i,v_{j+1})\subset H_1\cap H_2$.
    As a corollary of (i), the reflection of $\zeta(e_i,v_{j+1})$ with respect to any point in $e_i\oplus v_{j+1}$ lies in $\qd_i^j$.
       By (\ref{eqn:block_reflect}), $\block(e_i,v_{j+1})$ is the union of such reflections of $\zeta(e_i,v_{j+1})$.
        Together, $\block(e_i,v_{j+1})\subset \qd_i^j$.

   We prove (i) in the following. Notice that the intersection between $\partial P$ and the opposite quadrant of $\qd_i^j$ is a boundary-portion,
   denoted by $(A\circlearrowright D)$.
   Similarly, the intersection between $\partial P$ and the opposite quadrant of $\qd_i^{j+1}$ is a boundary-portion,
   denoted by $(B\circlearrowright C)$. We point out three facts:

   \quad (a) $[B\circlearrowright C]\subset [A\circlearrowright D]$.
   (b) $Z_i^j\in [A\circlearrowright D]$ and $Z_i^{j+1}\in [B\circlearrowright C]$.
   (c) $Z_i^j\leq_\gamma Z_i^{j+1}$, where $\gamma = [A\circlearrowright D]$.

    Fact~(a) follows from the fact that $e_i\prec e_{j+1}$, which follows from the assumption that $e_i$ is chasing $v_{j+1}$.
    Fact~(b) is an application of Fact~\ref{fact:Z_advanced_bound}.
   Fact~(c) is an application of the bi-monotonicity of $Z$-points (Fact~\ref{fact:Z_bi-monotonicity}).
   Combining facts (a), (b), and (c) would result $[Z_i^j\circlearrowright Z_i^{j+1}]\subseteq [A\circlearrowright C]$, which implies (i).
\end{itemize}
The proof of the last formula $\block(v_i,e_j)\subset \qd_i^j$ is symmetric to the preceding proof and is omitted.
\end{proof}
\smallskip
We introduce one more fact about the $Z$-points before proving Lemma~\ref{lemma:border-monotone}.

\begin{fact}\label{fact:Z-ray-fb}
Consider any edge pair $(e_a,e_{a'})$ for which $e_a\prec e_{a'}$. Let $e_b=back(Z_a^{a'})$ and $e_f=forw(Z_a^{a'})$.
\begin{enumerate}
\item The ray $r$ which originates from $v_a$ and has direction opposite to $e_{a'}$
        (assume it does not include its originate) lies on the right of $e_k$,
            for every $e_k$ in $\{e_{a'},\ldots,e_b\}$. See Figure~\ref{fig:singleblock_borderM}~(a).
\item The ray $r$ which originates from $v_{a'+1}$ and has direction the same as $e_a$
        (assume it does not include its originate) lies on the right of $e_k$,
            for every $e_k$ in $\{e_f,\ldots,e_a\}$. See Figure~\ref{fig:singleblock_borderM}~(b).
\end{enumerate}
\end{fact}

These two claims are obviously symmetric; so we only show the proof of the first one in the following.

\begin{proof}
By Fact~\ref{fact:dist-unique-location}, $Z_a^{a'}\in (v_{a'+1}\circlearrowright v_a)$.
This implies that $e_b\neq e_{a'}$.
Moreover, we claim that $e_b$ cannot be chasing $e_{a'}$.
Suppose to the opposite that $e_b\prec e_{a'}$.
Then, $d_{a'}(v_b)>d_{a'}(Z_{a}^{a'})$ and $d_{a}(v_b)>d_{a}(Z_{a}^{a'})$,
  which means $v_b$ has a larger distance-product to lines $(\el_a,\el_{a'})$ than $Z_{a}^{a'}$,
  contradicting the definition of $Z_{a}^{a'}$.
Therefore, $e_{a'}\prec e_b$.

Consider any $e_k$ in $e_{a'+1},\ldots,e_b$.
As $e_{a'}\prec e_b$, it holds that (i) $e_{a'}\prec e_k$.
Since $P$ lies in the closed half-plane delimited by $\el_k$ and lying on the right of $e_k$,
  we get (ii) $v_a$ (namely, the originate of $r$) lies on the right of $e_k$ or lies in $\el_k$.
  Also recall that (iii) the direction of $e_{a'}$ is opposite to that of $r$.
  Altogether, ray $r$ lies on the right of $e_k$.

It is clear that $r$ lies on the right of $e_{a'}$ as well. So the result holds for $e_k\in \{e_{a'},\ldots,e_b\}$.
\end{proof}

\begin{figure}[h]
\centering \includegraphics[width=\textwidth]{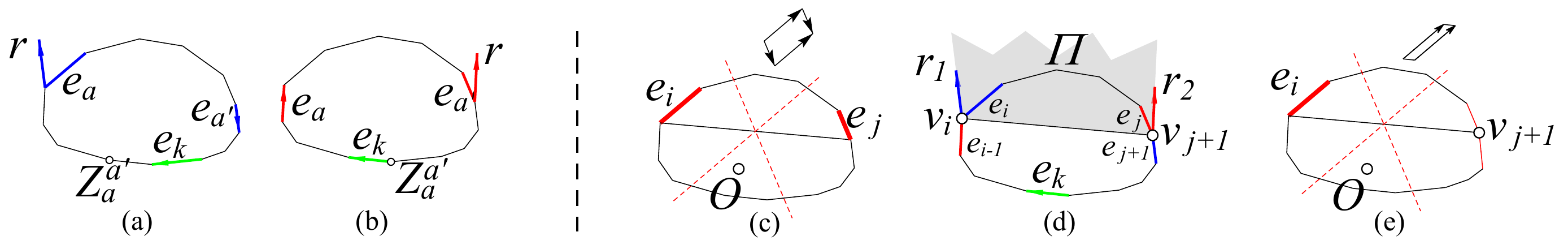}
\caption{Illustration of the proof of Lemma~\ref{lemma:border-monotone}}\label{fig:singleblock_borderM}
\end{figure}

\begin{proof}[Proof of Lemma~\ref{lemma:border-monotone}.]
Fix any point $O$ in $P$ that lies in the opposite quadrant of $\qd_u^{u'}$.
We shall prove:
(i) \emph{when a point $X$ travels along any border of $\block(u,u')$, it moves (strictly) in clockwise around $O$.}
\begin{itemize}
\item Case~1: \emph{both $u,u'$ are edges}.
        By Lemma~\ref{lemma:block-in-quad}, $\block(u,u')\subset\qd_u^{u'}$.
        As a corollary, the opposite quadrant of $\qd_u^{u'}$, including $O$, are on the right of each border of $\block(u,u')$ (see Figure~\ref{fig:singleblock_borderM}~(c)), which implies (i).

\item Case~2: \emph{both $u,u'$ are vertices, e.g., $(u,u')=(v_i,v_{j+1})$}.
    Recall that the unique border of $\block(u,u')$ equals the reflection of $\zeta(u,u')$ with respect to $(v_i+v_{j+1})/2$.
    Let $O'$ be the reflection of $O$ with respect to $(v_i+v_{j+1})/2$.
    It reduces to proving the mirror consequence as follows:

    \quad (i') \emph{when a point $X$ travels along $\zeta(u,u')$, it moves (strictly) in clockwise around $O'$}.

See Figure~\ref{fig:singleblock_borderM}~(d). It further reduces to proving that
        $O'$ lies on the right of any $e_k$ such that $e_k\cap \zeta(u,u')\neq \varnothing$.
        Let $r_1$ be the ray at $v_i$ which has opposite direction to $e_{j+1}$.
        Let $r_2$ be the ray at $v_{j+1}$ which has the same direction as $e_{i-1}$.
        Let $\Pi$ denote the region on the right of $r_1$, the left of $r_2$, and the left of $\overrightarrow{v_iv_{j+1}}$.

        Consider any $e_k$ that intersects $\zeta(u,u')$.
        Since $\zeta(u,u')$ starts at $Z_{i-1}^j$, edge $e_k$ is in $\{forw(Z_{i-1}^j),\ldots, e_{i-1}\}$.
        Applying Fact~\ref{fact:Z-ray-fb} (part 2) at $(a,a')=(i-1,j)$, ray $r_2$ is on the right of $e_k$.
        Since $\zeta(u,u')$ terminates at $Z_i^{j+1}$, edge $e_k$ is in $\{e_{j+1},\ldots, back(Z_i^{j+1})\}$.
        Applying Fact~\ref{fact:Z-ray-fb} (part 1) at $(a,a')=(i,j+1)$, ray $r_1$ is on the right of $e_k$.
        Together, $\Pi$ is on the right of $e_k$. (To be clear, assume $\pi$ contains $r_1,r_2$ but not segment $v_iv_{j+1}$.)

        Because $O$ lies in the opposite quadrant of $\qd_u^{u'}$,
          it lies on the right of $\overrightarrow{v_iv_{j+1}}$.
        Because $O$ lies in $P$, it lies on the right of $e_{i-1}$ (or on $\el_{i-1}$) and on the right of $e_{j+1}$ (or on $\el_{j+1}$).
        As a corollary, $O'\in \Pi$.

        Combining the above two results, $O'$ is on the right of $e_k$, thus we complete the proof for this case.

\item Case~3: $u,u'$ are an edge and a vertex, e.g.\ $u=e_i,u'=v_{j+1}$.
Recall that in this case $\block(u,u')$ has four borders; two of which are congruent to the only edge in $u,u'$;
  whereas the other two are reflections of $\zeta(u,u')$. See Figure~\ref{fig:singleblock_borderM}~(e).
    Applying the technique developed in Case~1, we can prove the part of result (i) that is associated with the former two borders.
    Applying the technique developed in Case~2, we can prove the part of result (i) that is associated with the latter two borders.
    (For proving the second part, a useful trick is that we only need to show $O$ lies on the right of each oriented segment of the lower border.
       This simply implies that $O$ lies on the right of each oriented segment of the other border reflected from $\zeta(u,u')$.)
    Since the proof is almost the same as the proofs for the above two cases, we do not give them in detail.
\end{itemize}
\end{proof}

\subsection{The lemma which shows that $\mathcal{C}$ interleaves $\partial P$ under some conditions (Lemma~\ref{lemma:interleave})}

\begin{proof}[Proof of Lemma~\ref{lemma:interleave}]
We refer to $\beta_1,\alpha_1,\ldots,\beta_q,\alpha_q$ as the 1st, 2nd, etc., the $2q$-th \emph{fragment}; $q\geq 3$. We shall prove that the concatenation of these $2q$ fragments, i.e. curve $\mathcal{C}$, interleaves $\partial P$, when the following conditions hold:\smallskip

\emph{(a) The concatenation of $\alpha_{i-1},\beta_i,\alpha_i$ interleaves $\partial P$ (for $1\leq i\leq q$).\quad
(b) There are $S_1,T_1,\ldots,S_q,T_q$ lying in clockwise order around $\partial P$ such that
    $\beta_i\cap \partial P\subset [S_i \circlearrowright T_i]$ whereas $\alpha_i\cap \partial P\subset [S_i \circlearrowright T_{i+1}]$ (for $1\leq i\leq q$).}\smallskip

We assume that at least one fragment in $\alpha_1,\ldots,\alpha_q$ intersects $\partial P$. The case where none of them intersects $\partial P$ is much easier (almost trivial) and can be proved similarly. Without loss of generality, assume that $\alpha_q$ intersects $\partial P$.

\smallskip Let $\mathcal{D}$ denote the concatenation of the first $2q-1$ fragments. Since $\mathcal{C}$ is the concatenation of $\alpha_q$ and $\mathcal{D}$, to show that it interleaves $\partial P$ reduces to proving three facts:\smallskip

(i) \emph{$\alpha_q$ interleaves $\partial P$}.\quad
(ii) \emph{$\mathcal{D}$ interleaves $\partial P$}. \quad
(iii) \emph{We can find two points $A,B$ on $\partial P$ such that the points in $\alpha_q\cap \partial P$ are restricted to $[A\circlearrowright B]$
                whereas the points in $\mathcal{D}\cap \partial P$ are restricted to $[B\circlearrowright A]$. }\smallskip

(Note that in (iii), we abuse the notation a little bit so that $[B\circlearrowright A]$ means the entire $\partial P$ when $A=B$ .)

\bigskip \noindent \emph{Proof of (i)}: This one simply follows from condition~(a).

\medskip \noindent \emph{Proof of (ii)}:
We need some notations. Regard $S_1$ as the starting point of $\partial P$.
For two points $A,A'$ on $\partial P$, we say that $A$ lies \emph{behind} $A'$ if $A=A'$ or, $A$ will be encountered later than $A'$ traveling around $\partial P$ starting from $S_1$.
We say that fragment $\gamma$ lies \emph{behind} fragment $\gamma'$, if all of the points in $\gamma\cap \partial P$ lie behind all of the points in $\gamma'\cap \partial P$.

According to condition (a), each fragment interleaves $\partial P$. Therefore, proving (ii) reduces to proving that
\[\text{the $k$-th fragment lies behind the first $k-1$ fragments, for $1<k<2q$.}\]
We prove this observation in the following.
\begin{itemize}
\item[Case~1:] $k=2$.
    Applying conditions (a) and (b), the concatenation of $\beta_1,\alpha_1$ interleaves $\partial P$
      and has all its intersections with $\partial P$ restricted to $[S_1\circlearrowright T_2]$.
    This means $\alpha_1$ (i.e.\ the 2nd fragment) lies behind $\beta_1$ (i.e.\ the 1st fragment).
\item[Case~2:] $k>2$ and $k$ is odd. Assume the $k$-th fragment is $\beta_i$.
    By condition (b), the first $k-2$ fragments have all their intersections with $\partial P$ restricted to $[S_1\circlearrowright T_{i-1}]$.
       However, $\beta_i\cap \partial P\subset [S_i \circlearrowright T_i]$. So, the $k$-th fragment $\beta_i$ lies behind the first $k-2$ fragments.
    Applying the technique for proving Case~1, fragment $\beta_i$ lies behind the $(k-1)$-th fragment $\alpha_{i-1}$.
    Together, the $k$-th fragment lies behind all the first $k-1$ fragments.
\item[Case~3:] $k>2$ and $k$ is even. Assume the $k$-th fragment is $\alpha_i$.
    Similar to Case~1, $\alpha_i$ lies behind the $(k-1)$-th and $(k-2)$-th fragments $\beta_i,\alpha_{i-1}$.
    Following the ideas given in Case~2, $\alpha_i$ lies behind the first $k-3$ fragments.
    (To be more clear, $\alpha_i\cap \partial P\subset [S_i\circlearrowright T_{i+1}]$ whereas the first $k-3$ fragments
       having all their intersections with $\partial P$ restricted to $[S_1\circlearrowright T_{i-1}]$.)
    Together, the $k$-th fragment lies behind all the first $k-1$ fragments.
\end{itemize}

\noindent \emph{Proof of (iii)}: We simply choose the first and last points of $\alpha_q\cap \partial P$ to be points $A,B$.
  Recall that $\alpha_q\cap \partial P\neq \varnothing$; so $A,B$ are well defined.
  Assume that $A\neq B$ in the following, otherwise fact~(iii) is trivial.

  By the definition of $A$ and $B$, points $\alpha_q\cap \partial P$ are contained in $[A\circlearrowright B]$.
  So, we only need to prove that \emph{$\mathcal{D}\cap \partial P \subset [B\circlearrowright A]$},
    namely, each fragment beside $\alpha_q$ has all its intersections with $\partial P$ restricted to $[B\circlearrowright A]$.

  \smallskip  By condition~(a), the concatenation of $\alpha_q,\beta_1,\alpha_1$, or $\alpha_{q-1},\beta_q,\alpha_q$ interleaves $\partial P$.
  So, for the four fragments $\alpha_1,\beta_1,\alpha_{q-1},\beta_q$, their intersections with $\partial P$ do not lie in $(A\circlearrowright B)$,
    and hence can only lie in $[B\circlearrowright A]$.

  \smallskip Next, consider any fragment $\gamma$ other than $\alpha_q,\beta_1,\alpha_1,\alpha_{q-1},\beta_q$. Applying condition~(b),
        the points in $\gamma\cap \partial P$ lie in $[S_2\circlearrowright T_{q-1}]$. Therefore, it reduces to proving that $[S_2\circlearrowright T_{q-1}]\subseteq [B\circlearrowright A]$. The proof is as follows.

  \medskip Applying condition~(b), $\alpha_q\cap \partial P$ are contained in $[S_q\circlearrowright T_1]$, and so $[A\circlearrowright B]\subseteq [S_q\circlearrowright T_1]$.

  Since $S_1,T_1,\ldots,S_q,T_q$ lie in clockwise order around $\partial P$, $[S_q\circlearrowright T_1]\subseteq [T_{q-1}\circlearrowright S_2]$.

 Together, $[A\circlearrowright B]\subseteq [T_{q-1}\circlearrowright S_2]$.  Equivalently, $[S_2\circlearrowright T_{q-1}]\subseteq [B\circlearrowright A]$.
\end{proof}

\section{Proofs of \BD and \IP}\label{sect:fT-major}

This section proves \BD and \IP (recall a sketch in section~\ref{sect:techover}).

\subsection{Preliminary: some observations}\label{subsect:extremal}

Recall that $(e_c,e_{c'})$ is \emph{extremal} if $e_c\prec e_{c'}$ and
   $[v_c\circlearrowright v_{c'+1}]$ is not contained in any other inferior portions,
   and recall that $\Delta(c,c')$ is defined as $\left\{(u,u')\mid \text{unit $u$ is chasing $u'$, and } forw(u),back(u')\in \{e_c,e_{c+1},\ldots,e_{c'}\} \right\}$.

\begin{fact}\label{fact:extremalpairs}
There exist at least three extremal pairs.
\end{fact}

\begin{proof}
Obviously there must be at least one extremal pair.
  This can be made stronger as follows. For any edge pair $(e_i,e_j)$ such that $e_i\prec e_j$,
    there is an extremal pair $(e_{i'},e_{j'})$ such that $[v_{i'}\circlearrowright v_{j'+1}]$ contains $e_i$ and $e_j$.

Assume $(e_i,e_j)$ is extremal.
  Choose any edge $e_k$ that does not lie in $[v_i\circlearrowright v_{j+1}]$.
  Obviously, \[\text{$e_i\prec e_j$, $e_j\prec e_k$ and $e_k\prec e_i$}.\]

Starting from $(e_k,e_i)$, we find extremal pair $(e_a,e_b)$ so that $[v_{a}\circlearrowright v_{b+1}]$ contains $e_k,e_i$.
  Notice that $[v_{a}\circlearrowright v_{b+1}]$ is inferior and thus cannot contain $e_j$.
Starting from $(e_j,e_k)$, we find extremal pair $(e_c,e_d)$ so that $[v_{c}\circlearrowright v_{d+1}]$ contains $e_j,e_k$.
  Notice that $[v_{c}\circlearrowright v_{d+1}]$ is inferior and thus cannot contain $e_i$.
Thus we get three extremal pairs.
\end{proof}

For any set $S$ of unit pairs, denote $BLOCK[S]=\{\block(u,u')\mid (u,u')\in S\}$.

\begin{lemma}\label{lemma:RayLemma}
Assume $(e_c,e_{c'})$ is extremal. Let $O=(v_c+v_{c'+1})/2$. Consider any $\block(u,u')\in BLOCK[\Delta(c,c')]$.
\begin{enumerate}
\item Region $\block(u,u')$ does not intersects the opposite quadrant of $\qd_c^{c'}$.
\item When point $X$ travels along any border of $\block(u,u')$, it moves in clockwise around $O$.
\end{enumerate}
\end{lemma}

\begin{figure}[h]
\centering\includegraphics[width=.85\textwidth]{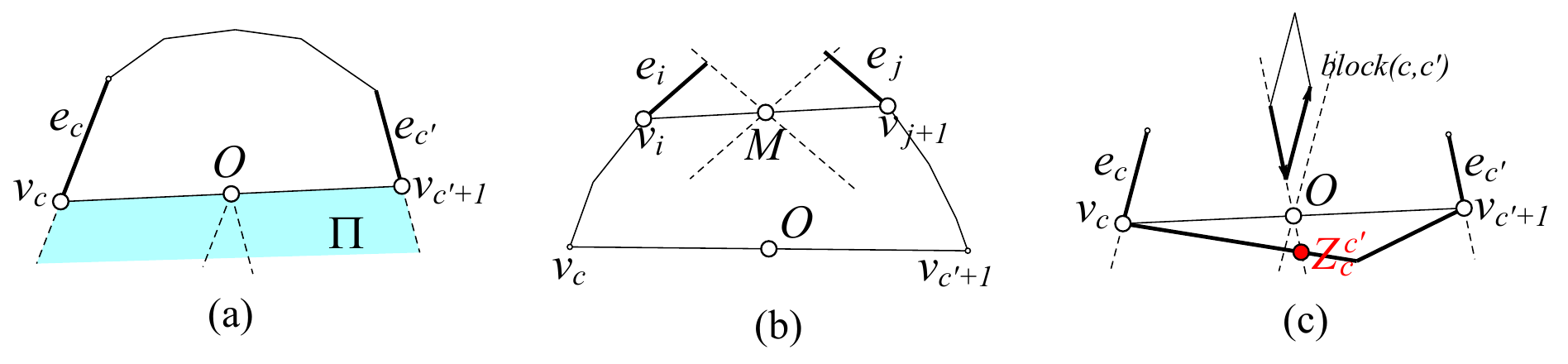}
\caption{Illustration of the proof of Lemma~\ref{lemma:RayLemma}.}\label{fig:raylemma}
\end{figure}

\begin{proof}
Let $e_i=forw(u)$ and $e_j=back(u')$. Because $(u,u')\in \Delta(c,c')$ and applying the definition of $\Delta(c,c')$,
 \begin{equation}\label{eqn:pRayLemma1}
e_i,e_j\text{ belong to }\{e_c,\ldots,e_{c'}\}\text{ and }e_i\preceq e_j.
\end{equation}

Recall that $\hp_i^j$ is the open half-plane delimited by the extended line of $\overline{v_iv_{j+1}}$ and lies on the left side of
    $\overrightarrow{v_iv_{j+1}}$.

Let $\Pi$ denote the region that lies on the right of $e_c$, $e_{c'}$ and  $\overrightarrow{v_cv_{c'+1}}$.
According to (\ref{eqn:pRayLemma1}) and the definition of $\hp^i_j$,
half-plane $\hp_i^j$ is disjoint with $\Pi$.
Further since $\block(u,u')\subset \qd_u^{u'}=\qd_i^j\subseteq\hp_i^j$,
    region $\block(u,u')$ is disjoint with $\Pi$.
Further since the opposite quadrant of $\qd_c^{c'}$ is a subregion of $\Pi$, we get part~1 of this lemma.

\medskip We prove part~2 in the following. First of all, assume that $(u,u')\neq (e_c,e_{c'})$.

When $(u,u')\neq (e_c,e_{c'})$, we claim that $(i,j)\neq (c,c')$.
Suppose $(i,j)=(c,c')$. Then, $(u,u')$ is one of $(e_c,e_{c'})$, $(e_c,v_{c'+1})$, $(v_c,e_{c'})$, $(v_c,v_{c'+1})$.
Since $(e_c,e_{c'})$ is extremal, $e_c$ is not chasing $v_{c'+1}$, $v_c$ is not chasing $e_{c'}$, and $v_c$ is not chasing $v_{c'+1}$.
Yet $u$ is chasing $u'$. Together, $(u,u')$ can only be $(e_c,e_{c'})$. This contradicts the assumption.

See Figure~\ref{fig:raylemma}~(b). Let $M=(v_i+v_{j+1})/2$. Consider the distance to $\el_j$. Because (\ref{eqn:pRayLemma1}),
$d_{\el_j}(v_c)\geq d_{\el_j}(v_i)$ and $d_{\el_j}(v_{c'+1})\geq d_{\el_j}(v_{j+1})$.
Notice that at least one of these inequalities is unequal since $(i,j)\neq (c,c')$. So,
\[d_{\el_j}(v_c)+d_{\el_j}(v_{c'+1})>d_{\el_j}(v_i)+d_{\el_j}(v_{j+1}).\]

The left and right sides equal to $2\cdot d_{\el_j}(O)$ and $2\cdot d_{\el_j}(M)$, respectively.
So $d_{\el_j}(O)>d_{\el_j}(M)$. Symmetrically, $d_{\el_i}(O)>d_{\el_i}(M)$.
Therefore, $O$ lies in the opposite quadrant of $\qd_i^j$, namely, it lies in the opposite quadrant of $\qd_u^{u'}$.
Further since $O\in P$, applying the monotonicity of the borders (Lemma~\ref{lemma:border-monotone}), we get part~2.

\smallskip When $(u,u')=(e_c,e_{c'})$, part~2 still holds.
    However, \emph{if $X$ travels along the two lower borders of $\block(e_c,e_{c'})$,
    it is possible that the orientation of $OX$ keeps invariant during the traveling process.}
    This occurs when $Z_c^{c'}$ lies on the boundary of the opposite quadrant of $\qd_c^{c'}$ as shown in Figure~\ref{fig:raylemma}~(c).
        (See Fact~\ref{fact:Z_advanced_bound} for more information.)
\end{proof}

\begin{note}
In most cases, $X$ moves in clockwise around $O$ \textbf{strictly}; namely, the orientation of $OX$ strictly increases. This is true for almost all borders; the only possible exceptions are the lower borders of $\block(e_{c},e_{c'})$.
\end{note}

\subsection{Proof of the \BD}

Recall \emph{local pairs} and \emph{global pairs} introduced in section~\ref{sect:techover}.
To prove the \BD, we prove:
\begin{itemize}
\item[(I)] When $\block(u,u'),\block(v,v')$ are a global pair, their intersection lies in the interior of $P$.
\item[(II)] When $\block(u,u'),\block(v,v')$ are a local pair, their intersection is empty!
\end{itemize}

As shown in section~\ref{sect:techover}, argument (I) easily follows from the peculiarity of the bounding-quadrants (Lemma~\ref{lemma:br_peculiar}).

Argument (II) is obviously equivalent to the following fact.

\begin{fact}\label{fact:local-area-disjoint}
For extremal pair $(e_c,e_{c'})$, the blocks in $BLOCK[\Delta(c,c')]$ are pairwise-disjoint.
\end{fact}

As mentioned in section~\ref{sect:techover}, we will prove Fact~\ref{fact:local-area-disjoint} by using the monotonicity of the borders (Lemma~\ref{lemma:border-monotone} and its successor Lemma~\ref{lemma:RayLemma}). First, we prove the following intermediate fact.

\begin{fact}\label{fact:local-area-disjoint-intermed}
For $(e_a,e_{a'})$ in $\Delta(c,c')$, all blocks in $BLOCK[\mathsf{U}(a,a')]$ are pairwise-disjoint, where
\[\mathsf{U}(a,a')=\left\{(u,u')\mid \text{$u$ is chasing $u'$, and $u,u'$ lie in $(v_a \circlearrowright v_{a'+1})$}\right\}.\]
\end{fact}

We introduce some notations for the proof of the above fact.
\begin{description}
\item[Tiling.] For a set $S$ of unit pairs, we call $BLOCK[S]$ a tiling if all the blocks in $BLOCK[S]$ are pairwise-disjoint.
\item[A quadrant region $\mathsf{SWEPT}_O(X,Y)$.]  For distinct points $O,X,Y$, imaging that a ray originated at $O$ rotates from $OX$ to $OY$ in clockwise, we denote by $\mathsf{SWEPT}_O(X,Y)$ the region swept by this ray.
\end{description}

\begin{figure}[h]
\begin{minipage}[b]{.5\textwidth}
\centering\includegraphics[width=.82\textwidth]{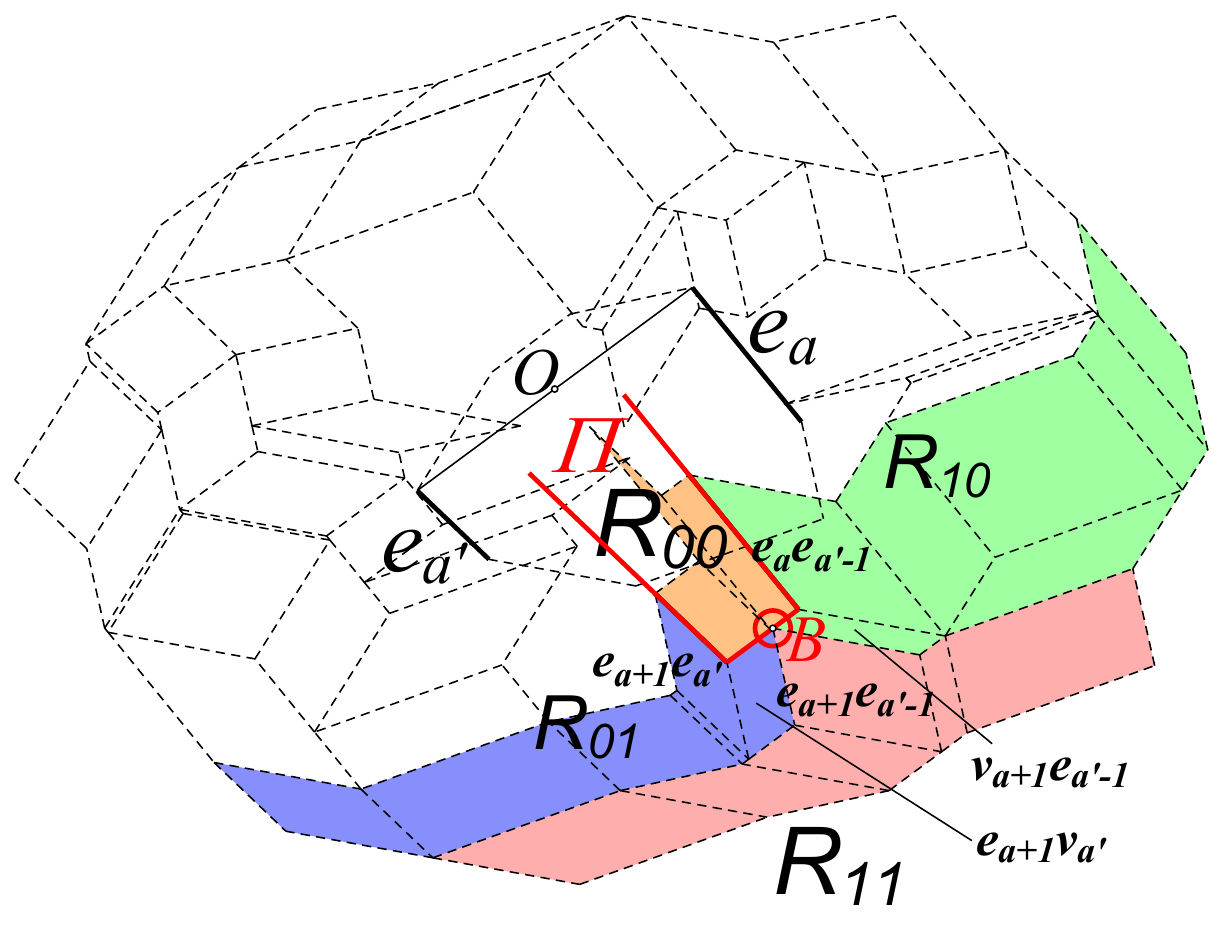}
\subcaption{}
\end{minipage}
\begin{minipage}[b]{.5\textwidth}
\centering\includegraphics[width=.82\textwidth]{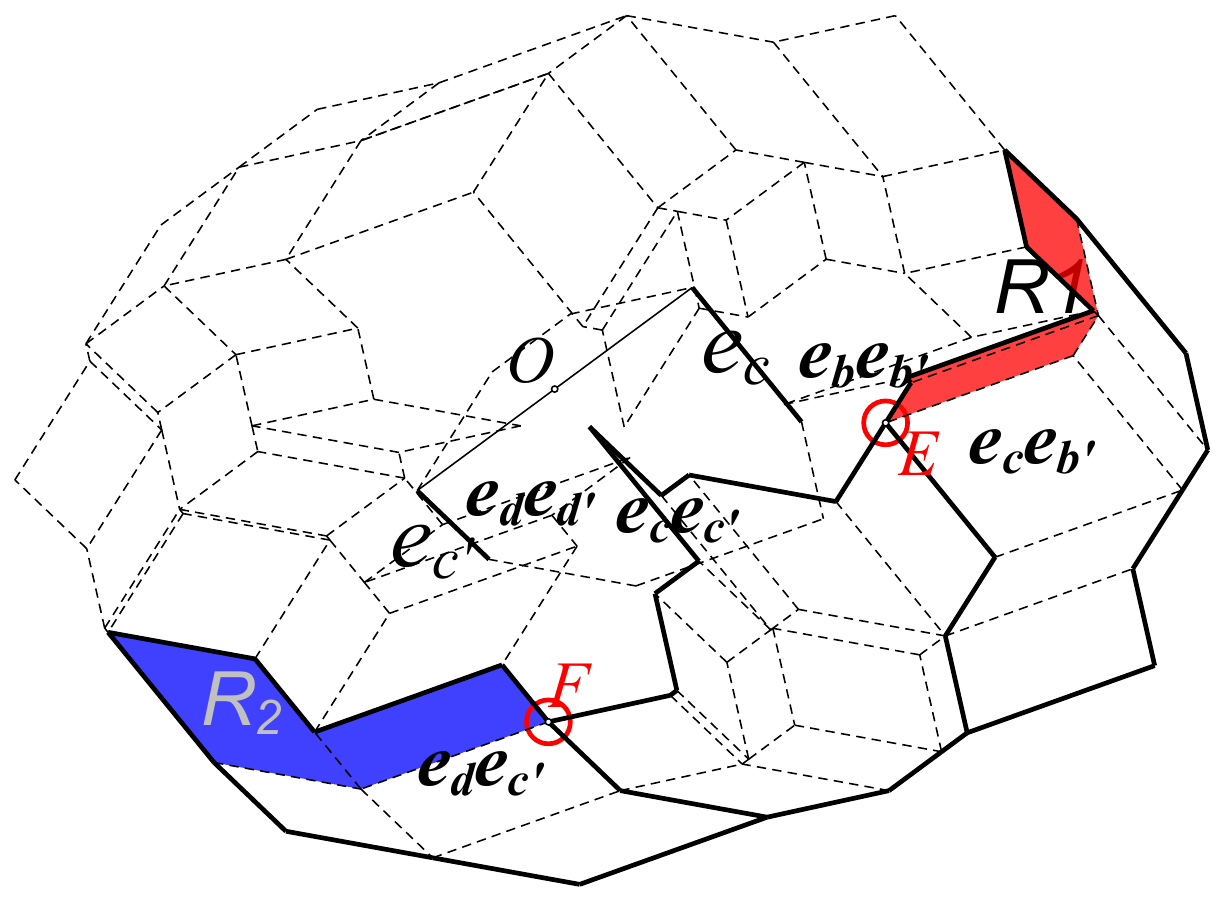}
\subcaption{}
\end{minipage}
\caption{Illustration of the proofs of the \BD}\label{fig:local_proof}
\end{figure}

\begin{proof}[Proof of Fact~\ref{fact:local-area-disjoint-intermed}]
We prove this fact by using induction on the number of edges $k$ in $(v_a\circlearrowright v_{a'})$.

\medskip \noindent \emph{Initial:} $k=2$, i.e., $a'=a+1$.
$BLOCK[\mathsf{U}(a,a')]$ contains only one block, $\block(e_a,e_{a+1})$, so the claim is trivial.

\medskip \noindent \emph{Induction:} $k>2$. Divide the unit pairs in $\mathsf{U}(a,a')$ into four parts distinguished by whether $\mathsf{U}(a,a'-1),\mathsf{U}(a+1,a')$ respectively contain them. Formally,
\begin{align*}
U_{10}&=\mathsf{U}(a,a'-1)-\mathsf{U}(a+1,a'),&U_{01}&=\mathsf{U}(a+1,a')-\mathsf{U}(a,a'-1),\\
U_{11}&=\mathsf{U}(a,a'-1)\cap \mathsf{U}(a+1,a'),&U_{00}&=\mathsf{U}(a,a')-\mathsf{U}(a,a'-1)-\mathsf{U}(a+1,a').
\end{align*}

By the induction hypothesis, $BLOCK[U_{01}]$, $BLOCK[U_{10}]$, $BLOCK[U_{11}]$ are tilings.
Moreover, since $U_{00}=\{(e_a,e_{a'}),(v_{a+1},e_{a'}),(e_a,v_{a'}),(v_{a+1},v_{a'})\}$
    only contains four unit pairs, by the definition of the blocks, $BLOCK[U_{00}]$ is also a tiling (see Figure~\ref{fig:blocks_def}; details omitted).
So, we only need to prove that $R_{00},R_{01},R_{10},R_{11}$ are pairwise-disjoint,
where $R_{00},R_{01},R_{10},R_{11}$ denotes the regions occupied by $BLOCK[U_{00}]$, $BLOCK[U_{01}]$, $BLOCK[U_{10}]$, $BLOCK[U_{11}]$, respectively.
It further reduces to proving the following four statements:\smallskip

(i) $R_{11},R_{10}$ are disjoint. \qquad (ii) $R_{11},R_{01}$ are disjoint.\medskip

(iii) $R_{01},R_{10}$ are disjoint. ~~\quad (iv) $R_{00}$ is disjoint with the other three regions.\medskip

See Figure~\ref{fig:local_proof}~(a).
Statement~(i) holds because $BLOCK[\mathsf{U}(a,a'-1)]$ is a tiling. Statement~(ii) holds because $BLOCK[\mathsf{U}(a+1,a')]$ is a tiling.
We prove the other two statements in the following. \medskip

\noindent \emph{Proof of (iii):}
Let $O=(v_a+v_{a'+1})/2$.
Let $A$ be an arbitrary point in the opposite quadrant of $\qd_c^{c'}$, and
let $B$ be the terminal point of the lower border of $\block(v_{a+1},e_{a'-1})$;
equivalently, $B$ is the starting point of the lower border of $\block(e_{a+1},v_{a'})$.
Our key observation is the following:
    \begin{equation}\label{eqn:key-obs-in-local}
        R_{10}\subset \mathsf{SWEPT}_O(A,B), \text{ whereas }R_{01}\subset \mathsf{SWEPT}_O(B,A).
    \end{equation}

The part $R_{10}\subset \mathsf{SWEPT}_O(A,B)$ is due to two reasons.
  1. When a point $X$ travels along the oriented borders in $BLOCK[U_{10}]$, it eventually reaches $B$.
  2. While $X$ is tracking down these borders, $OX$ keeps rotating in clockwise according to Lemma~\ref{lemma:RayLemma}.
The part $R_{01}\subset \mathsf{SWEPT}_O(B,A)$ is due to two similar reasons.

Further since $\mathsf{SWEPT}_O(B,A)$ is disjoint with $\mathsf{SWEPT}_O(A,B)$, we obtain statement~(iii).

\medskip \noindent \emph{Proof of (iv):}
Let $\Pi$ denote the region bounded by:
  $\C_1$ - the right lower border of $\block(e_a,e_{a'-1})$,
  $\C_2$ - the left lower border of $\block(e_{a+1},e_{a'})$, and
  $\C_3$ - the lower border of $\block(v_{a+1},v_{a'})$.
Observe that (1) $R_{00}$ is contained in $\Pi$; and
  (2) the united region of $R_{10},R_{01},R_{11}$ is also bounded by $\C_1,\C_2$ and $\C_3$ and hence is disjoint with $\Pi$.
Together, we get statement~(iv). The proofs of the last two observations are trivial yet burdensome and hence are omitted.
\end{proof}

Next, we prove Fact~\ref{fact:local-area-disjoint} and thus complete our proof of the \textsc{Block-disjointness}.

\begin{proof}[Proof of Fact~\ref{fact:local-area-disjoint}]
For convenience, let $(e_b,e_{b'}),(e_d,e_{d'})$ respectively denote the previous and next extremal pair of $(e_c,e_{c'})$ in the frontier-pair-list. We divide $\Delta(c,c')$ into three parts:
\[
U_1 = (\Delta(c,c')-\mathsf{U}(c,c'))\cap \mathsf{U}(b,b'), \quad
U_2 = (\Delta(c,c')-\mathsf{U}(c,c'))\cap \mathsf{U}(d,d'), \quad
U_3 = \mathsf{U}(c,c').
\]

It is obvious that $\Delta(c,c')=U_1\cup U_2\cup U_3$.
See Figure~\ref{fig:local_proof}~(b) for an illustration, where $R_1,R_2$ respectively indicate the regions occupied by $BLOCK[U_1]$, $BLOCK[U_2]$.

By Fact~\ref{fact:local-area-disjoint-intermed}, $BLOCK[\mathsf{U}(b,b')],BLOCK[\mathsf{U}(c,c')],BLOCK[\mathsf{U}(d,d')]$ are all tilings.
Therefore, $BLOCK[U_1]$, $BLOCK[U_2]$, and $BLOCK[U_3]$ are tilings. So, we only need to prove:
\begin{enumerate}
\item[(a)] Each block in $BLOCK[U_1]$ is disjoint with each block in $BLOCK[\Delta(c,c')-U_1]$.
\item[(b)] Each block in $BLOCK[U_2]$ is disjoint with each block in $BLOCK[\Delta(c,c')-U_2]$.
\end{enumerate}
We only show the proof of (a); the proof of (b) is symmetric. Clearly, (a) follows from
\begin{enumerate}
\item[(a1)] Each block in $BLOCK[U_1]$ is disjoint with each block in $BLOCK[\Delta(c,c')-\mathsf{U}(b,b')]$.
\item[(a2)] Each block in $BLOCK[U_1]$ is disjoint with each block in $BLOCK[\mathsf{U}(b,b')-U_1]$.
\end{enumerate}

\noindent Proof of (a1): Let $O=(v_c+v_{c'+1})/2$ and let $E$ be the common endpoint of the two lower borders of $\block(e_c,e_{b'})$.
Lemma~\ref{lemma:RayLemma} implies that
$BLOCK[U_1]$ lie in $\mathsf{SWEPT}_O(A,E)$ whereas $BLOCK[\Delta(c,c')-\mathsf{U}(b,b')]$ lie in $\mathsf{SWEPT}_O(E,A)$
(the proof is similar to the proof of the key observation (\ref{eqn:key-obs-in-local}) of Fact~\ref{fact:local-area-disjoint-intermed}).
Thus we obtain (a1).

\medskip \noindent Proof of (a2): Since $BLOCK[\mathsf{U}(b,b')]$ is a tiling, we have (a2).
\end{proof}

\subsection{Proof of the \IP}

\paragraph*{Step~1. Definition of $q$, the $2q$ fragments, and the $2q$ points $S_1,T_1,\ldots,S_q,T_q$.}

We choose $q$ to be the number of extremal pairs. Fact~\ref{fact:extremalpairs} states that $q\geq 3$.
Denote the $q$ extremal pairs in clockwise order by $(e_{c_1},e_{c'_1}), \ldots, (e_{c_q},e_{c'_q}).$

\smallskip Recall $\Delta(c,c')$ and recall the bottom borders of frontier blocks in subsection~\ref{subsect:pre-borders-sigmap}.
For any extremal pair $(e_c,e_{c'})$, denote
$\sigma(c,c')= \text{the concatenation of the bottom borders of the frontier blocks in } BLOCK[\Delta(c,c')]$,
  which is a \emph{directional} polygonal curve and is a fraction of $\sigma P$.
  The dashed curve in Figure~\ref{fig:sketch-key-properties}~(a) illustrates $\sigma(c,c')$.\smallskip

For each $1\leq i\leq q$, we define two fragments:
\begin{eqnarray}
\alpha_i=\text{the fragment of $\sigma P$ that is contained in both $\sigma(c_i,{c'}_i)$ and $\sigma(c_{i+1},{c'}_{i+1})$.}\\
\beta_i=\text{the fragment that belongs to $\sigma(c_i,c'_i)$ but does not belong to $\alpha_{i-1}$ or $\alpha_i$.}
\end{eqnarray}

See the left picture of Figure~\ref{fig:sigmaP_monotone} for an illustration of $\beta_1,\alpha_1,\ldots,\beta_q,\alpha_q$.

\begin{figure}[h]
\centering\includegraphics[width=.88\textwidth]{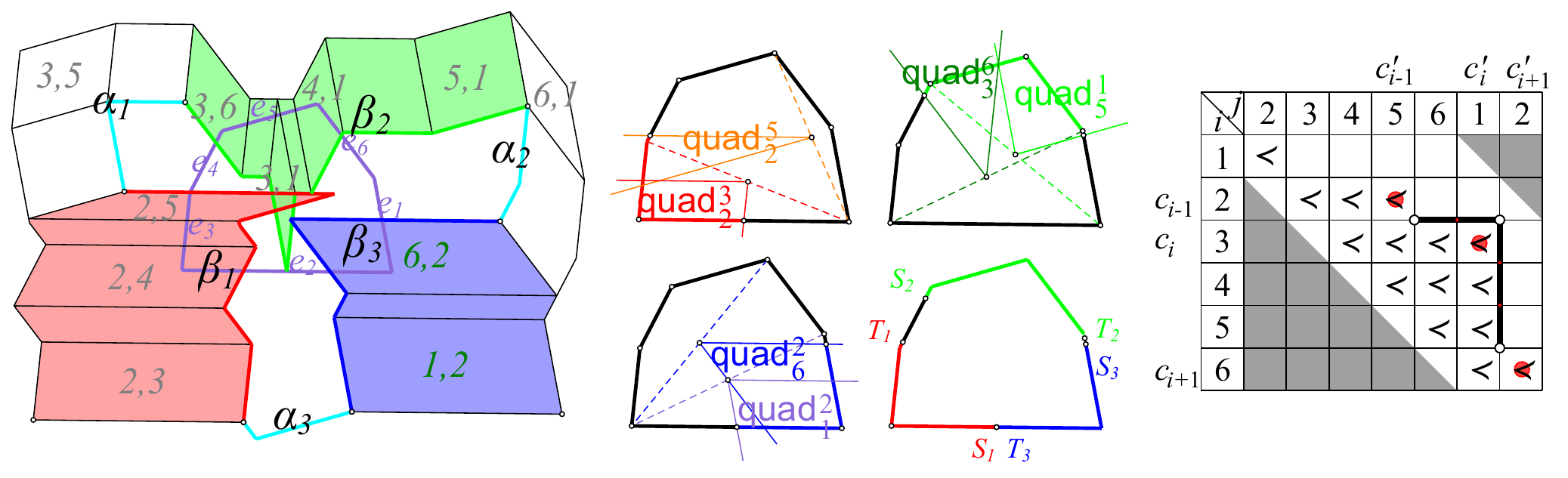}
\setcaptionwidth{.8\textwidth}
\caption{Illustration of the proof of the \IP.}\label{fig:sigmaP_monotone}
\end{figure}

For each extremal pair $(e_{c_i},e_{c'_i})$, define
\begin{equation}\label{eqn:aa'bb'}
(a_i,a'_i):=(c_i,c'_{i-1}+1),\qquad (b_i,b'_i):=(c_{i+1}-1,c'_i).
\end{equation}

\begin{fact}\label{fact:aba'b'}
$\beta_i$ begins with the bottom border of $\block(e_{a_i},e_{a'_i})$ and ends with the bottom border of $\block(e_{b_i},e_{b'_i})$.
\end{fact}

This fact simply follows the definition of $\beta_i$ and is illustrated by the right picture of Figure~\ref{fig:sigmaP_monotone}.\smallskip

Recall $\bp$ in (\ref{eqn:def_bb}). Using $a_i,a'_i,b_i,b'_i$ defined above, we further define
$2q$ points $S_1,\ldots,S_q$, $T_1,\ldots,T_q$ as follows. See the middle picture of Figure~\ref{fig:sigmaP_monotone} for an illustration.
\begin{equation}\label{def:XY}
  S_i= \text{the starting point of }\bp_{a_i}^{a'_i}, \quad T_i= \text{the terminal point of }\bp_{b_i}^{b'_i}.
\end{equation}

\paragraph*{Step~2. Verify these $2q$ points lying in order and prove the three conditions (\ref{eqn:local-interleave}), (\ref{eqn:delimitate-beta}), and (\ref{eqn:delimitate-alpha}).}

\begin{proof}[Proof of the fact that $S_1,T_1,\ldots,S_q,T_q$ lie in clockwise order around $\partial P$.]
Consider any pair of neighboring extremal pairs $(e_{c_i},e_{c'_i}),(e_{c_{i+1}},e_{c'_{i+1}})$.
Following Definition~\ref{def:extremal} and the definition in (\ref{eqn:aa'bb'}), we have
  \[\text{edges $e_{b_i},e_{b'_i},e_{a_{i+1}},e_{a'_{i+1}}$ are not in any inferior portion}.\]

According to this observation and the peculiarity of the bounding-quadrants (Lemma~\ref{lemma:br_peculiar}),
  $\bp_{b_i}^{b'_i}$ and $\bp_{a_{i+1}}^{a'_{i+1}}$ are disjoint (although their endpoints may coincide), for any $i~(1\leq i\leq q)$.
Combining this with (\ref{def:XY}) and the monotonicity of $\bp$ (Lemma~\ref{lemma:br_monotone}),
  the $q$ boundary-portions $(S_1 \circlearrowright T_1),\ldots,(S_q \circlearrowright T_q)$ are pairwise-disjoint and lie in clockwise order.
Therefore, $S_1,T_1,\ldots,S_q,T_q$ lie in clockwise order.
\end{proof}

\begin{proof}[Proof of condition~(\ref{eqn:local-interleave})]
Notice that the concatenation of $\alpha_{i-1},\beta_i,\alpha_i$ is exactly $\sigma(c_i,{c'}_i)$.
We shall prove that for each extremal pair $(e_c,e_{c'})$, the curve $\sigma(c,c')$ interleaves $\partial P$.
For ease of discussion, assume that
$\sigma(c,c')$ and $\partial P$ have a finite number of intersecting points.
Denote these points by $\I_1,\ldots,\I_x$, and assume that
\[\text{(i) $I_1,\ldots,I_x$ are sorted according to their priority on $\sigma(c,c')$.}\]
Denote $O=(v_c+v_{c'+1})/2$. Since (i) and by applying Lemma~\ref{lemma:RayLemma},
  rays $OI_1,\ldots,OI_x$ are in clockwise order. Further, because $O$ lies in $P$ , we get
\[\text{(ii) $\I_1,\ldots,\I_x$ lie in clockwise order around $\partial P$}.\]
Since $\I_1,\ldots,\I_x$ are all the intersecting points between $\sigma(c,c')$ and $\partial P$,  facts~(i) and (ii) together imply that: starting from $I_1$, regardless of whether we travel along $\sigma(c,c')$ (following its positive direction) or travel around $\partial P$ (in clockwise), we meet the points in $\sigma(c,c')\cap \partial P$ in identical order.
In other words, $\sigma(c,c')$ interleaves $\partial P$.
\end{proof}

Recall condition~(\ref{eqn:delimitate-beta}) states that $\beta_i\cap \partial P\subset [S_i\circlearrowright T_i]$.

\begin{proof}[Proof of condition~(\ref{eqn:delimitate-beta})]
Consider any frontier block whose bottom border is a fraction of $\beta_i$, e.g.\ $\block(u,u')$,
we shall prove that (I) the intersecting points between $\partial P$ and the bottom border of $\block(u,u')$ lie in $[S_i\circlearrowright T_i]$.

Denote by $\overline{\bp}_u^{u'}$ the closed set of $\bp_u^{u'}$, which contains $\bp_u^{u'}$ and its endpoints.
By Lemma~\ref{lemma:block-in-quad}, $\block(u,u')\subset \qd_u^{u'}$. So, the bottom border of $\block(u,u')$ lie in the closed set of $\qd_u^{u'}$.
    Therefore, the intersecting points between $\partial P$ and the bottom border of $\block(u,u')$ are contained in $\overline{\bp}_u^{u'}$.
    By the monotonicity of the $\bp$ (Lemma~\ref{lemma:br_monotone}) and the definition (\ref{def:XY}) of $S_i,T_i$, we get $\overline{\bp}_u^{u'}\subseteq [S_i\circlearrowright T_i]$. Together, we get (I).
\end{proof}

Condition~(\ref{eqn:delimitate-alpha}) can be proved the same way as condition~(\ref{eqn:delimitate-beta}); so proof omitted.

\paragraph*{Step~3. Complete the proof of \IP by applying Lemma~\ref{lemma:interleave}.}
Since the conditions of Lemma~\ref{lemma:interleave} are all satisfied, applying this lemma we obtain that $\sigma P$ interleaves $\partial P$.

\section{Proof: two more properties of $\Nest(P)$ (\MF \& \SC)}\label{sect:othertwo}

This section proves \MF and \SC (recall a sketch in section~\ref{sect:techover}).

\subsection{Proof of the \MF}\label{subsect:monoto-f}

We first give some terminologies and then prove the \MF.

See Figure~\ref{fig:outline_f-monotone}.
Let $K_1,\ldots,K_m$ denote all the intersecting points between $\sigma P$ and $\partial P$,
  and assume that they lie in clockwise order around $\partial P$.
Points $K_1,\ldots,K_m$ divide $\partial P$ into $m$ portions; each of which is called a \emph{$K$-portion}.

\begin{description}
\item[Top borders.]
    Recall the four borders of $\block(e_i,e_{i+1})$.
      We define the \emph{top border} of $\block(e_i,e_{i+1})$ as the concatenation of those two borders that are not the lower borders
       (recall the lower borders in subsection~\ref{subsect:pre-borders-sigmap}).
    Recall the unique border of $\block(v_i,v_{i+1})$. We define this border to be the \emph{top border} of $\block(v_i,v_{i+1})$.
\item[Outer boundary of $f(\T)$.] See Figure~\ref{fig:outerboundary}.
  We define the \emph{outer boundary} of $f(\T)$ as the concatenation of the top borders of $\block(e_1,e_2), \block(v_2,v_3), \ldots, \block(e_n,e_1), \block(v_1,v_2)$, which is closed by observing that the terminal point of the top border of one block equals the starting point of the top border of its next block.
\end{description}

\begin{figure}[h]
\begin{minipage}[b]{.5\textwidth}
\centering\includegraphics[width=.5\textwidth]{NestP-fmonotone.pdf}
\caption{Illustration of the $K$-points and $K$-portions.}\label{fig:outline_f-monotone}
\end{minipage}
\begin{minipage}[b]{.5\textwidth}
\centering\includegraphics[width=.5\textwidth]{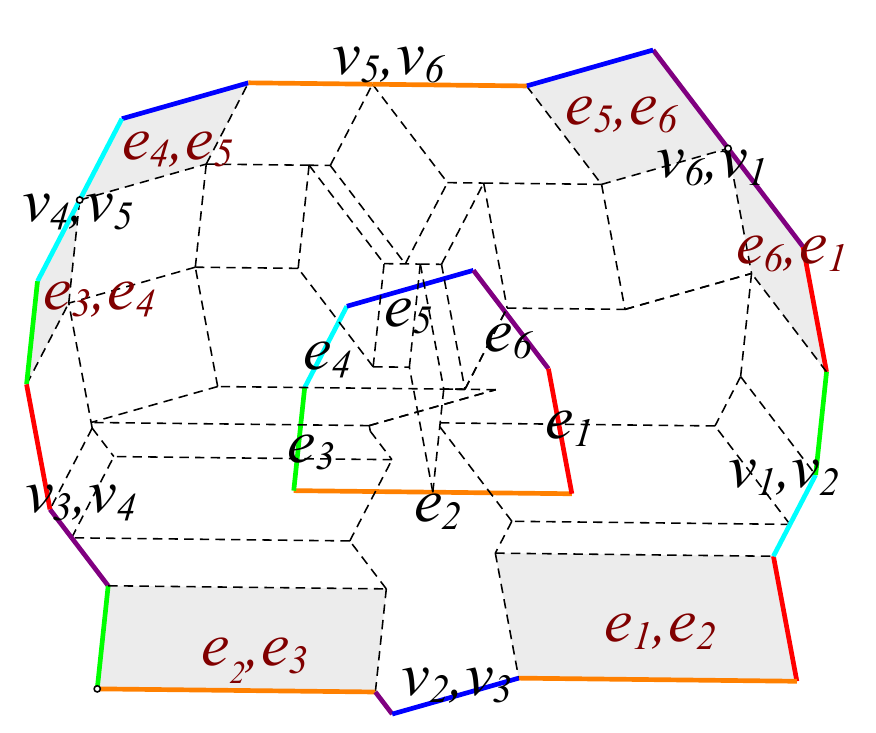}
\caption{Illustration of the outer boundary.}\label{fig:outerboundary}
\end{minipage}
\end{figure}

\begin{fact}\label{fact:outer-boundary-outside}
(1) All the top borders defined above lie outside $P$.
(2) The outer boundary of $f(\T)$ is a simple closed curve whose interior contains $P$.
(3) For any $K$-portion, it either lies in $f(\T)$ or lies outside $f(\T)$.
\end{fact}

\begin{proof}
(1) The top border of $\block(v_i,v_{i+1})$ is $\block(v_i,v_{i+1})$ itself, and by Lemma~\ref{lemma:block-in-quad} it lies in $\qd_i^i$ and hence lies outside $P$.
The top border of $\block(e_i,e_{i+1})$ is the concatenation of two borders;
  one is parallel to $e_i$ and the other is parallel to $e_{i+1}$.
Following from the fact $\block(e_i,e_{i+1})\subset\qd_i^{i+1}$,
  the former border lies on the left of $e_i$ whereas the latter one lies on the left of $e_{i+1}$; so both borders lie outside $P$.
Therefore, the top borders lie outside $P$.

\smallskip\noindent (2) We already know the outer boundary is closed.
Combining part~1 with the \textsc{Block-disjointness}, two top borders do not intersect.
Therefore, the outer boundary is simple and contains $P$ in its interior (see Figure~\ref{fig:outerboundary}).

\smallskip\noindent (3) Clearly, $f(\T)$ has an annular shape which is bounded by its inner and outer boundaries.
  Combining the above result on the outer boundary with the \textsc{Interleavity-of-$f$} of the inner boundary, we obtain part~3.
\end{proof}

\paragraph*{Extend the definition of $f_2^{-1}$.}
Recall $f_2^{-1}()$ in Theorem~\ref{thm:nestp}.
So far, it is defined on $f(\T)\cap \partial P$, but \textbf{not} on the $K$-points. This is implied by Fact~\ref{fact:fK-notdefinedyet} below.
However, to prove the monotonicity of $f_2^{-1}()$, it is convenient to also define $f_2^{-1}$ on the $K$-points.
Below we show a natural way to extend the definition of $f_2^{-1}$ on to those $K$-points.\smallskip

Consider any $K$-point $K_i$. Assume it comes from the bottom border of $\block(u,u')$. Recall the extended definition of $f^{-1,2}_{u,u'}()$ above Definition~\ref{def:g}. We simply define $f^{-1}_2(K_i):=f^{-1,2}_{u,u'}(K_i)$.

\begin{fact}\label{fact:fK-notdefinedyet}
$K_1,\ldots,K_m$ are not contained in $f(\T)\cap \partial P$.
\end{fact}

\begin{proof}
Without loss of generality, assume some $K$-point $K_i$ comes from the bottom border of the frontier block $\block(u,u')$.
We first argue that $u,u'$ cannot both be vertices. For a contradiction, suppose they are both vertices.
Then, $u'$ must be the clockwise next vertex of $u$ since $\block(u,u')$ is a frontier block (this can be observed from Algorithm~\ref{alg:FPL-def}).
Then, by Fact~\ref{fact:outer-boundary-outside}, $\block(u,u')$ lies outside $P$, so its bottom border (which is the block itself)
  has no intersection with $\partial P$.
This means $K_i$ cannot come from the bottom border of $\block(u,u')$. Contradictory.

Therefore, $u,u'$ comprise at least one edge.
    Then, according to Definitions~\ref{def:lower-borders} and \ref{def:bottom-borders},
     the lower borders and bottom border of $\block(u,u')$ are not contained $\block(u,u')$.
     This further implies that $K_i$ is not contained in $f(\T)$.
\end{proof}

\begin{proof}[Proof of the \MF]
We shall prove the circularly monotone property  of $f^{-2}(X)$. We state two observations first.
    (Recall function $g$ from $\sigma P$ to $\partial P$ in Definition~\ref{def:g}.)\smallskip

(i) \emph{The value of $f^{-1}_2()$ is continuous at the $K$-points.}\smallskip

(ii) \emph{Function $g$ is circularly monotone on curve $\sigma P$} (see Figure~\ref{fig:sigmaP-functionG}~(c)). \smallskip

The way we extend $f^{-1,2}_{u,u'}()$ onto the lower border of $\block(u,u')$ implies observation~(i).
Also according to this extension, $f^{-1,2}_{u,u'}()$ is monotone on the lower border of $\block(u,u')$,
  which implies observation~(ii).

\smallskip We then state two more observations.\smallskip

(iii) \emph{Points $f^{-1}_2(K_1),\ldots,f^{-1}_2(K_m)$ lie in clockwise order around $\partial P$.}\smallskip

(iv) \emph{Function $f^{-1}_2()$ is monotone on any $K$-portion that lies in $f(\T)$. In other words,
  when point $X$ travels along such a $K$-portion, $f^{-1}_2(X)$ moves in clockwise order around $\partial P$ non-strictly.}\smallskip

Proof of (iii): Since $K_1,\ldots,K_m$ lie in clockwise around $\partial P$, they lie in clockwise around $\sigma P$ due to the \IP,
  and thus $g(K_1),\ldots,g(K_m)$ lie in clockwise around $\partial P$ according to observation~(ii).
  Furthermore, notice that $f^{-1}_2(K_i)=g(K_i)$, we obtain observation~(iii).\smallskip

Proof of (iv): This observation follows from the \LMF (Lemma~\ref{lemma:LMF}); because when $X$ travels along a $K$-portion that lies in $f(\T)$, it travels inside some blocks (see Figure~\ref{fig:outline_f-monotone}).

\smallskip We now complete the proof. By Fact~\ref{fact:outer-boundary-outside}, $f(\T)\cap \partial P$ consists of those $K$-portions lying in $f(\T)$.
  Imagine that a point $X$ travels around $f(\T)\cap \partial P$ in clockwise;
    observation~(iv) assures that $f^{-1}_2(X)$ is monotone inside each $K$-portion,
    whereas observations~(i) and (iii) assure that $f^{-1}_2(X)$ is monotone between the $K$-portions.
\end{proof}


\subsection{Proof of the \SC}\label{sect:sector-continue}

\paragraph*{Outline of this subsection.} Assume $V$ is a fixed vertex.
\begin{enumerate}
\item  Prove that  $\sector(V)$ equals the 2-scaling of $({\bigcup}_{ (u,u')\in \Lambda_V} u\oplus u')$ with respect to $V$,
       where $$\Lambda_V =\{(u,u')\mid u \text{ is chasing }u', \text{ and }\zeta(u,u')\text{ contains }V \}.$$
        Moreover, define delimiting edges $e_{s_V},e_{t_V}$.
\item Prove a few crucial properties of $e_{s_V},e_{t_V}$. (see subsection~\ref{subsect:singlesector_st})
\item Define $\LV$, $\RV, \MID$ based on $e_{s_V}$ and $e_{t_V}$ (as shown in section~\ref{sect:techover}). Then,
    prove a structural property of $\Lambda_V$ applying $\MID$. (see subsection~\ref{subsect:MID-Lambda})
\item Define $\MIDSCALE,\LVS,\RVS$ as the 2-scalings of $\MID,\LV,\RV$ with respect to $V$, respectively.
    Prove that $\MIDSCALE$ is the closed set of $\sector(V)$.
    Based on this result, prove that $\LVS,\RVS$ are boundaries of $\sector(V)$
      and further prove the \SC. (see subsection~\ref{subsect:singlesector_finalproof})
\end{enumerate}

\begin{proof}[Proof of the above formula of $\sector(V)$.]
\[\begin{split}
\sector(V)=& f\left(\{(X_1,X_2,X_3)\in \T \mid X_2= V\}\right) \quad \text{(By definition (\ref{def:sector}))}\\
=& f\left({\bigcup}_{\text{$u$ is chasing $u'$}}\left\{(X_1,X_2,X_3)\mid X_1\in u', X_2= V,X_2\in \zeta(u,u'),X_3\in u\right\}\right)\\
=& f\left({\bigcup}_{u \text{ is chasing }u', V\in \zeta(u,u')}\left\{(X_1,V,X_3)\mid X_3\in u,X_1\in u'\right\}\right)\\
=& {\bigcup}_{(u,u')\in \Lambda_V} f(\{(X_1,V,X_3)\mid X_3\in u,X_1\in u'\})\\
=& {\bigcup}_{(u,u')\in \Lambda_V}\text{2-scaling of $(u\oplus u')$ with respect to $V$}\\
=&\text{2-scaling of $\bigg({\bigcup}_{(u,u')\in \Lambda_V}u\oplus u'\bigg)$ with respect to $V$.}
\end{split}
\]
\end{proof}

\begin{description}
\item [A relation $\leq_V$ between the edges.] We say $e_i$ is \emph{smaller than} $e_j$ or $e_j$ is \emph{larger than} $e_i$ (with respect to $V$),
   if $e_i$ would appear earlier than $e_j$ when we enumerate all the edges of $P$ in clockwise order from $forw(V)$ to $back(V)$.
Denote by $e_i\leq_V e_j$ if $e_i$ is smaller than or identical to $e_j$.
\end{description}
\newcommand{\UZ}{\omega}

The following definition is crucial to this subsection.
\begin{definition}[Delimiting edges $e_{s_V}$ \& $e_{t_V}$]\label{def:st}
Recall that $\D_i$ is the furthest vertex to $\el_i$. For any edge $e_i$, denote
\begin{equation}
\begin{array}{ccccc}
\UZ^+_i & = & \bigcup_{e_j: e_i\prec e_j}[v_{i+1}\circlearrowright Z_i^j] & = & [v_{i+1}\circlearrowright Z_i^{back(\D_i)}],\\
\UZ^-_i & = & \bigcup_{e_k: e_k\prec e_i}[Z_k^i \circlearrowright v_i]    & = & [Z_{forw(\D_i)}^i \circlearrowright v_i].
\end{array}
\end{equation}

We define $e_{s_V},e_{t_V}$ respectively as the smallest edge $e_i$ such that $\UZ^+_i$ contains $V$,
    and the largest edge $e_i$ such that $\UZ^-_i$ contains $V$.
    Throughout this subsection, ``smallest'' or ``largest'' is with respect to relation $\leq_V$ introduced above.
See Figure~\ref{fig:def_st} for an illustration for this definition.
We abbreviate $s_V,t_V$ as $s,t$.
\end{definition}

\textbf{Note}: Clearly, $V\in \UZ^+_{back(V)}$, so there is at least one element in $\{\UZ^+_i\}$ which contains $V$. So $s_V$ is well defined.\smallskip

\textbf{Note}: Clearly, $V\in \UZ^-_{forw(V)}$, so there is at least one element in $\{\UZ^-_i\}$ which contains $V$. So $t_V$ is well defined.

\begin{figure}[h]
\begin{minipage}{0.66\textwidth}
\centering \includegraphics[width=.85\textwidth]{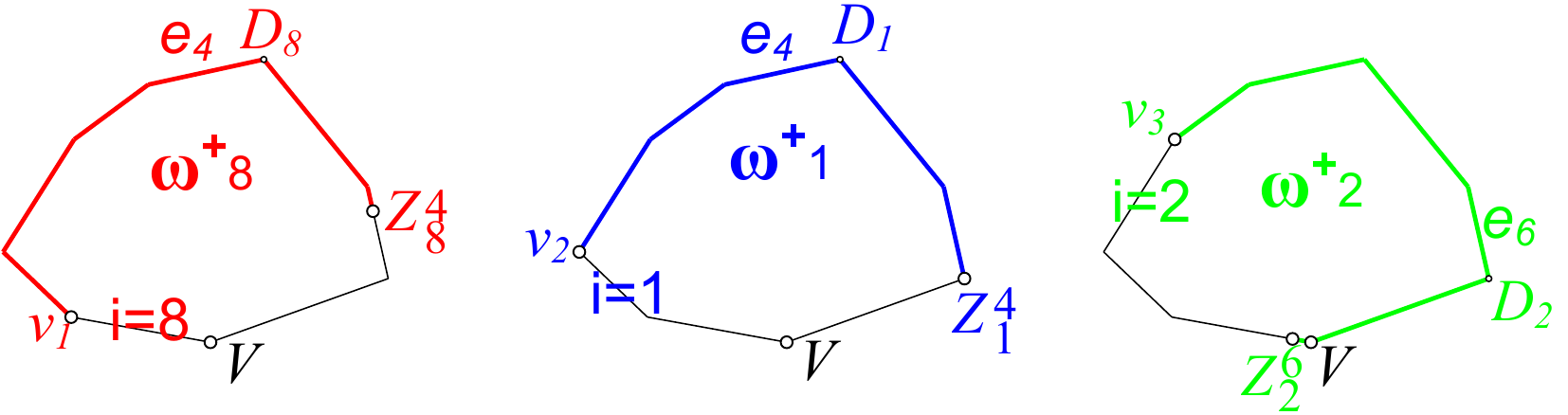}
\end{minipage}
\begin{minipage}{0.25\textwidth}
\centering\begin{tabular}{|c|}
  \hline
  $\UZ^+_8 = [ v_1 \circlearrowright Z_8^4]$ \\ \hline
  $\UZ^+_1 = [ v_2 \circlearrowright Z_1^4]$ \\ \hline
  $\UZ^+_2 = [ v_3 \circlearrowright Z_2^6]$ \\
  \hline
\end{tabular}
\end{minipage}\\ ~\\

\begin{minipage}{0.66\textwidth}
\centering \includegraphics[width=.85\textwidth]{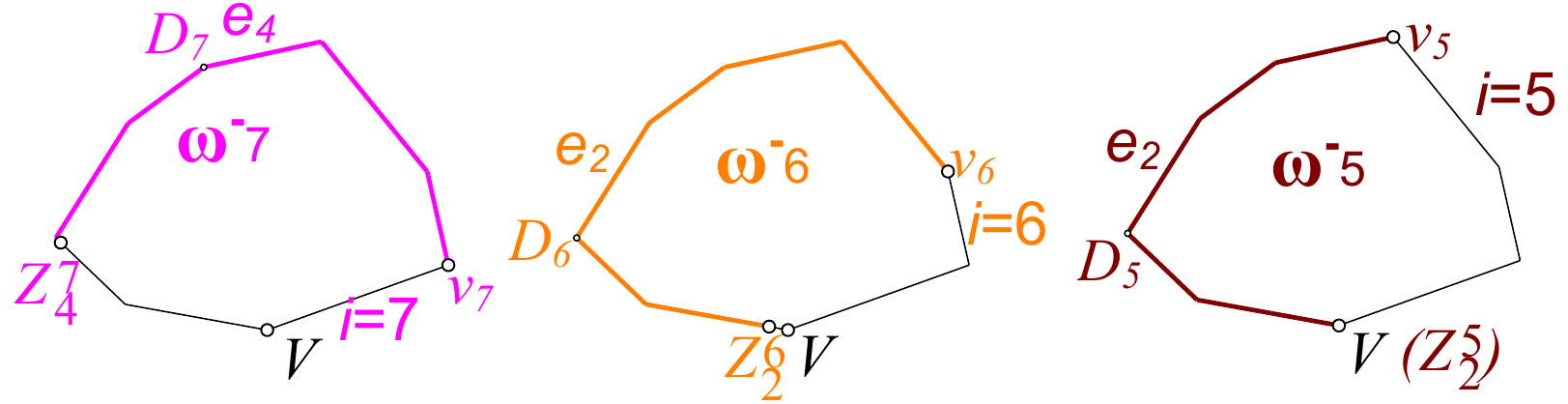}
\end{minipage}
\begin{minipage}{0.25\textwidth}
\centering\begin{tabular}{|c|}
  \hline
  $\UZ^-_7 = [ Z_4^7\circlearrowright v_7]$ \\ \hline
  $\UZ^-_6 = [ Z_2^6\circlearrowright v_6]$ \\ \hline
  $\UZ^-_5 = [ Z_2^5\circlearrowright v_5]$ \\
  \hline
\end{tabular}
\end{minipage}
\caption{Demonstration of the definitions of $s_V$ and $t_V$. Here, $s_V=2,t_V=5$.}\label{fig:def_st}
\end{figure}

\subsubsection{Crucial properties of $e_{s_V},e_{t_V}$}\label{subsect:singlesector_st}

\begin{fact}\label{fact:st1}
$e_s\preceq e_t$ and the inferior portion $[v_s\circlearrowright v_{t+1}]$ does not contain $V$.
\end{fact}

\begin{proof} Assume $V=v_1$ for simplicity. This proof is divided into three parts as shown below.\smallskip

\noindent \textbf{Part 1): prove that $e_s\leq_V e_t$.} To this end, we first state an observation:\smallskip

\quad (i) $e_s\leq_V e_{s^*}$ and $e_{t^*}\leq_V e_t$, where $e_{s^*}=forw(\D_n)$ and $e_{t^*}=back(\D_1)$.\smallskip

\noindent \emph{Proof of (i)}: See Figure~\ref{fig:s_chasing_t}~(a).
By Fact~\ref{fact:dist-unique-location}, $Z_{s^*}^{n}$ lies in $(V\circlearrowright v_{s^*})$.
Therefore, $V\in [v_{s^*+1} \circlearrowright Z_{s^*}^{n}]$.
Moreover, $[v_{s^*+1} \circlearrowright Z_{s^*}^{n}]\subseteq \UZ^+_{s^*}$ by the definition of $\UZ^+_{s^*}$.
Therefore, $V\in \UZ^+_{s^*}$, which implies $e_s\leq_V e_{s^*}$ due to the definition of $s$.
Symmetrically, $V\in \UZ^-_{t^*}$ and thus $e_{t^*}\leq_V e_t$.

\begin{figure}[b]
\centering \includegraphics[width=.62\textwidth]{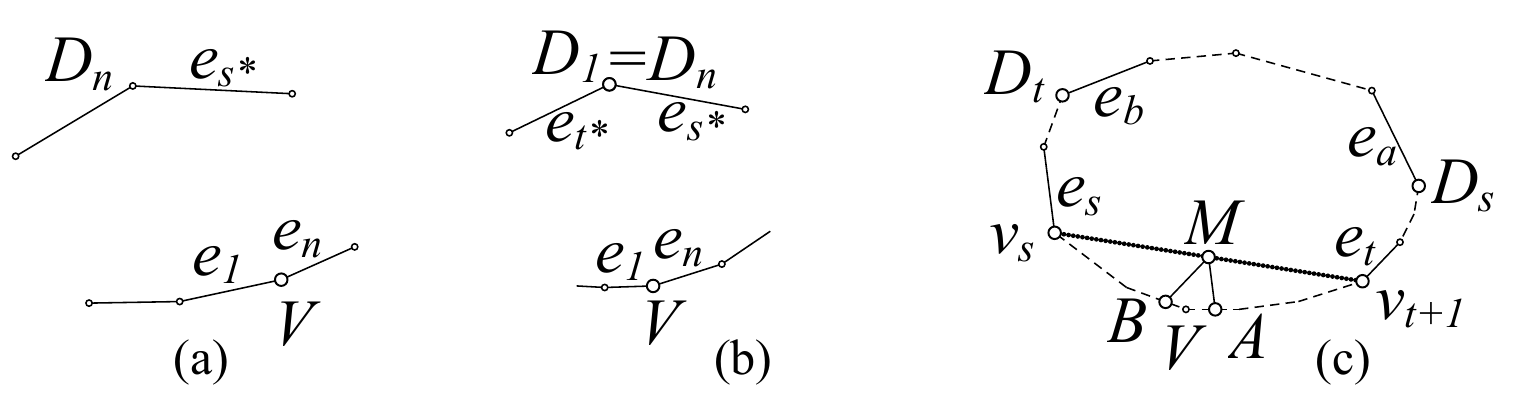}
\caption{Illustration of the proof of the relation between $e_s,e_t$ and $V$}\label{fig:s_chasing_t}
\end{figure}

We now use the above observation to prove that $e_s\leq_V e_t$.
\begin{itemize}
\item[Case~1] $\D_1\neq \D_n$. In this case $e_{s^*}\leq_V e_{t^*}$. Combining with observation (i), we get $e_s\leq_V e_t$.
\item[Case~2] $\D_1=\D_n$. See Figure~\ref{fig:s_chasing_t}~(b).
    In this case $Z_{t^*}^{s^*}$ is defined since $e_{s^*}$ is the next edge of $e_{t^*}$.
\begin{itemize}
\item [Case~2.1] $Z_{t^*}^{s^*}$ lies in $[V \circlearrowright \D_1]$. In this subcase, we argue that $e_s\leq_V e_{t^*}$.
 Since $V\in [\D_1\circlearrowright Z_{t^*}^{s^*}]$, whereas $[\D_1\circlearrowright Z_{t^*}^{s^*}]=[v_{t^*+1}\circlearrowright Z_{t^*}^{s^*}]\subseteq \UZ^+_{t^*}$,
 we get $V\in \UZ^+_{t^*}$, which implies that $e_s\leq_V e_{t^*}$ according to the definition of $s$.
 Combining $e_s\leq_V e_{t^*}$ with $e_{t^*}\leq_V e_t$ stated in observation (i), we get $e_s\leq_V e_t$.
\item [Case~2.2] $Z_{t^*}^{s^*}$ lies in $[\D_1 \circlearrowright V]$. In this subcase, we argue that $e_{s^*}\leq_V e_t$.
    The proof is symmetric to Case~2.1 and omitted.
  Combining $e_{s^*}\leq_V e_t$ with $e_s\leq_V e_{s^*}$ stated in observation (i), we get $e_s\leq_V e_t$.
\end{itemize}
\end{itemize}

\noindent \textbf{Part 2): prove that $[v_s\circlearrowright v_{t+1}]$ does not contain $V$.}

By the definition of $\UZ^+$, we get $V\notin\UZ^+_{forw(V)}$, which means $e_s\neq forw(V)$, i.e.\ $V\neq v_s$.
Symmetrically, $V\neq v_{t+1}$.
In addition, applying $e_s\leq_V e_t$, we get $V\notin (v_s\circlearrowright v_{t+1})$. Altogether, $V\notin [v_s\circlearrowright v_{t+1}]$.

\medskip \noindent \textbf{Part 3): prove that $e_s\preceq e_t$.}
For a contradiction, suppose that $e_t\prec e_s$, as shown in Figure~\ref{fig:s_chasing_t}~(c).

Denote $e_a=back(\D_s)$ and $e_b=forw(\D_t)$.
If $\D_s\neq \D_t$, denote $\rho=[\D_s\circlearrowright \D_t]$; otherwise, let $\rho$ denote the entire boundary of $P$ and assume that it starts and terminates at $\D_s$. Consider points $Z_s^a$ and $Z_b^t$, which lie in $\rho$ according to Fact~\ref{fact:Z_bi-monotonicity}.
We claim that (I) $Z_b^t\leq_\rho Z_s^a$ and (II) $Z_s^a<_\rho Z_b^t$, which lead to a contradiction.

\medskip \noindent \emph{Proof of (I).}
By definition of $s$, we have $V\in\UZ^+_{s}=[v_{s+1}\circlearrowright Z_s^a]$. This means $V\leq_\rho Z_s^a$.
By definition of $t$, we have $V\in \UZ^-_{t}=[Z_b^t\circlearrowright v_t]$. This means $Z_b^t\leq_\rho V$.
Together, we get (I).

\medskip \noindent \emph{Proof of (II).} Let $M=(v_s+v_{t+1})/2$.
Recall that $\pl_i(X)$ denotes the unique line at point $X$ that is parallel to $e_i$.
Let $A$ be the intersection of $\pl_{s}(M)$ and $[v_{t+1}\circlearrowright v_s]$, and $B$ the intersection of $\pl_{t}(M)$ and $[v_{t+1}\circlearrowright v_s]$.
We claim that $Z_s^a<_\rho A<_\rho B<_\rho Z_b^t$, which implies (II).
The inequality $A<_\rho B$ follows from the assumption $e_t\prec e_s$.
We prove $Z_s^a<_\rho A$ in the next paragraph; the proof of $B<_\rho Z_b^t$ is symmetric and omitted.

\medskip \noindent \emph{Proof of $Z_s^a<_\rho A$.} Denote by $h$ the open half-plane delimited by $\pl_{s}(M)$ and containing $v_{t+1}$.
By the definition of $\D_s$, it is further than $v_{t+1}$ on the distance to $\el_s$, so the mid point of $v_s$ and $\D_s$ is contained in $h$.
  Therefore the opposite quadrant of $\qd_s^a$, together with its boundary, are contained in $h$.
However, $Z_s^a$ lies in or on the boundary of the opposite quadrant of $\qd_s^a$ (by Fact~\ref{fact:Z_advanced_bound}).
Therefore, $Z_s^a\in h$, which implies that $Z_s^a<_\rho A$.
\end{proof}

\begin{fact}\label{fact:st2}
For $(u,u')\in \Lambda_V$, units $u,u'$ both lie in $[v_s\circlearrowright v_{t+1}]$.
\end{fact}

\begin{proof}
Let $e_a=back(u),e_{a'}=back(u'),e_b=forw(u),e_{b'}=forw(u')$.
Assume $(u,u')\in \Lambda_V$. Then, $V\in \zeta(u,u')$.
Notice that $V\in\zeta(u,u')=[Z_a^{a'}\circlearrowright Z_b^{b'}]\subseteq [v_{b+1}\circlearrowright Z_b^{b'}]\subseteq \UZ^+_b$.
So, $V\in \UZ^+_b$. This implies that $e_s\leq_V e_b$ by the definition of $s$.
Symmetrically, $V\in \UZ^-_{a'}$, which implies that $e_{a'} \leq_V e_t$ by the definition of $t$.

Since $u$ is chasing $u'$, we have $forw(u)\preceq back(u')$.
Together, $e_s\leq_V forw(u)\preceq back(u')\leq_V e_t$.
Further since $e_s\preceq e_t$ (By Fact~\ref{fact:st1}), $e_s\leq_V forw(u)\leq_V back(u')\leq_V e_t$,
hence $u$ and $u'$ both lie in $[v_s\circlearrowright v_{t+1}]$.
\end{proof}

\begin{fact}\label{fact:UZ-monotone}
(1) $V\in\UZ^+_i$ if and only if $e_s\leq_V e_i$.
(2) $V\in\UZ^-_j$ if and only if $e_j\leq_V e_t$,
\end{fact}

\begin{proof}
We only give the proof of Claim~1. Claim~2 is symmetric.

The ``only if'' part follows from the definition of $s$. We shall prove that $V\in \UZ^+_i$ for $e_i\in \{e_s,e_{s+1},\ldots,back(V)\}$.
We prove it by induction. Recall that $\UZ^+_i = [v_{i+1}\circlearrowright Z_i^{back(\D_i)}]$.
Initially, let $i=s$. We know $[v_{s+1}\circlearrowright Z_s^{back(\D_s)}]$ contains $V$ by the definition of $s$.
Next, consider $\UZ^+_{i+1}=[v_{i+2}\circlearrowright Z_{i+1}^{back(\D_{i+1})}]$.
See Figure~\ref{fig:st-properties}~(a). By the bi-monotonicity of the $Z$-points, $Z_{i+1}^{back(\D_{i+1})}$ lies in
  $[Z_i^{back(\D_i)}\circlearrowright v_{i+1}]$. This implies that $\UZ^+_{i+1}$ contains $V$.
\end{proof}

\begin{figure}[h]
\centering\includegraphics[width=0.7\textwidth]{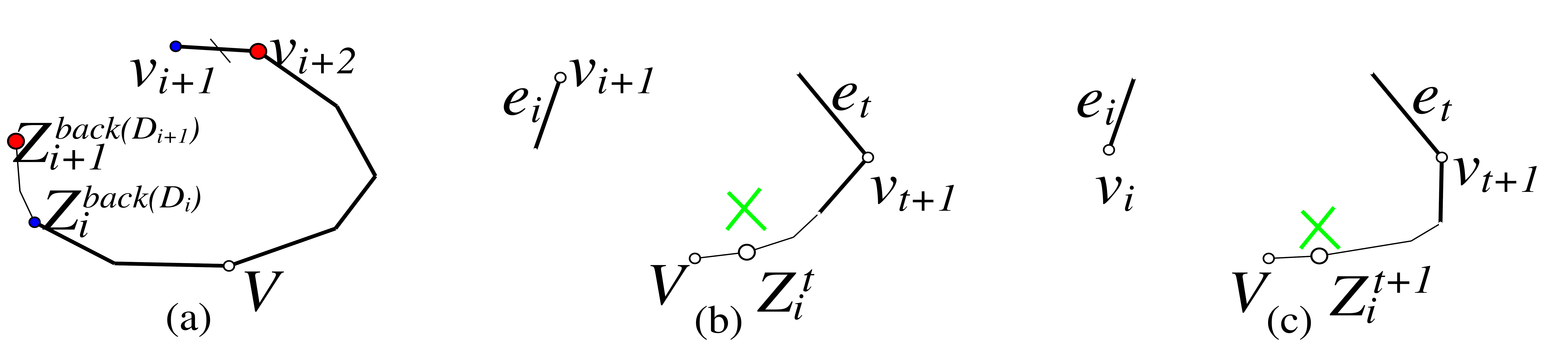}
\caption{Illustrations of the proofs of Fact~\ref{fact:UZ-monotone} and Fact~\ref{fact:st-detailed-facts}.}\label{fig:st-properties}
\end{figure}

\subsubsection{Definition of $\MID$ and a structural property of $\Lambda_V$}\label{subsect:MID-Lambda}

Recall set $\Delta_V$, region $\MID$ and its two boundaries $\LV,\RV$ defined in section~\ref{sect:techover} (see Figure~\ref{fig:mid_V}).

\begin{lemma}\label{lemma:MID-zeta} $(u,u')\in \Lambda_V$ if and only if
    \quad $(u,u')\in\Delta_V, u\oplus u'\subseteq \MID, \text{ and $u$ is chasing $u'$}$.
\end{lemma}
\textbf{Note:} Although $e_s\preceq e_t$, set $\Delta_V$ sometimes may contain unit pair $(u,u')$ such that $u$ is not chasing $u'$.
For example, $(v_s,v_{t+1})\in \Delta_V$, but it is possible that $v_s$ is not chasing $e_{t+1}$.

\smallskip Before proving Lemma~\ref{lemma:MID-zeta}, we state some additional observations of $s,t$ in Fact~\ref{fact:st-detailed-facts}.

\begin{fact}\label{fact:st-detailed-facts} Denote $\rho =[v_{t+1}\circlearrowright v_s]$.
\begin{enumerate}
\item[1.] For $e_i$ in $[v_s\circlearrowright v_{t+1}]$ and $e_i\neq e_t$, when $Z_i^t <_\rho V$, we have $e_i\prec e_{t+1}$.
\item[2.] For $e_i$ in $[v_s\circlearrowright v_{t+1}]$, when $e_i\prec e_{t+1}$, we have $Z_i^{t+1}\in (V\circlearrowright v_i)$.
\item[3.] For $e_j$ in $[v_s\circlearrowright v_{t+1}]$ and $e_j\neq e_s$, when $Z_s^j >_\rho V$, we have $e_{s-1}\prec e_j$.
\item[4.] For $e_j$ in $[v_s\circlearrowright v_{t+1}]$, when $e_{s-1}\prec e_j$, we have $Z_{s-1}^j\in (v_{j+1}\circlearrowright V)$.
\end{enumerate}
\end{fact}

We only show the proof of parts~1 and 2. Parts~3 and 4 are symmetric to parts~1 and 2, respectively.

\begin{proof}[Proof of part~1.] For a contradiction, suppose that $e_i\nprec e_{t+1}$, as shown in Figure~\ref{fig:st-properties}~(b).
This means $\D_i=v_{t+1}$, which further implies $\UZ^+_i=[v_{i+1}\circlearrowright Z_i^t]$.
Combining this equation with the assumption $Z_i^t <_\rho V$, we see $\UZ^+_i$ does not contain $V$,
Since $e_i$ lies in $[v_s\circlearrowright v_{t+1}]$, we know $e_s\leq_V e_i$ and so $\UZ^+_i$ contains $V$ by Fact~\ref{fact:UZ-monotone}. Contradictory.
\end{proof}

\begin{proof}[Proof of part~2.] For a contradiction, suppose that $Z_i^{t+1}$ does not lie in $(V\circlearrowright v_i)$.
Then, it must lie in $[v_{t+1} \circlearrowright V]$ according to Fact~\ref{fact:dist-unique-location}, as shown in Figure~\ref{fig:st-properties}~(c).
So $[Z_i^{t+1}\circlearrowright v_t]$ contains $V$.
Further since $[Z_i^{t+1}\circlearrowright v_t]\subseteq \UZ^-_{t+1}$,
  we get $V\in \UZ^-_{t+1}$.
This means $e_t$ is not the largest edge such that $\UZ^-_{t}$ contains $V$, contradicting the definition of $e_t$.
\end{proof}

\begin{proof}[Proof of Lemma~\ref{lemma:MID-zeta}]
According to Fact~\ref{fact:st2}, $\Lambda_V\subseteq \Delta_V$.
Thus it reduces to proving that for each $(u,u')\in\Delta_V$ such that $u$ is chasing $u'$,
\begin{equation}\label{eqn:zeta-MID}
\zeta(u,u')\text{ contains }V \quad \Leftrightarrow \quad u\oplus u' \subseteq \MID.
\end{equation}

 First, consider the trivial case $s=t$, where $(v_s,v_{t+1})$ is the only element in $\Delta_V$ so that $u$ is chasing $u'$
 ($v_i$ is always chasing $v_{i+1}$.) By definition, $\MID$ equals $v_s\oplus v_{t+1}$. It remains to show that $\zeta(v_s,v_{t+1})$ contains $V$.
By Fact~\ref{fact:st-detailed-facts}.2 and \ref{fact:st-detailed-facts}.4,
$Z_{s}^{t+1}\in (V \circlearrowright v_s)$ whereas $Z_{s-1}^t\in (v_{t+1}\circlearrowright V)$.
So, $V\in [Z_{s-1}^t \circlearrowright Z_{s}^{t+1}]$, i.e.\ $V\in \zeta(v_s,v_{t+1})$.

\begin{figure}[h]
\centering\includegraphics[width=\textwidth]{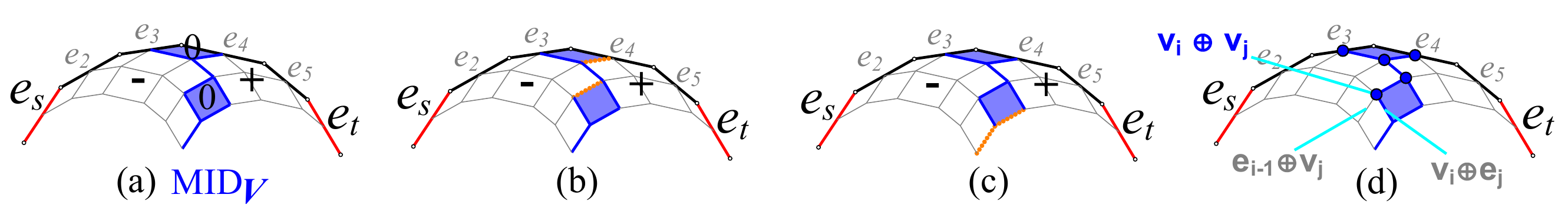}
\caption{Illustration of Statement~(\ref{eqn:zeta-MID}).}\label{fig:MID-zeta}
\end{figure}

Assume now $s\neq t$. Let $\rho=[v_{t+1} \circlearrowright v_s]$. Take an arbitrary unit pair $(u,u')\in \Delta_V$ such that $u$ is chasing $u'$.
\begin{itemize}
\item[Case~1:] \emph{$u,u'$ are both edges}.
    By definition of $\MID$, $u\oplus u'\subseteq \MID$ $\Leftrightarrow$ [$u\oplus u'$ is marked by `0'] $\Leftrightarrow$ $Z_u^{u'}=V$.
    By definition of $\zeta(u,u')$ in (\ref{eqn:zeta_chasing}), $V\in \zeta(u,u')$ $\Leftrightarrow$ $Z_u^{u'}=V$. Together, (\ref{eqn:zeta-MID}) holds.

\item[Case~2.1:] \emph{$(u,u')=(e_i,v_j)$ where $j\neq t+1$}. See the dotted segments in Figure~\ref{fig:MID-zeta}~(b).
    We have $u\oplus u' \subseteq \MID$ $\Leftrightarrow$
       [$e_i\oplus e_{j-1}$ is marked by `0/-' whereas $e_i\oplus e_j$ is marked by `0/+'] $\Leftrightarrow$
       [$Z_i^{j-1}\leq_\rho V\leq_\rho Z_i^j$] $\Leftrightarrow$
       $V\in \zeta(e_i,v_j)$. In particular, the first ``$\Leftrightarrow$'' is by the definition of $\MID$,
         the last ``$\Leftrightarrow$'' is by the definition of $\zeta(u,u')$.

\item[Case~2.2:] \emph{$(u,u')=(e_i,v_{t+1})$}. See the dotted segments in Figure~\ref{fig:MID-zeta}~(c).
    We have $u\oplus u'\subseteq \MID$ $\Leftrightarrow$ [$e_i\oplus e_t$ is marked by `0/-'] $\Leftrightarrow$ $Z_i^t \leq_\rho V$ $\Leftrightarrow$ $V\in \zeta(u,u')$.
    Here, the last ``$\Leftrightarrow$'' is non-obvious and is proved in the following.
    Since $u=e_i$ is chasing $u'=v_{t+1}$, we know $e_i\prec e_{t+1}$.
    Therefore, by Fact~\ref{fact:st-detailed-facts}.2, $Z_i^{t+1}\in (V\circlearrowright v_i)$.
        This implies that $Z_i^t \leq_\rho V$ $\Leftrightarrow$ $V\in [Z_i^t \circlearrowright Z_i^{t+1}]$.
        In other words, $Z_i^t \leq_\rho V$ $\Leftrightarrow$ $V\in \zeta(u,u')$.

\item[Case~2.3:] $u$ is a vertex and $u'$ is an edge. This is symmetric to Case~2.1 or Case~2.2.

\item[Case~3.1:] $(u,u')=(v_s,v_{t+1})$. (This does not necessarily occur since $v_s$ may not be chasing $v_{t+1}$.)\\
    Since $u$ is chasing $u'$, we have $e_{s-1}\prec e_t$ and $e_s\prec e_{t+1}$.
    By Fact~\ref{fact:st-detailed-facts}.2 and \ref{fact:st-detailed-facts}.4,
        $Z_{s-1}^t$ lies in $(v_{t+1}\circlearrowright V)$, whereas $Z_{s}^{t+1}$ lies in $(V \circlearrowright v_s)$.
    Therefore, $V$ lies in $[Z_{s-1}^t \circlearrowright Z_{s}^{t+1}]= \zeta(v_s,v_{t+1})$.
    Moreover, $v_s\oplus v_{t+1}$ must be contained in $\MID$. Thus (\ref{eqn:zeta-MID}) holds.

\item[Case~3.2:] $(u,u')=(v_i,v_j)$ where $i\neq s$ and $j\neq t+1$.
    See the dots in Figure~\ref{fig:MID-zeta}~(d) for illustrations.
    We have $v_i\oplus v_j\subseteq \MID$ $\Leftrightarrow$ [$e_{i-1}\oplus v_j\subseteq \MID$ or $v_i\oplus e_j\subseteq \MID$] $\Leftrightarrow$
        [$V\in\zeta(e_{i-1},v_j)$ or $V\in\zeta(v_i,e_j)$] $\Leftrightarrow$ [$V\in \zeta(v_i,v_j)$].
    The last second ``$\Leftrightarrow$'' applies the results in Case~2.
    The last ``$\Leftrightarrow$'' applies the fact that $\zeta(v_i,v_j)$ is the concatenation of $\zeta(e_{i-1},v_j)$ and $\zeta(v_i,e_j)$.

\item[Case~3.3:] [$u=v_s$ and $u'$ is a vertex other than $v_{t+1}$] or [$u'=v_{t+1}$ and $u$ is a vertex other than $v_s$].
        The proof of this case is similar to those of Case~3.1 and Case~3.2 and is omitted.
\end{itemize}
\end{proof}

\subsubsection{Proof of the enhanced version of \SC}\label{subsect:singlesector_finalproof}

\newcommand{\HS}{\frac{1}{2}\sector(V)}
For convenience, denote
\begin{equation}\label{eqn:HS-simple}
\HS= \text{$\frac{1}{2}$-scaling of $\sector(V)$ with respect to $V$} = {\bigcup}_{(u,u')\in \Lambda_V} u\oplus u'.
\end{equation}

\begin{lemma}\label{lemma:MID-connect-HS}
Let $\epsilon_V = \text{the union of }\{u\oplus u'\mid (u,u')\in \Delta_V, u\hbox{ is \textbf{not} chasing }u'\}$, then
\begin{equation}\label{eqn:HS-MID-epsilon}
\HS= \MID - \epsilon_V.
\end{equation}
Moreover, $\MID$ is the closed set of $\HS$. So $\MIDSCALE$ is the closed set of $\sector(V)$.
\end{lemma}

\begin{proof}
Regions $\HS,\MID,\epsilon_V$ are all unions of some regions in $\{u\oplus u'\mid (u,u')\in \Delta_V\}$.
  Obviously, the regions in $\{u\oplus u'\mid (u,u')\in \Delta_V\}$ are pairwise-disjoint.
Therefore, proving (\ref{eqn:HS-MID-epsilon}) reduces to proving that for any $(u,u')\in \Delta_V$,
\[
\begin{gathered}
u\oplus u'\subseteq \HS \Leftrightarrow [u\oplus u'\subseteq \MID\text{ and }u\oplus u'\nsubseteq \epsilon_V]; \text{equivalently,}\\
[(u,u')\in \Lambda_V] \Leftrightarrow [\text{$u$ is chasing $u'$ and $u\oplus u'\subseteq \MID$}],
\end{gathered}
\]
which holds according to Lemma~\ref{lemma:MID-zeta}.

\medskip Next, we show that $\MID$ is the closed set of $\HS$.

For simplification, assume that $s\neq t$; the case $s=t$ is trivial.
Denote
\[
\begin{gathered}
\epsilon^{(1)}_V= \text{the union of }\{u\oplus u'\mid (u,u')\in \Delta_V, u\hbox{ is not chasing }u', \text{ and }u=v_s\}, \\
\epsilon^{(2)}_V= \text{the union of }\{u\oplus u'\mid (u,u')\in \Delta_V, u\hbox{ is not chasing }u', \text{ and }u'=v_{t+1}\}.
\end{gathered}
\]

Because $e_s\preceq e_t$, when $(u,u')\in \Delta_V$ and $u\hbox{ is not chasing }u'$, either $u=v_s$ or $u'=v_{t+1}$.
Therefore, $\epsilon_V = \epsilon^{(1)}_V \cup \epsilon^{(2)}_V$.
In addition, we point out the following obvious facts.
\begin{itemize}
\item[(I)] $\epsilon^{(1)}_V\subseteq \alpha$, where $\alpha$ denotes the unique route that terminates at the midpoint of $e_s$.
\item[(II)] $\epsilon^{(2)}_V\subseteq \beta$, where $\beta$ denotes the unique route that terminates at the midpoint of $e_t$
\end{itemize}
\begin{figure}[h]
\centering\includegraphics[width=.75\textwidth]{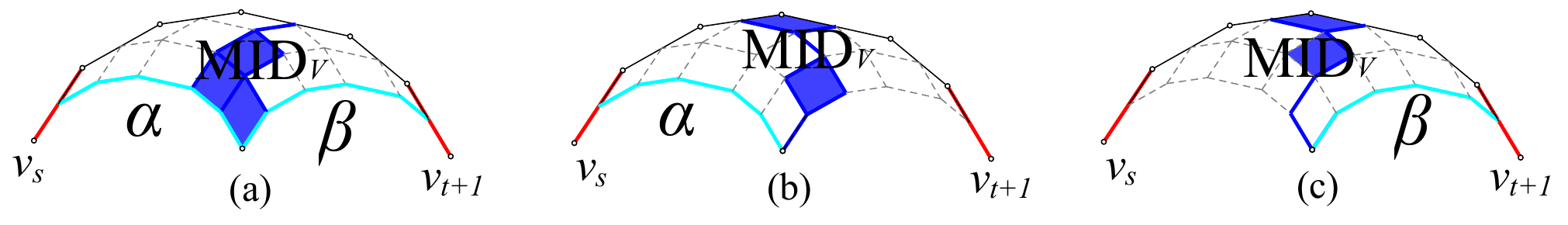}
\caption{$\MID$ is the closed set of $\HS$.}\label{fig:epsilon}
\end{figure}

Next, we discuss three different cases.
\begin{itemize}
\item[Case~1:] $Z_s^t=V$. See Figure~\ref{fig:epsilon}~(a). In this case, by the definition of $\MID$, for any subregion $R$ of $\alpha\cup \beta$, the closed set of $\MID-R$ equals $\MID$.
  In particular, $\epsilon_V$ is a subset of $\alpha\cup\beta$ by (I) and (II), so the closed set of $\MID-\epsilon_V$ (i.e.\ the closed set of $\HS$) is $\MID$.
\item[Case~2:] $Z_s^t<_\rho V$.
    By Fact~\ref{fact:st-detailed-facts}.1, $e_s\prec e_{t+1}$.
    So every unit in $[v_s\circlearrowright v_{t-1}]$ beside $v_s$ is chasing $v_{t+1}$. ($v_s$ may be chasing or not.)
    So, $\epsilon^{(2)}_V\subseteq \epsilon^{(1)}_V$.
    So, $\epsilon_V=\epsilon^{(1)}_V\cup \epsilon^{(2)}_V=\epsilon^{(1)}_V\subseteq \alpha$.
    Therefore, the closed set of $\MID-\epsilon_V$ (i.e.\ the closed set of $\HS$) is $\MID$. See Figure~\ref{fig:epsilon}~(b).
\item[Case~3:] $V<_\rho Z_s^t$. This case is symmetric to Case~2. See Figure~\ref{fig:epsilon}~(c). We omit its proof.
\end{itemize}
\end{proof}

Recall that the 2-scaling of $\LV,\RV,\MID$ with respect to $V$ are respectively defined to be $\LVS,\RVS,\MIDSCALE$.
In Figure~\ref{fig:sectors-nestp-sigma}, the red and blue curves indicate
    $\mathcal{L}^\star_{v_1},\ldots,\mathcal{L}^\star_{v_n}$ and
    $\mathcal{R}^\star_{v_1},\ldots,\mathcal{R}^\star_{v_n}$ respectively.

Since $\LV,\RV$ are boundaries of $\MID$, curves $\LVS,\RVS$ are boundaries of $\MIDSCALE$.
Further by Lemma~\ref{lemma:MID-connect-HS}, $\LVS,\RVS$ are boundaries of $\sector(V)$.
To prove the \textsc{Sector-continuity}, we prove the following enhanced lemma.

\begin{lemma}\label{lemma:sector-continuity}
If the common starting point of $\LVS,\RVS$ lies in $P$,
then, $\LVS$ has a unique intersecting point with $\partial P$ and so does $\RVS$
and $\sector(V)\cap \partial P$ is a boundary-portion from $\LVS\cap \partial P$ to $\RVS\cap \partial P$.
(This does not mean $\sector(V)\cap \partial P=[\LVS\cap \partial P\circlearrowright \RVS\cap \partial P]$; endpoints may not be contained.) Otherwise, $\sector(V)\cap \partial P$ is empty.
\end{lemma}

\begin{proof}
Recall the roads and routes in section~\ref{sect:techover}.
For each road or route, call its 2-scaling with respect to $V$ a \emph{scaled-road} or \emph{scaled-route}.
Assume that each scaled-road (or scaled-route) has the same direction as its corresponding unscaled road (or route).
We have two observations:
\begin{itemize}
\item[(i)] The $2$-scaling of $[v_s \circlearrowright v_{t+1}]$ with respect to $V$ lies in the exterior of $P$.
\item[(ii)] If we travel along a given scaled-route, we eventually get outside $P$ and never return to $P$ since then.
Therefore, there is exactly one intersection between this scaled-route and $\partial P$ if its starting point lies inside $P$; and no intersection otherwise.
\end{itemize}

\noindent \emph{Proof of (i):} This observation follows from the relation $V\notin [v_s \circlearrowright v_{t+1}]$ stated in Fact~\ref{fact:st1}.

\medskip \noindent \emph{Proof of (ii):} Because all of the routes terminate at $[v_s \circlearrowright v_{t+1}]$, the scaled-routes terminate at the $2$-scaling of $[v_s \circlearrowright v_{t+1}]$ with respect to $V$.
Further by observation~(i), the scaled-routes terminate at the exterior of $P$.
Therefore, we will eventually get outside $P$ traveling along any scaled-route.
Moreover, any road $e_i\oplus v_j$ where $(e_i,v_j)\in
\Delta_V$ has a property that we do not return to $P$ from outside $P$ traveling along the 2-scaling of $e_i\oplus v_j$. This follows from
    observations~(ii.1) and (ii.2) below.
The roads in $\{v_i\oplus e_j\mid (v_i,e_j)\in \Delta_V\}$ also have this property due to similar reasons.
Further since the scaled-routes consist of these scaled-roads, we obtain observation~(ii).
\begin{itemize}
\item[(ii.1)] The 2-scaling of $e_i\oplus v_j$ with respect to $V$ is a translation of $e_i$ that lies on the right of $\overrightarrow{v_{t+1}v_i}$.
        (Regard that the translation of $e_i$ has the same direction as $e_i$.)
\item[(ii.2)] When we travel along a translation of $e_i$ that lies on the right of $\overrightarrow{v_{t+1}v_i}$, we will not go into $P$ from outside $P$.
\end{itemize}
\noindent \emph{Proof of (ii.1)}: $e_i\oplus v_j$ lies on the right of $\overrightarrow{v_{t+1}v_i}$, whereas $V$ lies on its left; thus we get observation~(ii.1).

\medskip \noindent \emph{Proof of (ii.2):} Since $e_s\preceq e_t$ and $(e_i,v_j)\in \Delta_V$, we get $e_i\prec e_t$, which implies observation~(ii.2).

\bigskip Let $S^\star_V$ denote the common starting point of all scaled-routes (including $\LVS$ and $\RVS$).
    Equivalently, $S^\star_V$ is the 2-scaling of $v_s\oplus v_{t+1}$ with respect to $V$. The following observation follows from observations~(i) and (ii).\smallskip
\begin{enumerate}
\item[(iii)] If $S^\star_V$ lies in $P$, then $\MIDSCALE\cap \partial P=[\mathcal{L}^\star_V\cap \partial P\circlearrowright \mathcal{R}^\star_V\cap \partial P]$;
    otherwise $\MIDSCALE\cap \partial P$ is empty.\smallskip
\end{enumerate}
\noindent \emph{Proof of (iii):} When $S^\star_V$ lies outside $P$, by (i) and (ii), all the boundaries that bound $\MIDSCALE$, including $\LVS,\RVS$ and a fraction of the 2-scaling of $[v_s\circlearrowright v_{t+1}]$ with respect to $V$, lie in the exterior of $P$.
Therefore $\MIDSCALE$ lies in the exterior of $P$, which implies that $\MIDSCALE\cap \partial P$ is empty.
When $S^\star_V$ lies in $P$, the boundaries of $\MIDSCALE$ have exactly two intersections with $\partial P$.
    Therefore, $\MIDSCALE \cap \partial P$ either equals $[\mathcal{L}^\star_V \cap \partial P\circlearrowright \mathcal{R}^\star_V\cap \partial P]$,
    or equals $[\mathcal{R}^\star_V \cap \partial P\circlearrowright \mathcal{L}^\star_V\cap \partial P]$.
    We argue that it does not equal the latter one.
    Notice that region $\MIDSCALE$ is always on our right side when we travel along $\mathcal{L}^\star_V$.
    This implies that $\MIDSCALE \cap \partial P\neq [\mathcal{R}^\star_V \cap \partial P\circlearrowright \mathcal{L}^\star_V\cap \partial P]$.

\medskip At last, we can see this lemma easily follows from observation~(iii) together with Lemma~\ref{lemma:MID-connect-HS}.
\end{proof}

\section{Computation -- answer the location queries mentioned in Theorem~\ref{thm:nestp-location}}\label{sect:location-answers}

\paragraph*{Outline.} Assume $V$ is a vertex.
\begin{enumerate}
\item Compute the two endpoints of $\sector(V)\cap \partial P$.  (see subsection~\ref{subsect:algorithms-endpoints})
\item Compute the (set of) units intersected by $\sector(V)$.  (see subsection~\ref{subsect:alg_S})
\item Compute $w$ such that $\sector(W)$ contains $V$.   (also see subsection~\ref{subsect:alg_S})
\item Compute $u,u'$ such that $\block(u,u')$ contains $V$. (see the last two subsections)
\end{enumerate}

\textbf{Hint.} Parts~1 and 4 (especially, part~4) are more nontrivial than parts~2 and 3.
We suggest that the reader skipping subsection~\ref{subsect:alg_S} for the first read.
(Moreover, notice that part~1 and part~4 are highly \textbf{symmetric}; see Remark~\ref{remark:symmetric}.)

\begin{lemma}[Lemma~7 of \cite{arxiv:n2}; Computational aspect of the $Z$-points]\label{lemma:Z-compute}~
\begin{enumerate}
\item For $(e_i,e_j)$ in which $e_i\prec e_j$,
    point $Z_i^j$ can be computed in $O(1)$ time given the unit containing this point.
\item Given $i,j,k$ such that $e_i\prec e_j$ and $v_k\in(v_{j+1} \circlearrowright v_i)$.
     There are three possibilities of $Z_i^j$ (due to Fact~\ref{fact:dist-unique-location}):
    (i) it equals $v_k$; (ii) it lies in $(v_{j+1} \circlearrowright v_k)$; or (iii) it lies in $(v_k \circlearrowright v_i)$.
    We can distinguish them in $O(1)$ time.
\item Given edge pairs $(a_1,b_1),\ldots,(a_m,b_m)$ such that $a_i\prec b_i~(1\leq i\leq m)$ and that $a_1,\ldots,a_m$ lie in clockwise order and $b_1,\ldots,b_m$ lie in clockwise order, we can compute $Z_{a_1}^{b_1},\ldots,Z_{a_m}^{b_m}$ altogether in $O(m+n)$ time.
\end{enumerate}
\end{lemma}

\subsection{Compute the two endpoints of $\sector(V)\cap \partial P$}\label{subsect:algorithms-endpoints}

Assume the reader is familiar with the notations and results given in subsection~\ref{sect:sector-continue}.
Especially, recall $\leq_V$, smallest and largest (with respect to $\leq_V$), $s$ and $t$, and the marks `-/+/0' of the regions $u\oplus u' \mid (u,u')\in \Delta_V$.

By Lemma~\ref{lemma:sector-continuity}, computing the endpoints of $\sector(V)\cap \partial P$ means computing $\LVS\cap \partial P$ and $\RVS\cap \partial P$.
In the following we show how do we compute $\LVS \cap \partial P$; computing $\RVS\cap \partial P$ is symmetric and is omitted.

Briefly, we shall find the segment piece of $\LVS$ that intersects $\partial P$ by a binary search, as outlined in section~\ref{sect:techover}.

\paragraph*{An explicit definition for $\LV$.}
To describe our algorithm, we need an explicit definition of $\LV$ as follows.

\begin{figure}[h]
\centering \includegraphics[width=.49\textwidth]{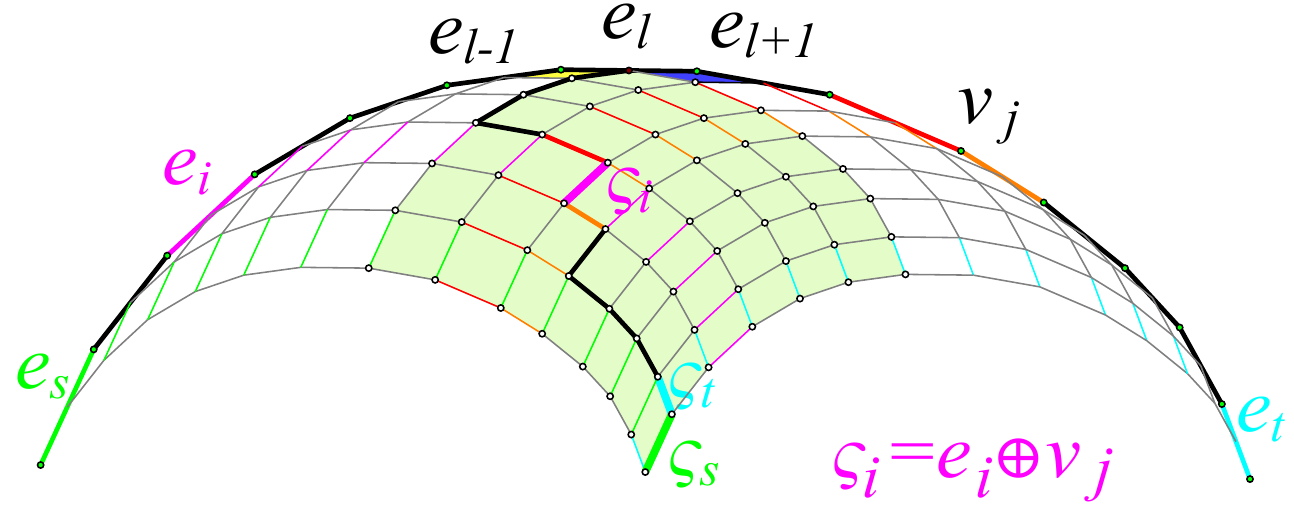}
\caption{Notations used in the algorithm for computing $\LVS\cap \partial P$.}\label{fig:sectorV_algorithm}
\end{figure}

We first introduce an edge $e_l$. See Figure~\ref{fig:sectorV_algorithm}. It is defined as the unique edge in $[v_s\circlearrowright v_{t+1}]$ such that
\begin{itemize}
\item[I] \emph{For $e_i$ such that $e_s\leq_V e_i \leq_V e_{l-1}$, region $e_i\oplus e_{i+1}$ is marked by `-'.}
\item[II] \emph{For $e_i$ such that $e_l\leq_V e_i \leq_V e_{t-1}$, region $e_i\oplus e_{i+1}$ is marked by `+/0'.}
\end{itemize}
Since $\LV$ divides the regions marked by `-' from those marked by `+/0', it terminates at the midpoint of $e_l$.\clearpage

We then introduce two types of roads.
\begin{description}
\item[$A$-type roads.] For any edge $e_i$ in $[v_s\circlearrowright v_l]$,
  let $e_j$ denote the smallest edge in $[v_{l+1}\circlearrowright v_{t+1}]$ such that $e_i\oplus e_j$ is marked by `0/+' (or denote $e_j=e_{t+1}$ if no such edge exists); define $\varsigma_i=e_i\oplus v_j$ and call it an \emph{$A$-type road}.
\item[$B$-type roads.] For any edge $e_j$ in $[v_{l+1}\circlearrowright v_{t+1}]$,
  let $e_i$ denote the smallest edge in $[v_s\circlearrowright v_l]$ such that $e_i\oplus e_j$ is marked by `0/+' (or denote $e_i=e_l$ if no such edge exists);
  define $\varsigma_j=v_i\oplus e_j$ and call it a \emph{$B$-type road}.
\end{description}
Explicitly, $\LV$ can be defined as the route that consists of all the $A$-type and $B$-type roads.

\medskip The following observations of these A-type and B-type roads are obvious.
\begin{itemize}
\item[A] \emph{The order of the $A$-type roads in $\LV$ is determined, and equals to $\varsigma_s,\varsigma_{s+1},\ldots,\varsigma_{l-1}$}.
\item[B] \emph{The order of the $B$-type roads in $\LV$ is determined, and equals to $\varsigma_{t},\varsigma_{t-1},\ldots,\varsigma_{l+1}$.}
\end{itemize}

\begin{lemma}\label{lemma:compute_route_endpoints}
\begin{enumerate}
\item We can compute $s,t,l$ in $O(\log n)$ time.
\item When $\varsigma_i$ is defined (in other words, $e_i$ lies in $[v_s\circlearrowright v_{t+1}]$ and $e_i\neq e_l$), we can compute the endpoints of $\varsigma_i$ in $O(\log n)$ time.
    Moreover, we can distinguish the following in $O(\log n)$ time:
    $\varsigma^\star_i$ intersects $\partial P$, or $\varsigma^\star_i$ lies in the interior of $P$, or $\varsigma^\star_i$ lies in the exterior of $P$,
      where $\varsigma^\star_i$ denotes the 2-scaling of $\varsigma_i$ with respect to $V$.
\item Recall that $S^\star_V$ denotes the starting point of $\LVS$. We can compute $S^\star_V$ in $O(1)$ time.
    Moreover, if $S^\star_V$ lies in $P$, we can compute $\mathcal{L}^\star_V\cap \partial P$ in $O(\log^2 n)$ time.
\end{enumerate}
\end{lemma}
\begin{proof}
\noindent 1. First, we show how do we compute $s$; the value of $t$ can be computed symmetrically.

We need two facts. (The first one is Fact~\ref{fact:UZ-monotone}~part~(1).)\smallskip

(i) \emph{For every edge $e_i$, it holds that $e_s\leq_V e_i$ if and only if $\UZ^+_i$ contains $V$.}\smallskip

(ii) \emph{Given an edge $e_i$, we can determine whether $\UZ^+_i$ contains $V$ in $O(1)$ time.}\smallskip

\noindent\emph{Proof of (ii).} To determine whether $\UZ^+_i$ contains $V$ is to determine the relation between $Z_i^j$ and $V$, where $e_j=back(\D_i)$, which can be determined in $O(1)$ time by Lemma~\ref{lemma:Z-compute}~part~2.\smallskip

Applying facts~(i) and (ii), $s$ can be determined in $O(\log n)$ time by a binary search.

\smallskip Next, we show how to compute $l$. By Lemma~\ref{lemma:Z-compute}, we can determine whether $e_i\oplus e_{i+1}$ is marked by `-', `0', or `+' in $O(1)$ time. So, based on properties I and II stated above, we can compute $l$ in $O(\log n)$ time by a binary search.

\medskip \noindent 2. First, we show how do we compute road $\varsigma_i$. Assume that $e_s\leq_V e_i\leq_V e_{l-1}$; otherwise $e_{l+1}\leq_V e_i\leq_V e_t$ and it is symmetric.
It reduces to finding the edge $e_j$ defined in the paragraph ``$A$-type roads'' above.
According to the bi-monotonicity of the $Z$-points (Fact~\ref{fact:Z_bi-monotonicity}), $e_j$ is the unique edge such that $e_i\oplus e_{j-1}$ is marked by `-' while $e_i\oplus e_j$ is marked by `+/0'.
As we can compute each mark in $O(1)$ time, we can search $j$ in $O(\log n)$ time.

After computing $\varsigma_i$, we easily obtain $\varsigma^\star_i$.
  We then distinguish the relation between $\varsigma^\star_i$ and $\partial P$.
First, compute whether the endpoints of $\varsigma^\star_i$ lie in $P$, which takes only $O(\log n)$ time because $P$ is convex.
If both the endpoints lie in $P$, segment $\varsigma^\star_i$ lies in $P$;
if both lie outside $P$, segment $\varsigma^\star_i$ lies outside $P$;
otherwise, $\varsigma^\star_i$ intersects $\partial P$.

\medskip \noindent 3. Finally, we show how do we compute the (potential) intersecting point $\LVS\cap \partial P$.

Since $\LV$ starts at $(v_s+v_{t+1})/2$,
  point $S^\star_V$ is at the 2-scaling of $(v_s+v_{t+1})/2$ with respect to $V$, which can be computed in $O(1)$ time.
Now, assume $S^\star_V$ lies in $P$, so $\LVS$ has one intersection with $\partial P$.

\smallskip We design two \emph{subroutines}: one assumes that there is an A-type road whose 2-scaling (with respect to $V$) intersects $\partial P$,
    and it seeks for this road.
  The other is symmetric. It assumes there is a B-type road whose 2-scaling (with respect to $V$) intersects $\partial P$ and seeks for this road.
    Clearly, one subroutine would success.

According to observations A and B stated above, the A-type roads $\varsigma_s,\varsigma_{s+1},\ldots,\varsigma_{l-1}$ are in order on $\LV$;
  so do the B-type roads. So, a binary search can be applied in designing our two subroutines.
Each searching step costs $O(\log n)$ time due to part~2 of this lemma; hence the total running time is $O(\log^2n)$.
\end{proof}

\subsection{Which units does $\sector(V)$ intersect \& which sector does $V$ lie in?}\label{subsect:alg_S}

Assuming the endpoints of $\sector(V)\cap \partial P$ are known for each vertex $V$,
  we now proceed to compute the interval of units that intersect $\sector(V)$ and the (unique) sector that contains $V$ for each vertex $V$.

\paragraph*{Compute the units that intersect $\sector(V)$.}
Let $u_L=\unit(\LVS\cap \partial P)$ and $u_R=\unit(\RVS\cap \partial P)$.
They can be computed while we compute the two endpoints of $\sector(V)\cap\partial P$.
In most cases the units that intersect $\sector(V)$ are the units (in clockwise) from $u_L$ to $u_R$.
  Exceptional cases are discussed in the following note.

\begin{note} Sometimes an endpoint of $\sector(V)\cap\partial P$ is not contained in the sector.
This is because $\sector(V)$ is not always a closed set (see Lemmas~\ref{lemma:MID-connect-HS} and \ref{lemma:sector-continuity}).
Under a degenerate case, this endpoint may happen to lie at a vertex $V^*$ of $P$,
and then, by definition, we could not include $V^*$ into the set of units that intersect $\sector(V)$.
\end{note}

\subsubsection*{Compute the sector that contains $V$ for each vertex $V$ by a \textbf{sweeping algorithm}.}

\newcommand{\Cur}{\mathrm{Current}}
\newcommand{\Fur}{\mathrm{Future}}

Recall the event-points and their two types of tags mentioned in section~\ref{sect:techover}.
We have two groups of \emph{event-points}.
One group contains the points in $\{\LVS\cap \partial P,\RVS\cap \partial P\}$;
and the other contains the intersecting points between $\sigma P$ and $\partial P$, namely, the $K$-points (recall the $K$-points in subsection~\ref{subsect:monoto-f}).
Notice that all the event-points lie in $\partial P$.

Below we show how to define our event-points and their tags precisely.
We use two procedures -- an adding and a removing procedure.
The removing procedure removes redundant event-points added in the first procedure.

\begin{figure}[b]
\begin{minipage}{.32\textwidth}
\centering\includegraphics[width=0.98\textwidth]{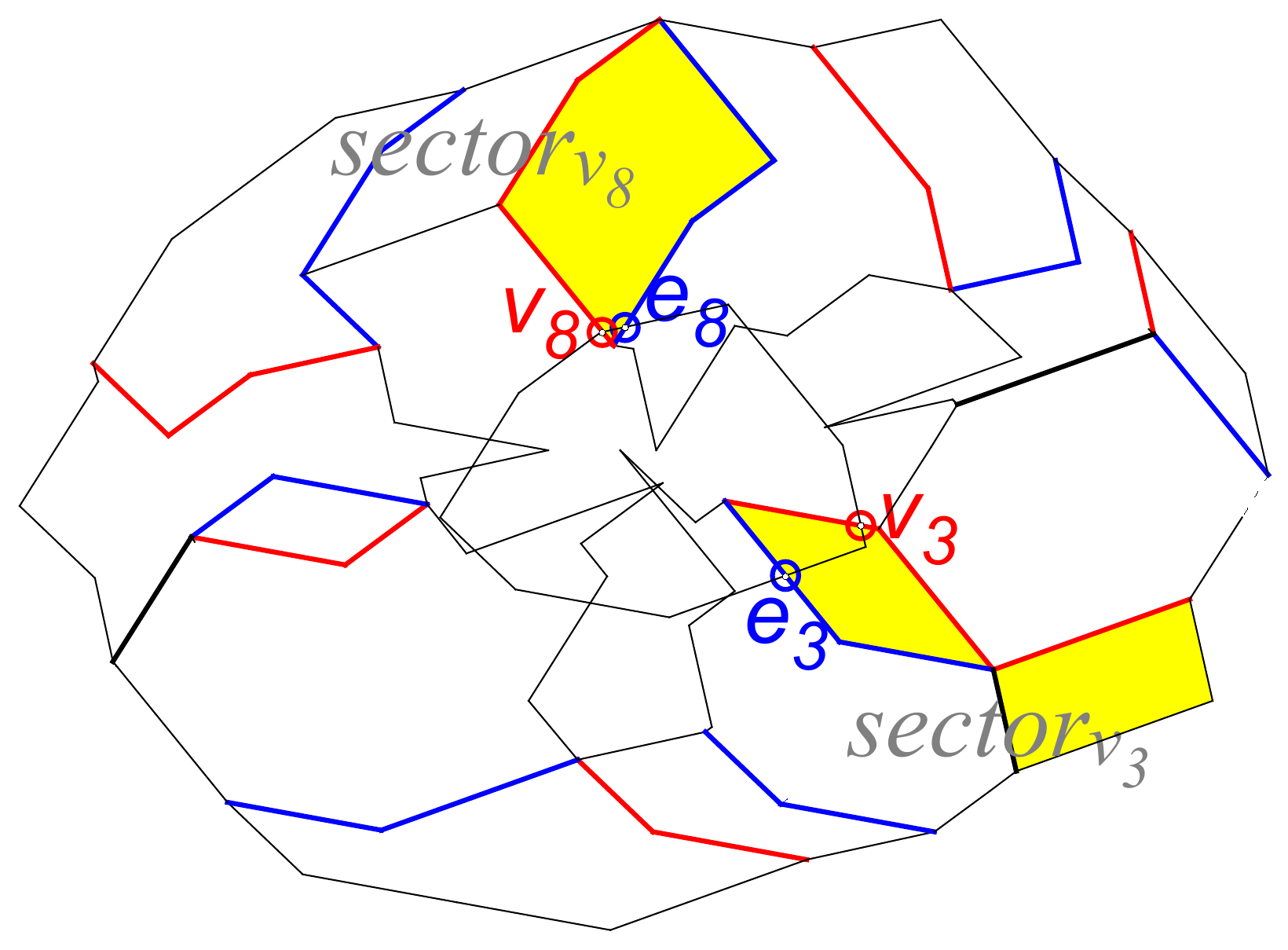}
\end{minipage}
\begin{minipage}{.32\textwidth}
\centering\includegraphics[width=0.98\textwidth]{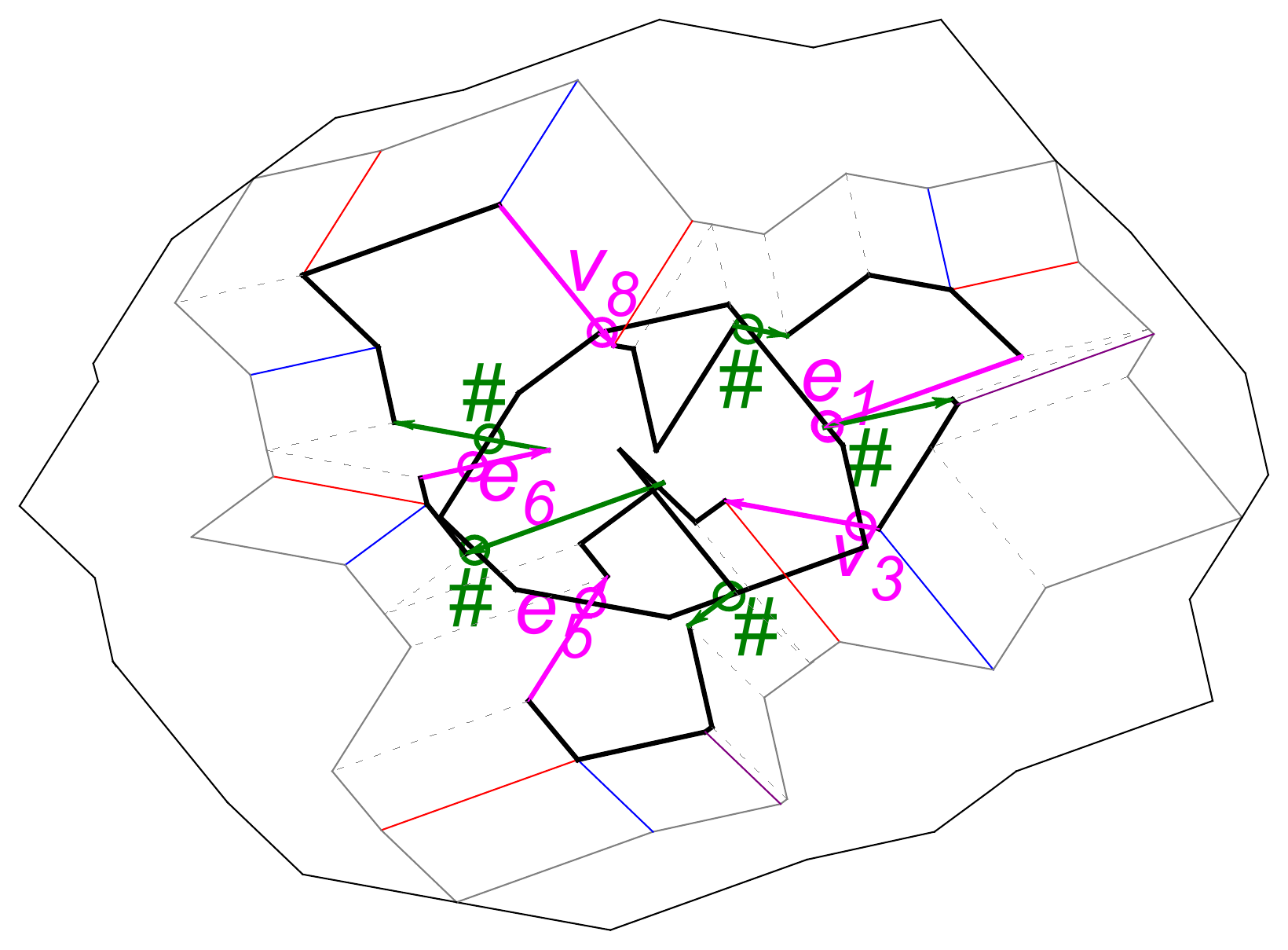}
\end{minipage}
\begin{minipage}{.32\textwidth}
\centering\includegraphics[width=\textwidth]{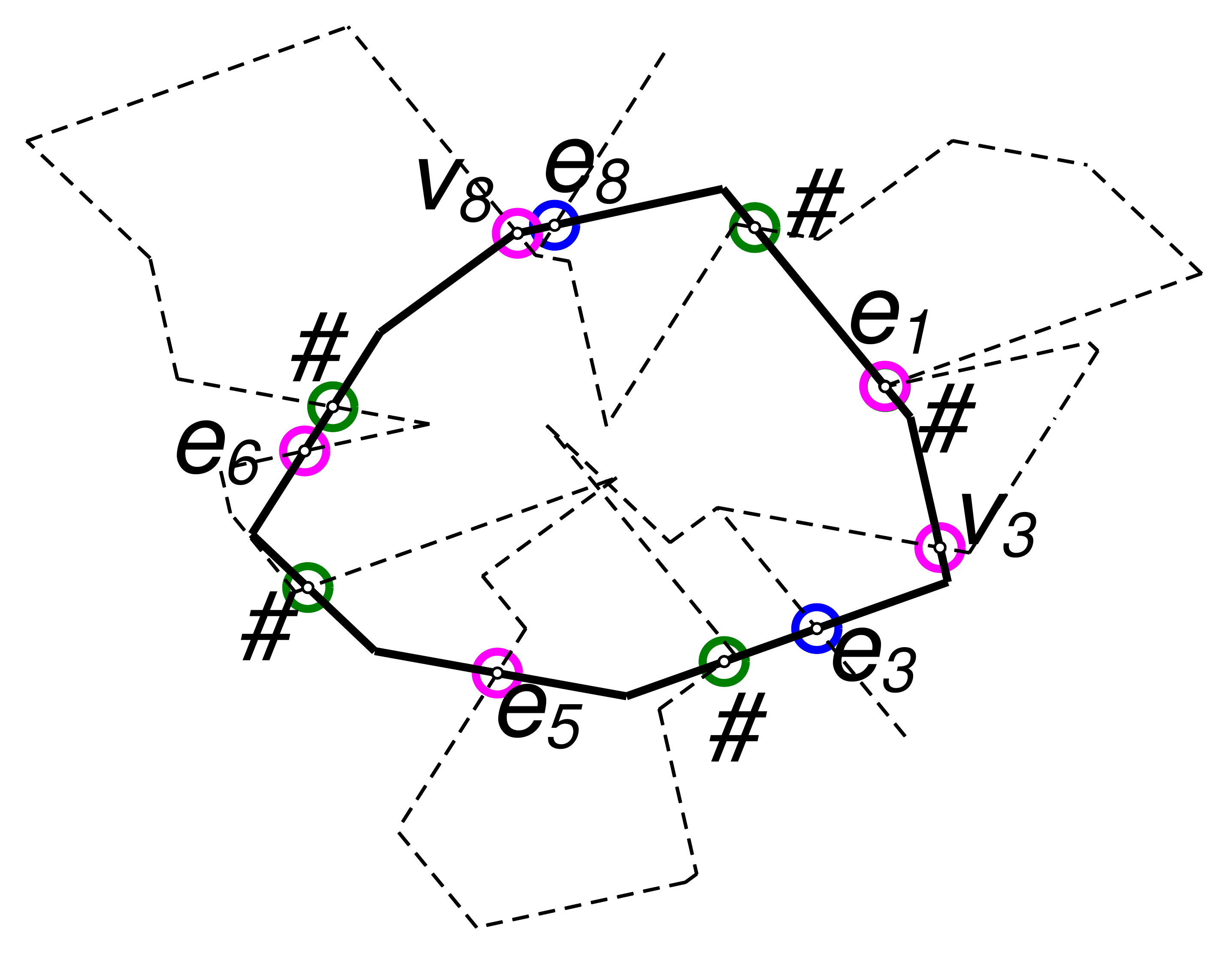}
\end{minipage}
\caption{Definition of the \emph{event-points}. Their \emph{future-tags} are labeled in the figure.}\label{fig:tags}
\end{figure}

\subparagraph{Adding procedure.} See Figure~\ref{fig:tags}. The left picture exhibits the event-points in Group~1 defined below;
    the middle one exhibits the event-points in Group~2 defined below.

\smallskip \noindent\emph{Group~1:} For every vertex $V$ for which $\sector(V)$ intersects $\partial P$, add two event-points $\LVS\cap \partial P$ and $\RVS\cap \partial P$, and define
\begin{equation}\label{def:tag1}
\begin{aligned}
\Cur(\LVS \cap \partial P)&=V, & \Fur(\LVS \cap \partial P)&=V, \\
\Cur(\RVS \cap \partial P)&=V, & \Fur(\RVS \cap \partial P)&=forw(V).
\end{aligned}
\end{equation}

\smallskip \noindent\emph{Group~2:}
For every $K_i\in \sigma P\cap \partial P$, add event-point $K_i$.
Notice that $\sigma P$ consists of several directional line segments.
Assume $K_i$ comes from directional line segment $\overrightarrow{AB}$ of $\sigma P$.
Recall $g$ in Definition~\ref{def:g}. Define
\begin{equation}\label{def:tag2}
\Cur(K_i)=\text{`\#'},
\qquad \Fur(K_i)=\begin{cases}
  \text{`\#'}, & \text{when } A\in P,B\notin P;\\
  \unit(g(K_i)), & \text{when }A\notin P,B\in P.
\end{cases}
\end{equation}

\textbf{Note}: The special symbol `\#' is introduced for indicating the outside of $f(\T)$.
When $\Cur(E)=\text{`\#'}$, no sector contains event-point $E$.
When $\Fur(E)=\text{`\#'}$, no sector contains $(E\circlearrowright E')$, where $E'$ denotes the clockwise next event-point of $E$.
The reason that $\Cur(K_i)$ should be `\#' is explained in Fact~\ref{fact:fK-notdefinedyet}.

\subparagraph{Removing procedure.}
        If there are multiple event-points locating at the same position, we keep only one of them according to the following priority:
        First, $\{\sigma P\cap \partial P\}$.
        Then, $\{\RVS\cap \partial P\}$.
        Last, $\{\LVS\cap \partial P\}$

\medskip As a consequence of the \SM and \IP, we get:
\begin{corollary}\label{corol:S}
Take any point $X$ in $\partial P$. If $X$ lies at some event-point $E$, it belongs to $\sector(\Cur(E))$.
Otherwise, it belongs to $\sector(\Fur(E^*))$, where $E^*$ is the closest event-point preceding $X$ in clockwise order.
\end{corollary}
\textbf{Note}: $X$ belongs to no sector when we say it belongs to $\sector(\text{`\#'})$.

\bigskip Our algorithm is as follows.
1. ADD: Find all the event-points and compute their tags.
2. SORT: Sort these points in clockwise order.
3. REMOVE: Remove the redundant event-points.
4. SWEEP: For each vertex, compute the closest event-point preceding it (in clockwise) and then compute the sector containing it by applying Corollary~\ref{corol:S}.

\medskip The event-points from Group~1 and their tags can be computed efficiently as shown in subsection~\ref{subsect:algorithms-endpoints}.
    In the following we show how to compute the event-points from Group~2 and compute their tags.

\begin{lemma}\label{lemma:sigmaP_O(n)}
The polygonal curve $\sigma P$ consists of $O(n)$ sides and can be computed in $O(n)$ time.
The points in $\sigma P\cap \partial P$ are of size $O(n)$ and can be computed in $O(n\log n)$ time.
Moreover, the future-tag of each such point can be computed in $O(1)$ amortized time.
(The current tags of these event-points are trivially `\#'; see (\ref{def:tag2})).
\end{lemma}

\begin{proof}
Recall frontier-pair-list and bottom borders in subsection~\ref{subsect:pre-borders-sigmap}.
Proving that $\sigma P$ is of size $O(n)$ (namely, the bottom borders have in total $O(n)$ sides) reduces to showing that
(i) the bottom borders of blocks in $\{\block(u,u') \mid (u,u')\in \text{frontier-pair-list, } u,u' \text{ are both edges}\}$ have in total $O(n)$ sides, and
(ii) the bottom borders of blocks in $\{\block(u,u') \mid (u,u')\in \text{frontier-pair-list, at least one of } u,u' \text{ is a vertex}\}$ have in total $O(n)$ sides.

\smallskip \noindent Proof of (i): Clearly, the frontier-pair-list contains $O(n)$ unit pairs, and the bottom border of $\block(u,u')$ has at most two sides when $u,u'$ are both edges.
Therefore, we obtain (i).

\smallskip \noindent Proof of (ii): Let $(u_1,u'_1),\ldots,(u_m,u'_m)$ denote the sublist of the frontier-pair-list that contains all of the edge pairs.
Let $Z_i=Z_{u_i}^{u'_i}$ for short. Combining the following two observations, we obtain (ii).
1. For any two neighboring edge pairs, e.g.\ $(u_i,u'_i)$ and $(u_{i+1},u'_{i+1})$, there is another unit pair (denoted by $u,u'$) in the frontier-pair-list between $(u_i,u'_i)$ and $(u_{i+1},u'_{i+1})$ (see Figure~\ref{fig:sigmaP-functionG}~(b)), and the bottom border of $\block(u,u')$ is exactly the reflection of $[Z_i \circlearrowright Z_{i+1}]$.
2. Boundary-portions $[Z_1\circlearrowright Z_2],\ldots,[ Z_m\circlearrowright Z_1]$ form a partition of $\partial P$.
    This is because $Z_1,\ldots,Z_m$ lie in clockwise order $\partial P$, which is due to the bi-monotonicity of the $Z$-points (Fact~\ref{fact:Z_bi-monotonicity}), as shown in Figure~\ref{fig:sigmaP-functionG}~(c).

\smallskip Next, we show that $\sigma P$ can be computed in $O(n)$ time.
We compute $\sigma P$ in three steps; each costs $O(n)$ time.
(Step 1) Compute the frontier-pair-list by Algorithm~\ref{alg:FPL-def}.
(Step 2) Compute $Z_1,\ldots,Z_m$. This cost $O(n+m)=O(n)$ time by Lemma~\ref{lemma:Z-compute}~part~3 since they lie in clockwise order.
(Step 3) Construct each side in the bottom border of each frontier block. Each side can be constructed in $O(1)$ time according to the definition of bottom borders.

\smallskip To compute $\sigma P\cap \partial P$, we can try each side of $\sigma P$ and compute its intersection with $\partial P$. According to the common computational geometric result, by $O(n)$ time preprocessing, the intersection between a segment and the boundary of a fixed convex polygon $P$ can be computed in $O(\log n)$ time. Thus, this takes $O(n\log n)$ time.

\smallskip To compute the future-tag of each intersection $K_i$ in $\sigma P\cap \partial P$ reduces to
  computing $\unit(g(K_i))$, due to (\ref{def:tag2}). Notice that
    (1) $\unit(g(\cdot))$ has the property that it is identical within any side of $\sigma P$ (by the definition of $g$), and
    (2) while computing $\sigma P$, we can compute the value of $\unit(g(\cdot))$ for each side of $\sigma P$.
Therefore, by sweeping around $\sigma P$, we can compute $\unit(g(K_i))$ for all the intersections $K_i$ in $\sigma P\cap \partial P$ in linear time.
\end{proof}

\subparagraph*{Running time analysis of the sweeping algorithm.} The ADD step requires $O(n\log^2n)$ time for Group~1 (as shown in subsection~\ref{subsect:algorithms-endpoints}),
  and $O(n\log n)$ time for Group~2 (by Lemma~\ref{lemma:sigmaP_O(n)}).
Also according to Lemma~\ref{lemma:sigmaP_O(n)}, there are in total $O(n)$ event-points.
So the SORT step runs in $O(n\log n)$ time (or even in $O(n)$ time).
The REMOVE and SWEEP steps are trivial and they cost $O(n)$ time.
Therefore, the algorithm runs in $O(n\log^2n)$ time.

\begin{remark}
There is an $O(n)$ time algorithm for computing $\sigma P\cap \partial P$, improving the one given in the proof of Lemma~\ref{lemma:sigmaP_O(n)}.
We sketch it in the following. Initially, choose a pair of edges, one from $\sigma P$ and one from $\partial P$.
  In each iteration, compute their intersection and change one edge to its clockwise next one.
    The selection of edge-to-change follows a specific (and involved) rule.
We can prove that, by selecting good initial edges and rule, the algorithm does not miss any intersection in $\sigma P\cap\partial P$.
    The analysis is complicated. So we do not present it in detail.
\end{remark}


\subsection{Two delimiting edges $e_{p_V},e_{q_V}$ for block locating}\label{subsect:alg_B-pq}

Fix $V$ to be a vertex. Let $\block(u^*_1,u^*_2)$ denote the block containing $V$.
Recall section~\ref{sect:techover} for a sketch on how do we compute $(u^*_1,u^*_2)$.
It mentioned that we have to find two \textbf{delimiting edges} $e_{p_V},e_{q_V}$ so that
\begin{eqnarray}
& \text{$e_p\prec e_{q}$. Moreover, $(v_p\circlearrowright v_{q+1})$ contains $V$.} \label{eqn:pq1}\\
& \text{$u^*_1 \in [v_p\circlearrowright V)$, and $u^*_2\in (V\circlearrowright v_{q+1}]$}               \label{eqn:pq2}.
\end{eqnarray}

This subsection shows how to find $(e_{p_V},e_{q_V})$ satisfying (\ref{eqn:pq1}) and (\ref{eqn:pq2}) by utilizing the bounding-quadrants of the blocks.
Although it is rather short, this subsection contains the most important ideas for proving Theorem~\ref{thm:nestp-location}.

\smallskip Assume $V=v_i$ henceforth for convenience. Denote
\begin{equation}\label{eqn:nabla}
    \nabla_V:=\left\{(u,u')\mid \hbox{unit $u$ is chasing unit $u'$},
        u \in (\D_i\circlearrowright V),
        u'\in (V\circlearrowright \D_{i-1})\right\}.
\end{equation}

  \begin{fact}\label{fact:nabla}
$(u_1^*,u^*_2)\in \nabla_V$.
\end{fact}

\begin{proof}
By the definition of $\nabla_V$, it reduces to proving that $u_1^*\in (\D_i\circlearrowright V)$ while $u_2^*\in (V\circlearrowright \D_{i-1})$.

Let $e_a=forw(u^*_1),e_{a'}=back(u^*_2)$. Since $u^*_1$ is chasing $u^*_2$, we have $e_a\preceq e_{a'}$.
By Lemma~\ref{lemma:block-in-quad}, $V\in \block(u^*_1,u^*_2)\subset \qd_{u^*_1}^{u^*_2}=\qd_a^{a'}\subseteq \hp_a^{a'}$, hence $V=v_i\in (v_a\circlearrowright v_{a'+1})$. Together, $e_a\prec e_i$ and $e_{i-1}\prec e_{a'}$.

Since $e_a\prec e_i$, $e_a\in (\D_i\circlearrowright V)$, i.e.\ $forw(u^*_1)\in (\D_i\circlearrowright V)$. So, $u^*_1\in [\D_i\circlearrowright V)$.

Since $e_{i-1}\prec e_{a'}$, $e_{a'}\in (V\circlearrowright \D_{i-1})$, i.e.\ $back(u^*_2)\in (V\circlearrowright \D_{i-1})$. So, $u^*_2\in (V\circlearrowright \D_{i-1}]$.

In the following we further argue that $u^*_1\neq \D_i$ and $u^*_2\neq \D_{i-1}$.
Because $e_{a'}\in (V\circlearrowright \D_{i-1})$, it also lies in $(V\circlearrowright \D_i)$. So, $e_{a'}\preceq back(\D_i)$. Therefore, $back(\D_i)\nprec e_{a'}$, i.e.\ $back(\D_i)\nprec back(u^*_2)$. Therefore, $\D_i$ is not chasing $u^*_2$.
This means that $u^*_1\neq \D_i$ because $u^*_1$ must be chasing $u^*_2$.
Symmetrically, $u^*_2\neq \D_{i-1}$.
\end{proof}

\begin{figure}[b]
\centering \includegraphics[width=.7\textwidth]{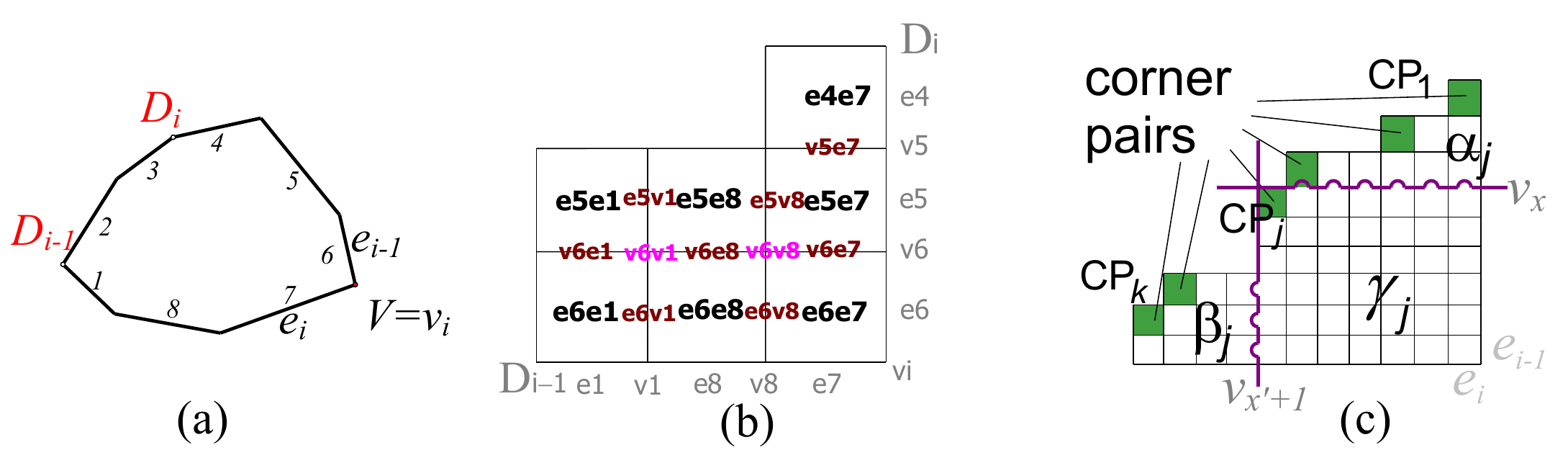}
\caption{An illustration of set $\nabla_V$ and its corner pairs.}\label{fig:nabla}
\end{figure}

Following the definition of chasing and due to the definition of $\nabla_V$, all elements $(u,u')$ in $\nabla_V$ constitute into a ``staircase'' structure
  when they are filled into a 2-dimensional matrix (as shown in Figure~\ref{fig:nabla}~(b)), where
    elements with the same $u$ are in the same row and rows are sorted from top to bottom according to the clockwise order of $u$, whereas
     elements with the same $u'$ are in the same column and columns are sorted from right to left according to the clockwise order of $u'$.
See also Figure~\ref{fig:nabla}~(c) for an illustration of the ``staircase'' structure.

We define the ``corners of this staircase'' as the corner pairs.
Formally, $(e_x,e_{x'})$ in $\nabla_V$ is a \textbf{\emph{corner pair}}, if neither $(e_{x-1},e_{x'})$ nor $(e_x,e_{x'+1})$ belongs to $\nabla_V$.
(This concept is similar to extremal pair introduced in Definition~\ref{def:extremal}.)
Denote by $\mathsf{CP}_1,\ldots,\mathsf{CP}_k$ the corner pairs, which are sorted so that
  $\mathsf{CP}_1$ is at topmost and $\mathsf{CP}_k$ at leftmost.

\medskip As demonstrated below, we are going to pick a corner pair of $\nabla_V$ to be $(e_p,e_q)$.
The exact definition of $p,q$ relies on an interesting observation of $\nabla_V$ (Fact~\ref{fact:test_cri} below). We need some notations.
\begin{description}
\item [Subsets $\alpha_j,\beta_j,\gamma_j$ of $\nabla_V$.]
Consider $\mathsf{CP}_j=(e_x,e_{x'})$.
If we cut $\nabla_V$ along the horizontal line corresponding to $v_x$ and the vertical line corresponding to $v_{x'+1}$,
  we get three chunks; the unit pairs in the top chunk are in $\alpha_j$; those in the left chunk are in $\beta_j$; and the rest form a rectangular shape and they are in $\gamma_j$. Formally,

$\alpha_j=\{(u,u')\in \nabla_V\mid u\hbox{ lies in }(\D_i\circlearrowright v_x)\}$,

$\beta_j=\{(u,u')\in \nabla_V\mid u'\hbox{ lies in }(v_{x'+1}\circlearrowright \D_{i-1})\}$,

$\gamma_j=\left\{(u,u')\in \nabla_V\mid u \hbox{ lies in $[v_x\circlearrowright V)$}, u' \hbox{ lies in $(V\circlearrowright v_{x'+1}]$}\right\}$.

See the illustration of subsets $\alpha_j,\beta_j$ and $\gamma_j$ in Figure~\ref{fig:nabla}~(c). For convenience, denote $\alpha_{k+1}=\nabla_V$.
\end{description}

Recall $\{\bp_u^{u'}\}$ in Definition~\ref{def:quad}. For any subset $S$ of $\nabla_V$, denote $\bp[S]=\bigcup_{(u,u')\in S}\bp_{u}^{u'}$.

\begin{fact}\label{fact:test_cri}
$\bp[\alpha_{j+1}]\cap \bp[\beta_j]=\varnothing~ (1\leq j\leq k).$
\end{fact}

Before proving Fact~\ref{fact:test_cri}, we state explicit formulas for $\bp[\alpha_j]$ and $\bp[\beta_j]$.
Let $a_j,b_j~(1<j\leq k)$ respectively denote
  the edge pairs at the upper right and lower left corners of $\alpha_j$;
Let $c_j,d_j~(1\leq j<k)$ respectively denote
  the edge pairs at the upper right and lower left corners of $\beta_j$; see Figure~\ref{fig:test_criterion}~(a).
Recall that $\rho.s$ and $\rho.t$ denote the starting and terminal point of boundary-portion $\rho$.
Applying the monotonicity of $\bp$ (Lemma~\ref{lemma:br_monotone}), we have
\begin{eqnarray}
\bp[\alpha_j]=(\bp[a_j].s\circlearrowright \bp[b_j].t)\text{, for any }1<j\leq k. \label{eqn:alpha}\\
\bp[\beta_j]=(\bp[c_j].s\circlearrowright \bp[d_j].t)\text{, for any }1\leq j<k  \label{eqn:beta}.
\end{eqnarray}

\begin{figure}[b]
\centering\includegraphics[width=.6\textwidth]{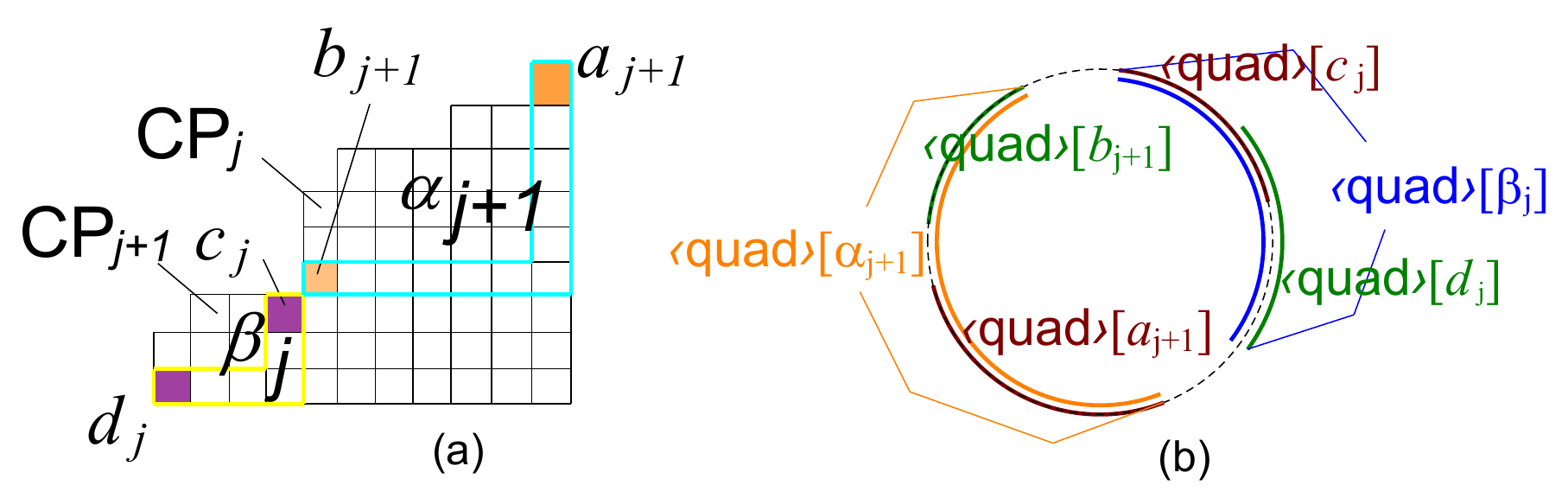}
\caption{Illustration of the proof of Fact~\ref{fact:test_cri}.}\label{fig:test_criterion}
\end{figure}

\begin{proof}[Proof of Fact~\ref{fact:test_cri}]
When $j=k$, set $\beta_j$ is empty and the equation is trivial.

Next, we assume that $j<k$. We state the following observations.
\begin{eqnarray}
&& \text{$\bp[a_{j+1}].s, \bp[b_{j+1}].s, \bp[c_j].s, \bp[d_j].s$ lie in clockwise order.}\label{eqn:s_monotone}\\
&& \text{$\bp[a_{j+1}].t, \bp[b_{j+1}].t, \bp[c_j].t, \bp[d_j].t$ lie in clockwise order.}\label{eqn:t_monotone}\\
&&   \bp[a_{j+1}]\text{ has no overlap with }\bp[d_j].\label{eqn:test_cri_proof1}\\
&&   \bp[b_{j+1}]\text{ has no overlap with }\bp[c_j].\label{eqn:test_cri_proof2}
\end{eqnarray}
Proof: The first two facts follow from the monotonicity of $\bp$; the proof of (\ref{eqn:test_cri_proof1}) is given in the following; the proof of  (\ref{eqn:test_cri_proof2}) is similar to that of (\ref{eqn:test_cri_proof1}) and omitted. Recall $V=v_i$. Notice that $a_{j+1}=(forw(\D_i),e_i)$ and $d_j=(e_{i-1},back(\D_{i-1}))$.
Observing that edges $forw(\D_i),e_i,e_{i-1},back(\D_{i-1})$ do not lie in any inferior portion,
 applying the peculiarity of the bounding-quadrants (Lemma~\ref{lemma:br_peculiar}), $\qd_{forw(\D_i)}^i\cap \qd_{i-1}^{back(\D_{i-1})}$ lie in the interior of $P$.
So, $\qd_{forw(\D_i)}^i\cap \partial P$ is disjoint with $\qd_{i-1}^{back(\D_{i-1})}\cap \partial P$. This further implies (\ref{eqn:test_cri_proof1}).

\medskip See Figure~\ref{fig:test_criterion}~(b). Combining the observations,
  $\bp[a_{j+1}].s$, $\bp[b_{j+1}].s$, $\bp[b_{j+1}].t$, $\bp[c_j].s$, $\bp[d_j].s$, $\bp[d_j].t$
lie in clockwise order around $\partial P$. In particular, points
  $\bp[a_{j+1}].s$, $\bp[b_{j+1}].t$, $\bp[c_j].s$, $\bp[d_j].t$ lie in clockwise order around $\partial P$.
So $(\bp[a_{j+1}].s\circlearrowright \bp[b_{j+1}].t)$ is disjoint with $(\bp[c_j].s \circlearrowright \bp[d_j].t)$.
Further, by (\ref{eqn:alpha}) and (\ref{eqn:beta}), $\bp[\alpha_{j+1}]$ is disjoint with $\bp[\beta_j]$.
\end{proof}

We state a trivial fact before showing the definition of $p_V$ and $q_V$.
\begin{enumerate}
\item[(i)] \emph{If $(u^*_1,u^*_2)$ belongs to some set $S$, then $V\in \bp [S]$.}
\end{enumerate}
\begin{proof}
$V\in \block(u^*_1,u^*_2)\cap \partial P \subseteq \qd_{u^*_1}^{u^*_2}\cap \partial P \subseteq \bp_{u^*_1}^{u^*_2}\subseteq \bigcup_{(u,u')\in S}\bp_{u}^{u'}=\bp[S].$
\end{proof}

\begin{definition}[Delimiting edges $e_{p_V}$ and $e_{q_V}$]\label{def:pq}
Recall that $(u^*_1,u^*_2)\in \nabla_V$ (see Fact~\ref{fact:nabla}). So $V\in \bp[\nabla_V]$ by fact~(i) above.
Further, since $\varnothing=\alpha_1\subset \ldots\subset \alpha_{k+1}=\nabla_V$, there must be a unique index in $1,\ldots,k$, denoted by $h$,
    such that $V \notin \bp[\alpha_{h}]$ but $V\in \bp[\alpha_{h+1}]$.
We choose the corner pair $\mathsf{CP}_h$ to be $(e_{p_V},e_{q_V})$.\smallskip
\end{definition}

By defining $p_V,q_V$ this way, we can verify that (\ref{eqn:pq1}) and (\ref{eqn:pq2}) are satisfied.

\begin{proof}[Proof of (\ref{eqn:pq1})] Since $(e_p,e_q)\in \nabla_V$, we get $e_p\prec e_{q}$ and $V\in(v_p\circlearrowright v_{q+1})$ by (\ref{eqn:nabla}).
\end{proof}

\begin{proof}[Proof of (\ref{eqn:pq2})]
By the definition of $h$, we know $V \notin \bp[\alpha_{h}]$ and $V\in \bp[\alpha_{h+1}]$.

Since $V \notin \bp[\alpha_{h}]$, we know $(u^*_1,u^*_2)\notin \alpha_{h}$ due to fact~(i).
Since $V \in \bp[\alpha_{h+1}]$, we get $V \notin \bp[\beta_{h}]$ according to Fact~\ref{fact:test_cri}, which further implies that
$(u^*_1,u^*_2)\notin \beta_{h}$ due to fact~(i).

However, by Fact~\ref{fact:nabla}, $(u^*_1,u^*_2)\in \nabla_V=\alpha_{h}\cup \beta_{h}\cup \gamma_{h}$.
So $(u^*_1,u^*_2)$ must belong to $\gamma_{h}$, i.e.\ $(u^*_1,u^*_2)\in \left\{(u,u')\in \nabla_V\mid u \hbox{ lies in $[v_p\circlearrowright V)$}, u' \hbox{ lies in $(V\circlearrowright v_{q+1}]$}\right\}$.
This implies (\ref{eqn:pq2}).
\end{proof}

\begin{lemma}
Given $1\leq j\leq k$, in $O(1)$ time we can determine whether $V$ lies in $\bp[\alpha_j]$.
As a corollary, we can compute $h$  and thus compute $(e_p,e_q)$ (using definition~\ref{def:pq}) in $O(\log n)$ time.
\end{lemma}

\begin{figure}[h]
\centering\includegraphics[width=.7\textwidth]{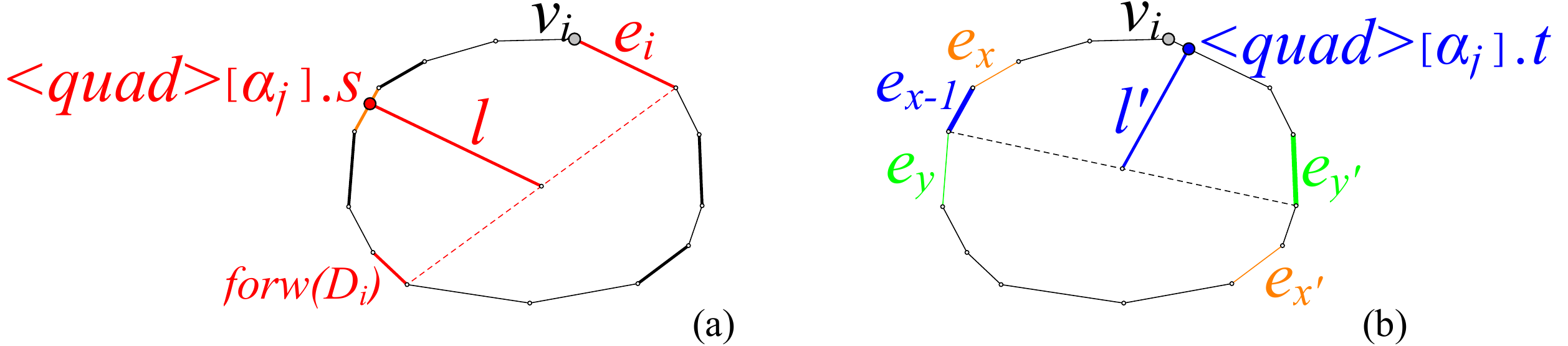}
\caption{Compute $e_p,e_q$.}\label{fig:preprocess_pq}
\end{figure}

\begin{proof}
The case $j=1$ is trivial since $\bp[\alpha_j]=\varnothing$. Assume that $j>1$.
Assume $\mathsf{CP}_j=(e_x,e_{x'}),\mathsf{CP}_{j-1}=(e_y,e_{y'})$.
(We can compute $\mathsf{CP}_j$ and $\mathsf{CP}_{j-1}$ in $O(1)$ time.
Except for the first and last element of $\mathsf{CP}$, the other corner pairs are extremal pairs.
Yet we can obtain a list of extremal pairs beforehand and use it to compute $\mathsf{CP}_j$.)

Recall that (i) $\bp[\alpha_j]= (\bp[a_j].s\circlearrowright \bp[b_j].t)$, by (\ref{eqn:alpha}), where
   $a_j,b_j$ are respectively the edge pairs at the upper right and lower left corners of $\alpha_j$.
   Observe that $a_j=(forw(\D_i),e_i)$ and $b_j=(e_{x-1},e_{y'})$.
   According to Definition~\ref{def:quad},
  (ii)  $\bp[a_j].s$ is at the intersection between $l$ and $[\D_i\circlearrowright v_{i+1}]$,
where $l$ denotes the line at $(\D_i+v_{i+1})/2$ parallel to $e_i$ (see Figure~\ref{fig:preprocess_pq}~(a)),
   and (iii) $\bp[b_j].t$ is at the intersection between $l'$ and $[v_{x-1}\circlearrowright v_{y'+1}]$,
where $l'$ denotes the line at $(v_{x-1}+v_{y'+1})/2$ parallel to $e_{x-1}$ (see Figure~\ref{fig:preprocess_pq}~(b)).
Altogether, $V\in \bp[\alpha_j]$ if and only if $V$ lies in the open half-plane delimited by $l'$ and containing $e_{x-1}$.
Further, since in $O(1)$ time we can compute $l'$ and find the side of $l'$ that contains $V$,
   we can determine whether $V\in \bp[\alpha_j]$ in $O(1)$ time, as claimed above.
\end{proof}



\subsection{Find the block that contains $V$}\label{subsect:alg_B_SV=E}

Fix a vertex $V$.
Given $w$ such that $V\in \sector(w)$, and $p_V,q_V$ such that (\ref{eqn:pq1}) and (\ref{eqn:pq2}) hold ,
  we now compute $(u^*_1,u^*_2)$ so that $\block(u^*_1,u^*_2)$ contains $V$.
  (Unit $w$ and indices $p_V,q_V$ are computed in subsections~\ref{subsect:alg_S} and \ref{subsect:alg_B-pq}.) \smallskip

Recall the sketch of our algorithm in section~\ref{sect:techover} and recall the following notations.
A unit pair $(u,u')$ in which $u$ is chasing $u'$ is \emph{alive} if
  $u\in [v_p\circlearrowright V)$ and $u'\in (V\circlearrowright v_{q+1}]$.
Moreover, it is \emph{active} if it is alive and $\zeta(u,u')\text{ intersects }w$.
Furthermore, for each active pair $(u,u')$, define $\cell(u,u'):=\block(u,u')\cap \sector(w)$ and call it a \emph{cell}.

\begin{fact}\label{fact:active}
  $(u^*_1,u^*_2)$ is active, and $\cell(u^*_1,u^*_2)$ is the unique cell that contains $V$.
\end{fact}

\begin{proof}
We know $(u^*_1,u^*_2)$ is alive due to (\ref{eqn:pq2}).

Assume $f^{-1}(V)=(X_1,X_2,X_3)$.
Because $V\in \block(u^*_1,u^*_2)$, we know $X_3\in u^*_1$ and $X_1\in u^*_2$.
Because $(X_1,X_2,X_3)\in \T$, point $X_2\in \zeta(\unit(X_3),\unit(X_1))$. Together, $X_2\in \zeta(u^*_1,u^*_2)$.
Because $V\in \sector(w)$, we know $X_2\in w$.
Therefore, $\zeta(u^*_1,u^*_2)$ intersects $w$ (at $X_2$), which means that $(u^*_1,u^*_2)$ is active.

Since $\block(u^*_1,u^*_2)$ and $\sector(w)$ both contain $V$, their intersection $\cell(u^*_1,u^*_2)$ contains $V$.

At last, we argue that $\cell(u^*_1,u^*_2)$ is the unique cell that contains $V$. If, to the opposite, $V$ lies in two distinct cells, it lies in two distinct blocks, which contradicts \textsc{Block-disjointness}.
\end{proof}

By Fact~\ref{fact:active}, it reduces to finding the cell that contains $V$.
Next, there are two cases, depending on whether $w$ is an edge or a vertex.
We first concentrate on the typical case where $w$ is an edge. Assume $w=e_k$.

\subparagraph*{Outline.} We first prove some observations of the cells (e.g., a monotonicity as stated in Fact~\ref{fact:cell-monotone}),
  define those regions called layers,
    and prove a monotonicity between the layers (Fact~\ref{fact:layer-monotone}), and then present our algorithm.

\subsubsection{Basic observations}

\begin{description}
\item[Active edges.]
An edge $e_j$ in $(v_p\circlearrowright V)$ is \emph{active} if there is at least one unit $u$ such that $(e_j,u)$ is active;
an edge $e_j$ in $(V\circlearrowright v_{q+1})$ is \emph{active} if there is at least one unit $u$ such that $(u,e_j)$ is active.
\end{description}

\begin{fact}\label{fact:active-consecutive}
\begin{enumerate}
\item For $e_j\in(v_p\circlearrowright V)$ that is active,
  the units in $\{u\mid (e_j,u)\text{ is active}\}$ are consecutive.
The clockwise first and last such units can be found in $O(\log n)$ time.
  For $e_j\in(V\circlearrowright v_{q+1})$ that is active,
  the units in $\{u\mid (u,e_j)\text{ is active}\}$ are consecutive.
The clockwise first and last such units can be found in $O(\log n)$ time.
\item The active edges in $(v_p\circlearrowright V)$ are consecutive.
  The clockwise first and last such edges can be computed in $O(\log n)$ time.
  Similarly, the active edges in $(V\circlearrowright v_{q+1})$ are consecutive.
  The clockwise first and last such edges can be computed in $O(\log n)$ time.
\end{enumerate}
\end{fact}

\begin{proof} For any edge $e_j$ in $(v_p\circlearrowright V)$, denote
    $b(j)=\begin{cases}q+1 & \text{if $e_j\prec e_{q+1}$;}\\
                            q& \text{otherwise.}\end{cases}$
Denote $b=b(j)$ when $j$ is clear.

Assume $V=v_i$. Denote $\Pi_j=\big(\zeta(e_j,e_i),\zeta(e_j,v_{i+1}),\ldots,\zeta(e_j,v_{b}),\zeta(e_j,e_b)\big)$. We state two observations:
\begin{itemize}
\item[(i)] $\Pi_j=\left(Z_j^i,[Z_j^i\circlearrowright Z_j^{i+1}],\ldots,[Z_j^{b-1}\circlearrowright Z_j^{b}],Z_j^{b}\right)$. (By definition of $\zeta(e_j,u)$)
\item[(ii)] $Z_j^i,\ldots,Z_j^b$ lie in clockwise order on $\rho=[v_{b+1}\circlearrowright v_j]$. (By bi-monotonicity of the $Z$-points)
\end{itemize}

\noindent Proof of part~1. Assume $e_j\in (v_p\circlearrowright V)$; the case where $e_j\in (V\circlearrowright v_{q+1})$ is symmetric.
By observations~(i) and (ii), the elements in $\Pi_j$ that intersect $e_k$ are consecutive. This simply implies that $U=\{u\mid (e_j,u)\text{ is active}\}$ consists of consecutive units.
Computing the first unit in $U$ reduces to computing  $h$ such that $Z_j^{h-1}\leq_\rho v_k<_\rho Z_j^h$,
which can be computed in $O(\log n)$ time using Lemma~\ref{lemma:Z-compute} and binary search.
The last unit in $U$ can be computed similarly.\smallskip

\noindent Proof of part~2. Let $\pi_j$ be the union of all portions in $\Pi_j$, which equals $[Z_j^i \circlearrowright Z_j^{b(j)}]$ due to observations~(i) and (ii).
Applying the bi-monotonicity of the $Z$-points (fact~\ref{fact:Z_bi-monotonicity}), the starting points of $\pi_p,\ldots,\pi_{i-1}$ lie in clockwise order around $\partial P$, and so do their terminal points.
So, the ones in $\pi_p,\ldots,\pi_{i-1}$ that intersect $e_k$ are consecutive.
This means the active edges in $(v_p\circlearrowright V)$ are consecutive, since $e_j$ is active if and only if $\pi_j$ intersects $e_k$.
Computing the first and last active edges in $(v_p\circlearrowright V)$ reduces to computing the first and last elements in $\pi_p,\ldots,\pi_{i-1}$ that intersect $e_k$.
By Lemma~\ref{lemma:Z-compute}, in $O(1)$ time we can determine whether $\pi_j$ is contained in $[v_{b(j)+1}\circlearrowright v_k]$ or in $[v_{k+1}\circlearrowright v_j]$, or intersects $e_k$. So, by a binary search, in $O(\log n)$ time we can compute these two edges.
\end{proof}

\begin{fact}\label{fact:cell-parallelogram}
Given an active pair $(e_j,u)$ (or $(u,e_j)$), region $\cell(e_j,u)$ (or $\cell(u,e_j)$) is a parallelogram with two sides congruent to $e_j$, and it can be computed in $O(1)$ time.
\end{fact}

\begin{proof}
Assume $(e_j,u)$ is active. This means $\zeta(e_j,u)$ intersects with $e_k$, and
\begin{equation}\label{def:cell-full}
        \cell(e_j,u)=f(\{(X_1,X_2,X_3)\mid X_1=u,X_2\in \zeta(e_j,u)\cap e_k,X_3\in e_j\}).
        \end{equation}

\noindent \emph{Case~1}: $u$ is an edge, e.g.\ $u=e_{j'}$.
Since $\zeta(e_j,u)=Z_j^{j'}$ and it intersects $e_k$, point $Z_j^{j'}$ lies in $e_k$ and hence can be computed in $O(1)$ time according to Lemma~\ref{lemma:Z-compute}.
Then, $\cell(e_j,e_{j'})$ is the 2-scaling of $e_j\oplus e_{j'}$ with respect to $Z_j^{j'}$,
  which is a parallelogram with two sides congruent to $e_j$ and which can be computed in $O(1)$ time.\medskip

\noindent \emph{Case~2}: $u$ is a vertex, e.g.\ $u=v_{j'}$. First, we argue that $\zeta(e_j,v_{j'})$ is not a single point.
    Suppose to the opposite that $\zeta(e_j,v_{j'})$ is a single point.
    Then, its two endpoints $Z_j^{j'-1},Z_j^{j'}$ must be identical, and must lie in $e_k$ since $\zeta(e_j,v_{j'})$ intersects $e_k$.
    However, by Fact~\ref{fact:dist-unique-location}, when $Z_j^{j'-1},Z_j^{j'}$ lie in $e_k$, they lie at $(\I_{j,k}+\I_{j'-1,k})/2$, $(\I_{j,k}+\I_{j',k})/2$,
        respectively, which do not coincide because $\I_{j'-1,k}\neq \I_{j',k}$.
    Contradictory. Following this argument, $\zeta(e_j,u)\cap e_k$ is a segment that is not a single point.
        Combining this fact with (\ref{def:cell-full}), $\cell(e_j,v_{j'})$ is a parallelogram with two sides congruent to $e_j$. (To see this more clearly, we refer to Figure~\ref{fig:block-borders}~(c).)
    Moreover, applying Lemma~\ref{lemma:Z-compute}, segment $\zeta(e_j,v_{j'})\cap e_k$ can be computed in $O(1)$ time,  and thus $\cell(e_j,v_{j'})$ can be computed in $O(1)$ time.

\smallskip The proof of the claim on $\cell(u,e_j)$ is symmetric and omitted.
\end{proof}

\subsubsection{Monotonicity of cells, definition of layers, and monotonicity of layers}

\begin{fact}\label{fact:cell-monotone}
See Figure~\ref{fig:cells-layers}~(a) and (b) for illustrations of the following statements (and see the proof below).
\begin{enumerate}
\item For $e_j$ in $(v_p\circlearrowright V)$ that is active,
  $\cell(e_j,u_s)$, $\ldots, \cell(e_j,u_t)$ are contiguous and lie monotonously in the opposite direction of $e_k$,
    where $u_s,\ldots,u_t$ list the units in $\{u\mid (e_j,u)\text{ is active}\}$ in clockwise order.
\item  For $e_j$ in $(V\circlearrowright v_{q+1})$ that is active,
   $\cell(u_s,e_j)$, $\ldots, \cell(u_t,e_j)$ are contiguous and lie monotonously in the opposite direction of $e_k$,
    where $u_s,\ldots,u_t$ list the units in $\{u\mid (u,e_j)\text{ is active}\}$ in clockwise order.
\end{enumerate}
\end{fact}

\begin{definition}[Layers]See Figure~\ref{fig:cells-layers} and Figure~\ref{fig:cell_monotone}~(b). We define two types of layers.
\begin{itemize}
\item[(A)] Let $e_j$ be an active edge in $(v_p\circlearrowright V)$.
Assume $\{u\mid (e_j,u)\text{ is active}\}=\{u_s,\ldots,u_t\}$ (in clockwise order).
Let $\mathsf{body}_j$ denote the region united by regions $\cell(e_j,u_s)$, $\ldots, \cell(e_j,u_t)$.

Clearly, $\mathsf{body}_j$ is a region with two borders congruent to $e_j$ since the cells have borders congruent to $e_j$ (according to Fact~\ref{fact:cell-parallelogram}).
  By removing these two borders, we can get an extension of $\mathsf{body}_j$ which contains two strip regions parallel to $e_k$.
  This extension is defined as $\mathsf{layer}_j$ and is called an \textbf{\emph{A-type layer}}.

\item[(B)] Let $e_j$ be an active edge in $(V\circlearrowright v_{q+1})$.
Assume $\{u\mid (u,e_j)\text{ is active}\}=\{u_s,\ldots,u_t\}$ (in clockwise order).
Let  $\mathsf{body}_j$ denote the region united by regions $\cell(u_s,e_j)$, $\ldots, \cell(u_t,e_j)$.

Clearly, $\mathsf{body}_j$ is a region with two borders congruent to $e_j$ since the cells have borders congruent to $e_j$ (according  to Fact~\ref{fact:cell-parallelogram}).
  By removing these two borders, we can get an extension of $\mathsf{body}_j$ which contains two strip regions parallel to $e_k$.
  This extension is defined as $\mathsf{layer}_j$ and is called a \textbf{\emph{B-type layer}}.
\end{itemize}
\end{definition}

\begin{figure}[h]
\begin{minipage}[b]{0.5\textwidth}
 \centering\includegraphics[width=.97\textwidth]{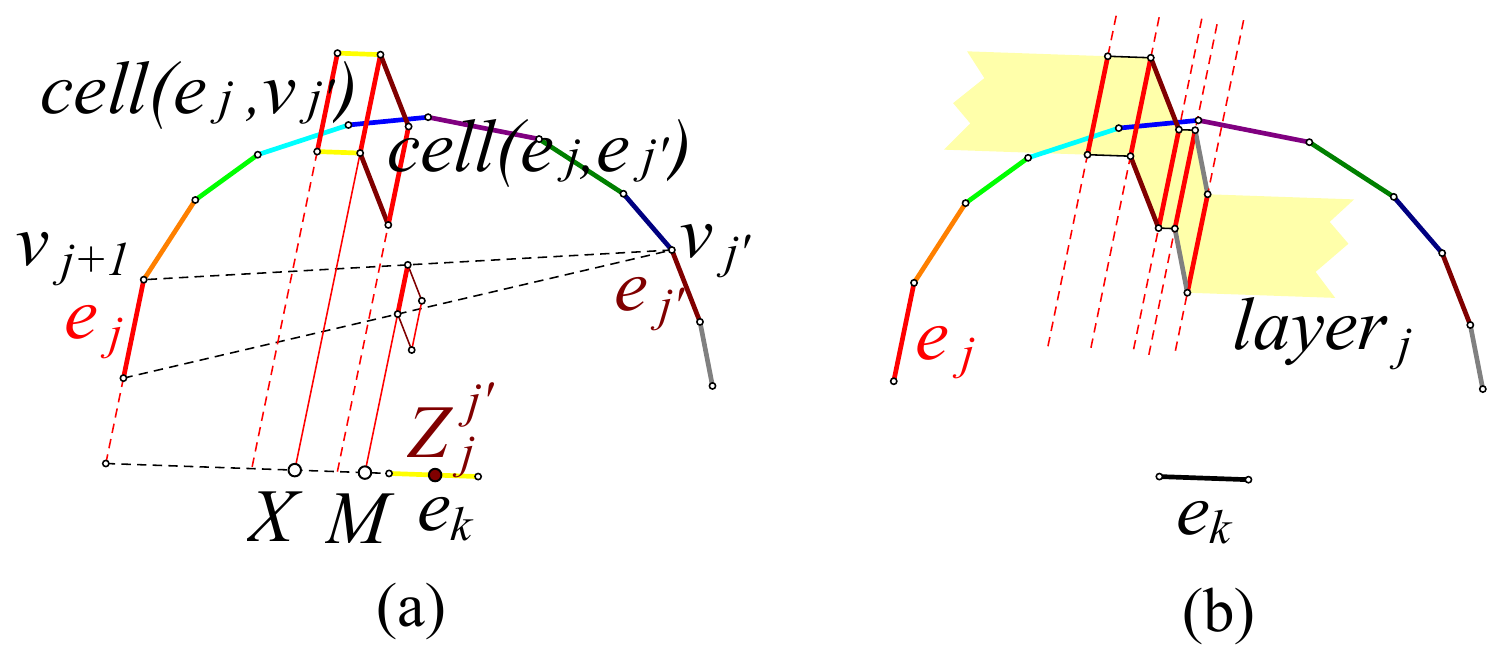}
 \caption{Monotonicity of cells \& definition of layers.}\label{fig:cell_monotone}
\end{minipage}
\begin{minipage}[b]{0.5\textwidth}
 \centering\includegraphics[width=.93\textwidth]{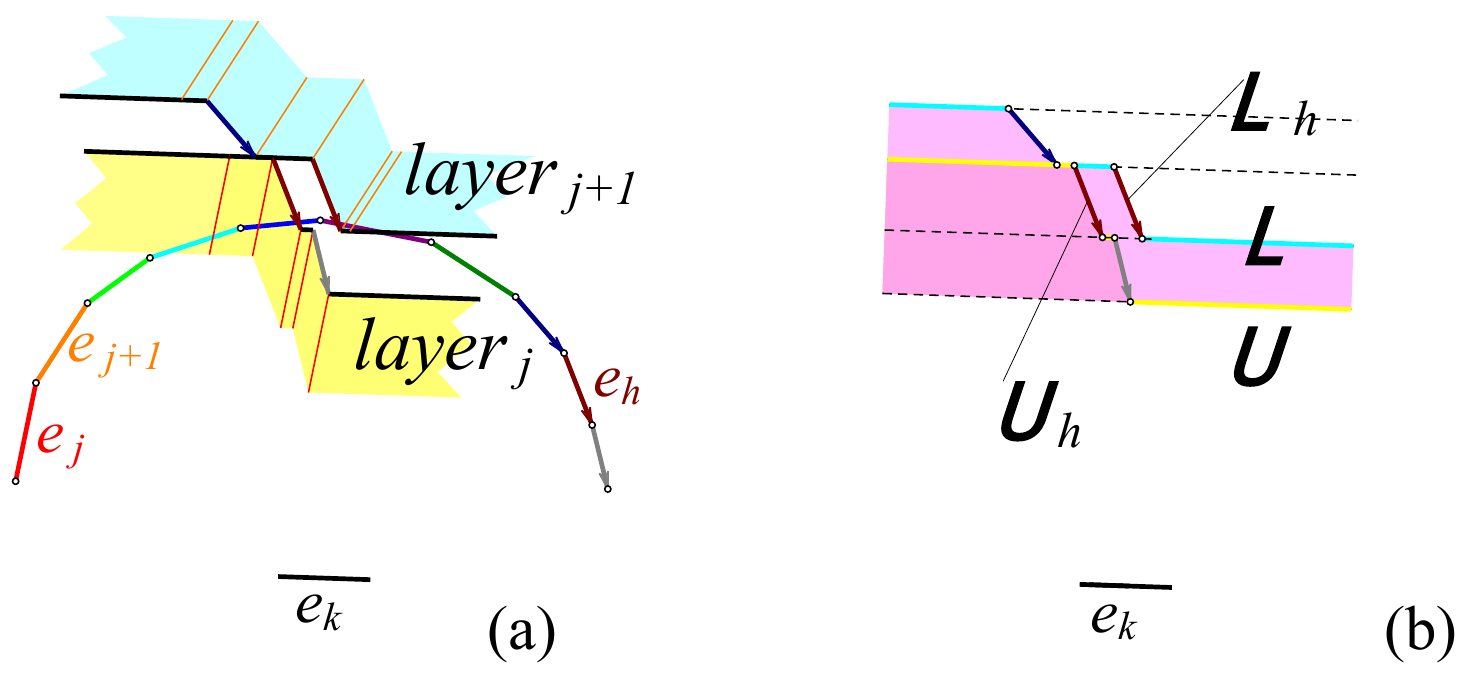}
 \caption{Monotonicity of layers.}\label{fig:layer_monotone}
\end{minipage}
\end{figure}

\begin{proof}[Proof of Fact~\ref{fact:cell-monotone}]
We show how to prove part~1; the proof of part~2 is symmetric and omitted.

See Figure~\ref{fig:cell_monotone}~(a). Let us consider the projections of these cells along direction $e_j$ onto $\el_k$,
    it reduces to proving that these projections are pairwise-disjoint and are arranged in order.
Take two incident units in $\{u_s,\ldots,u_t\}$, e.g.\ $v_{j'}$ and $e_{j'}$.
  (For incident units $e_{j'},v_{j'+1}$, the proof is similar.)
    Let $M$ be the projection of $(v_{j+1}+v_{j'})/2$; and $X$ the reflection of $Z_j^{j'}$ with respect to $M$.
    We state that the projection of $\cell(e_j,e_{j'})$ terminates at $X$ while the projection of $\cell(e_j,v_{j'})$ starts at $X$.
      This follows the definition of cells. More details are trivial and omitted.
\end{proof}

\begin{fact}\label{fact:layer-monotone}
(1) All the layers lie in the closed half-plane delimited by $\el_k$ and containing $P$.
More importantly, (2) all the A-type layers are pairwise-disjoint and lie monotonously in the direction perpendicular to $e_k$.
    Symmetrically, all the B-type layers have the same monotonicity.
\end{fact}

\begin{proof}
(1) Denote by $H$ the half-plane bounded by $\el_k$ and containing $P$.
Proving that all layers lie in $H$ reduces to proving that all cells lie in $H$,
  which further reduces to proving that $\sector(e_k)\subset H$.
Now, let $X$ be an arbitrary point in $\sector(e_k)$, we shall prove $X\in H$.
Notice that there is $(X_1,X_2,X_3)\in \T$ such that $X_2\in e_k$ and $f(X_1,X_2,X_3)=X$.
Because $X_1,X_3\in \partial P$, point $(X_1+X_3)/2$ lies in $H$.
Since $X_2\in e_k$, point $X_2$ lies on the boundary of $H$.
Together, the 2-scaling of $(X_1+X_3)/2$ with respect to $X_2$, which equals $X$, lies in $H$.

\medskip \noindent (2) Each layer has two boundaries; we refer to them as the \emph{lower border} and the \emph{upper border}, so that the lower one is closer to $\el_k$ than the upper one.
Assume that $\mathsf{layer}_j$ and $\mathsf{layer}_{j+1}$ are consecutive A-type layers. See Figure~\ref{fig:layer_monotone}~(a).
We shall prove that the upper border of $\mathsf{layer}_j$ (denoted by $\mathcal{U}$) lies between $\el_k$ and the lower border  of $\mathsf{layer}_{j+1}$ (denoted by $\mathcal{L}$). Make an auxiliary line parallel to $\el_k$ at each vertex of the two borders;
  these auxiliary lines cut the plane into \emph{slices}, as shown in Figure~\ref{fig:layer_monotone}~(b).
It reduces to proving that (i) in each slice, the region under $\mathcal{U}$ is contained in the region under $\mathcal{L}$.
Now, consider any slice that intersects both $\mathcal{U}$ and $\mathcal{L}$ (e.g.\ the middle one in the figure). (The proof for other slices are similar and easier.) Then, there is an edge $e_h$, such that the part of $\mathcal{U}$ that lies in this slice (labeled by $\mathcal{U}_h$ in the figure) and
    the part of $\mathcal{L}$ that lies in this slice (labeled by $\mathcal{L}_h$) are both translations of $e_h$.
Applying the monotonicity of cells within $\mathsf{layer}_h$ (Fact~\ref{fact:cell-monotone}), we have a monotonicity between these two translations of $e_h$
  that implies statement~(i). We can prove the monotonicity of the B-type layers symmetrically.
\end{proof}

\subsubsection{Algorithm for computing $u^*_1,u^*_2$}

\begin{lemma}\label{lemma:algB-binarysearch}
\begin{enumerate}
    \item Given an active edge $e_j$, we can do the following tasks in $O(\log n)$ time:\\
    (a) Determine whether $V$ lies in $\mathsf{layer}_j$; if not, determine which side of $\mathsf{layer}_j$ it lies on.\\
    (b) Determine whether $V$ lies in $\mathsf{body}_j$; if so, find the unique cell in $\mathsf{body}_j$ containing $V$.
    \item We can compute $u^*_1,u^*_2$ in $O(\log^2n)$ time.
\end{enumerate}
\end{lemma}

\begin{proof}
\newcommand{\chop}{\mathsf{chop}}

1. Assume $e_j\in (v_p\circlearrowright V)$; otherwise $e_j\in (V\circlearrowright v_{q+1})$ and it is symmetric.
By Fact~\ref{fact:cell-parallelogram}, the cells in $\{\cell(e_j,u)\mid (e_j,u)\text{ is active}\}$ are parallelograms with two sides parallel to $e_j$.
Those sides parallel to $e_j$ can be extended so that they divide the plane into several regions as shown in Figure~\ref{fig:cell_monotone}~(b).
We refer to each such region as a \emph{chop} and denote the one containing $\cell(e_j,u)$ by $\chop_{u}$.
Notice that
\begin{enumerate}
\item[(1)] We can compute the first and last unit in $\{u\mid (e_j,u)\text{ is active}\}$ in $O(\log n)$ time (Fact~\ref{fact:active-consecutive}).
\item[(2)] We can compute $\chop_{u}$ in $O(1)$ time, since $\cell(e_j,u)$ can be computed in $O(1)$ time (Fact~\ref{fact:cell-parallelogram}).
\item[(3)] The chops have the same monotonicity as their corresponding cells. (see Fact~\ref{fact:cell-monotone})
\end{enumerate}
Altogether, by a binary search, we can find the chop that contains $V$, which costs $O(\log n)$ time.
After the chop containing $V$ has been computed, we can easily solve tasks (a) and (b) in $O(1)$ time.

\medskip 2. We design two \emph{subroutines}.
  One assumes that $V$ is contained in an $A$-layer (i.e.\ it assumes that $u^*_1$ is an edge),
  the other assumes that $V$ is contained in a $B$-layer (i.e.\ it assumes that $u^*_2$ is an edge).
We describe the first one in the following; the other is symmetrical.
First, compute the first and last active edges $e_g,e_{g'}$ in $(v_p\circlearrowright V)$, which costs $O(\log n)$ time by Fact~\ref{fact:active-consecutive}.2.
Then, using part~1~(a) and a binary search,
  find the only $A$-layer in $\mathsf{layer}_g,\ldots,\mathsf{layer}_{g'}$ that contains $V$.
  If failed, terminate this subroutine.
Otherwise, assume that $\mathsf{layer}_j$ contains $V$, check whether $\mathsf{body}_j$ contains $V$ using part~1~(b).
If so, we find the cell containing $V$ and thus obtain $(u^*_1,u^*_2)$. It costs $O(\log^2n)$ time.

\smallskip \textsc{Correctness}: If $u^*_1$ is an edge, the first subroutine obtains $(u^*_1,u^*_2)$; if $u^*_2$ is an edge, the second subroutine obtains $(u^*_1,u^*_2)$;
however, in a degenerate case, $u^*_1,u^*_2$ can both be vertices, and the two subroutines both fail to find $(u^*_1,u^*_2)$.
(This case is indeed degenerate because it implies a parallelogram inscribed in $P$ with three corners lying on the vertices.)
Nevertheless, our algorithms can handle the degenerate case using the following modification.

\smallskip \textsc{Modification}: When $u^*_1,u^*_2$ are both vertices, $V$ lies on the boundary of some $\cell(u,u')$ such that at least one of $u,u'$ is an edge (see fact~(i) below).
\emph{We first find a cell that contains $V$ or its boundary contains $V$.
If we only find a cell whose boundary contains $V$, we use $O(1)$ extra time to find the cell that contains $V$ which is nearby.}
\begin{itemize}
\item[(i)]\emph{If $(v_j,v_{j'})$ is active and $X\in \cell(v_j,v_{j'})$, at least one of following holds.
   1. $(v_j,e_{j'-1})$ is active and $X$ lies in the boundary of $\cell(v_j,e_{j'-1})$.
   2. $(e_j,v_{j'})$ is active and $X$ lies in the boundary of $\cell(e_j,v_{j'})$.}
\end{itemize}
Proof of (i): Denote $M=(v_j+v_{j'})/2$ and denote by $X'$ the reflection of $X$ with respect to $M$.
Because $\cell(v_j,v_{j'})$ is the reflection of $\zeta(v_j,v_{j'})\cap e_k$ with respect to $M$, point $X'$ lies in $\zeta(v_j,v_{j'})\cap e_k$.
Notice that $\zeta(v_j,v_{j'})$ is the concatenation of $\zeta(v_j,e_{j'-1})$ and $\zeta(e_j,v_{j'})$.
Point $X'$ lies in $\zeta(v_j,e_{{j'}-1})\cap e_k$ or $\zeta(e_j,v_{j'})\cap e_k$.
In the former case, $(v_j,e_{j'-1})$ is active and the reflection of $X'$ with respect to $M$ (which equals $X$) lies on the boundary of $\cell(v_j,e_{j'-1})$;
in the latter case, $(e_j,v_{j'})$ is active and $X$ lies on the boundary of $\cell(e_j,v_{j'})$.
\end{proof}

\begin{remark}[Similarities between computing $\sector(V) \cap \partial P$ and answering Vertex-in-Block query]
We locate among A-type and B-type roads or layers. Roads (of the same type) admit good monotonicities, so as the layers.
\end{remark}

\subsubsection*{Compute the block containing $V$ when $V$ lies in $\sector(v_k)$}

In the above we assume $w$ is an edge. We now discuss the easier case where $w$ is vertex. Assume $w=v_k$.

In fact, by regarding $v_k$ as a sufficiently small edge, this case can be regarded as a special case of the edge case (where $w$ is an edge). But we can design a more simple solution to the vertex case, as shown below.

Let $(X_1,X_2,X_3)$ denote the preimage of $V$ under function $f$.
Clearly, $u^*_1,v_k,u^*_2$ are respectively the units containing $X_3,X_2,X_1$.
Moreover, due to (\ref{eqn:pq2}),
$[v_p\circlearrowright V)$  contains $u^*_1$; and $(V\circlearrowright v_{q+1}]$ contains $u^*_2$.
Therefore,
    \[X_1 \in (V\circlearrowright v_{q+1}],X_2= v_k,X_3\in [v_p\circlearrowright V), \text{and $VX_1X_2X_3$ is a parallelogram}.\]

\begin{lemma}\label{lemma:vertexcase}
There is a unique parallelogram $A_0A_1A_2A_3$ whose corners $A_0,A_1,A_2,A_3$ respectively lie in $V,(V \circlearrowright v_{q+1}],v_k,[v_p\circlearrowright V)$, and we can compute it in $O(\log^2n)$ time.
\end{lemma}

Applying Lemma~\ref{lemma:vertexcase}, we can find $(X_1,X_2,X_3)$ in $O(\log^2n)$ time and thus obtain $(u^*_1,u^*_2)=(\unit(X_3),\unit(X_1))$.

\begin{proof}
Suppose there are two such parallelograms, $VA_1v_kA_3$ and $VB_1v_kB_3$. Because their centers coincide at $(v_k+V)/2$,
  quadrant $ABA'B'$ is a parallelogram.
    Recall that $[v_p \circlearrowright v_{q+1}]$ is an inferior portion by (\ref{eqn:pq1}) and notice that it contains $A_1,A_3,B_1,B_3$.
  Thus we get a parallelogram inscribed in an inferior portion, which is impossible.

Computing $A_0A_1A_2A_3$ is equivalent to finding two points $A_3,A_1$ respectively restricted to $[v_p \circlearrowright V),(V \circlearrowright v_{q+1}]$ so that their mid point lies at $(v_k+V)/2$, which is further equivalent to computing the intersection between $[v_p \circlearrowright V)$ and the reflection of $(V \circlearrowright v_{q+1}]$ with respect to $(v_k+V)/2$. Using standard methods in computational geometry, we can find this intersection in $O(\log^2n)$ time by a binary search. We omit further details.
\end{proof}

\subparagraph{Acknowledgement.} This work is mainly done during my Ph. D. at Tsinghua University. This work was supported by National Basic Research Program of China Grant 2007CB807900, 2007CB807901, and the National Natural Science Foundation of China Grant 61033001, 61061130540, 61073174.

\medskip    The author thanks god for his amazing creation.
After studying this mysterious structure, we got an impression that everything is perfect and rotating.
The author thanks Haitao Wang for taking part in fruitful discussions and for many other helps,
and thanks Andrew C. Yao, Jian Li, Danny Chen, Wolfgang Mulzer, Matias Korman, Donald Sheehy, Kevin Matulef, Junhui Deng, and anonymous reviewers from past conferences for many precious suggestions.
Last but not least, the author appreciates the developers of Geometer's Sketchpad${}^\circledR$.

\bibliographystyle{latex8}
\bibliography{MAP}

\begin{thebibliography}{1}\setlength{\itemsep}{-1ex}\small

\bibitem{BergCG}
M.~d. Berg, O.~Cheong, M.~Kreveld, and M.~Overmars.
\newblock {\em Computational Geometry: Algorithms and Applications}.
\newblock Springer-Verlag TELOS, SC, CA, USA, 3rd edition, 2008.

\bibitem{Othershape-rect-EuroCG14}
S.~Cabello, O.~Cheong, and L.~Schlipf.
\newblock Finding largest rectangles in convex polygons.
\newblock {\em European Workshop on Computational Geometry}, 2014.

\bibitem{arxiv:n2}
K.~Jin.
\newblock Finding all maximal area parallelograms in a convex polygon.
\newblock {\em CoRR}, abs/1711.00181, 2017.

\bibitem{TPStechnique}
D.~Kirkpatrick and J.~Snoeyink.
\newblock Tentative prune-and-search for computing fixed-points with
  applications to geometric computation.
\newblock {\em Fundamenta Informaticae}, 22(4):353--370, Dec. 1995.

\bibitem{Othershape-square-Allerton87}
N.~Pano, Y.~Ke, and J.~O'Rourke.
\newblock Finding largest inscribed equilateral triangles and squares.
\newblock In {\em Proceeding of Annual Allerton Conference on Communication,
  Control, and Computing}, pages 869--878, 1987.

\bibitem{Placement-ST-CGTA94}
M.~Sharir and S.~Toledo.
\newblock Extremal polygon containment problems.
\newblock {\em Computational Geometry: Theory and Applications}, 4(2):99--118,
  jun 1994.

\bibitem{LecturesPolytopes}
G.~Ziegler.
\newblock {\em Lectures on Polytopes}.
\newblock Graduate texts in mathematics. Springer-Verlag, 1995.

\end{thebibliography}

\appendix

\section{Appendix}
\tikzstyle{KeyNotation} =  [thick, rectangle, rounded corners, minimum width=1.8cm, minimum height=0.6cm, text centered, draw=black, fill=green!30]
\tikzstyle{KeyNotation2} = [thick, rectangle, rounded corners, minimum width=1.8cm, minimum height=0.6cm, text centered, draw=black, fill=blue!30]
\tikzstyle{KeyNotation3} = [thick, rectangle, rounded corners, minimum width=1.8cm, minimum height=0.6cm, text centered, draw=black, fill=red!30]
\tikzstyle{arrow} = [black, thick,->,>=stealth]
\begin{figure}[h]
\centering
\begin{tikzpicture}[node distance=1.55cm]
\small
\node (Zpoints)[KeyNotation]{$Z$-points};
\node (FTBS)   [KeyNotation, right of=Zpoints,  xshift=3.3cm]{$f,\T,\block,\sector$};
\node (nestp)  [KeyNotation3, below of=FTBS,     yshift=0.2cm]{$\Nest(P)$ (including $\sigma P$)};
\node (BQ)     [KeyNotation, right of=FTBS,     xshift=4.2cm]{Bounding-quadrants of blocks};
\node (keyI)   [KeyNotation2, below of=BQ,     yshift=0.2cm]{The answer to the queries};
\draw [arrow](Zpoints) -- (FTBS);
\draw [arrow](FTBS) -- (nestp);
\draw [arrow](nestp) -- node [anchor=south] {} (keyI);
\draw [arrow](FTBS) -- node [anchor=south] {} (keyI);
\draw [arrow](BQ) -- node [anchor=north] {} (nestp);
\draw [arrow](BQ) -- node {} (keyI);
\end{tikzpicture}
\setcaptionwidth{.95\textwidth}
\caption{Key geometric objects studied in this manuscript and their relations.}\label{fig:flow}
\end{figure}

\begin{lemma}
The total number of segments in $\Nest(P)$ is $O(n^2)$.
\end{lemma}

\begin{proof}
By the explicit definition of $\LVS,\RVS$ given in subsection~\ref{subsect:algorithms-endpoints}, either of them consists of $O(n)$ segments.
So it reduces to showing that (i) \emph{the boundary-portions in
   $\{\zeta(u,u)' \mid \text{unit }u\text{ is chasing unit }u'\}$ have in total $O(n^2)$ line segments}.
For any $e_i$, applying the bi-monotonicity of the $Z$-points (Fact~\ref{fact:Z_bi-monotonicity}),
  the boundary-portions in $\{\zeta(e_i,u') \mid \text{$e_i$ is chasing unit $u'$}\}$ have in total $O(n)$ line segments.
Similarly, for any $e_j$,
  the boundary-portions in $\{\zeta(u,e_j) \mid \text{unit $u$ is chasing $e_j$}\}$ have in total $O(n)$ line segments.
These together imply that the number of line segments of $\{\zeta(v_i,v_j)\}$ is also $O(n^2)$ (see Figure~\ref{fig:blocks_def}).
Altogether, we obtain statement~(i). So the lemma holds.
\end{proof}

\end{document}